\documentclass[12pt]{amsart}
\usepackage{amsfonts}
\usepackage{changepage}
\usepackage{natbib}
\usepackage{mathrsfs}
\usepackage{xr}
\usepackage{eurosym}
\usepackage{amssymb, amsmath, url,color,  amsthm, amssymb,graphicx,multirow, appendix}
\usepackage[margin=0.8in]{geometry}
\usepackage{bigints}
\usepackage{boondox-cal}
\usepackage[scr=boondoxo,scrscaled=1.05]{mathalfa}
\usepackage{lscape,amssymb, amsmath,verbatim,graphicx}
\usepackage{epsfig,subfigure,float,subfigure}
\usepackage{graphicx}
\usepackage{epsfig}
\usepackage{threeparttable}
\usepackage{ragged2e}
\usepackage{epstopdf}
\usepackage{booktabs}
\usepackage{float}
\usepackage{multirow}
\usepackage{adjustbox}
\usepackage{lscape,lipsum}
\usepackage{url}
\usepackage{natbib}
\usepackage[T1]{fontenc}
\usepackage{a4wide, url}
\usepackage[english]{babel}
\usepackage{graphicx}
\usepackage{array}
\usepackage{multirow}
\usepackage{amsthm}
\usepackage{amsmath}
\usepackage[foot]{amsaddr}
\usepackage{amsfonts}
\usepackage{amssymb}
\usepackage{pgf}
\usepackage{paralist}
\usepackage{pdflscape}
\usepackage{lscape,amssymb, amsmath,verbatim,graphicx}
\usepackage{epsfig,subfigure,float,subfigure}
\usepackage{graphicx}
\usepackage{epsfig}
\usepackage{amssymb}
\usepackage{amsmath}
\usepackage{amsfonts}
\usepackage{geometry}
\usepackage{scalefnt}
\usepackage{graphicx}
\usepackage{caption}
\usepackage{setspace}
\usepackage{array}
\usepackage{multirow}
\usepackage{amsthm}
\usepackage{amsmath}
\usepackage{amsfonts}
\usepackage{amssymb}
\usepackage[shortlabels]{enumitem}
\usepackage{pgf}
\usepackage{paralist}
\usepackage{pdflscape}
\usepackage{rotating}
\usepackage{scalefnt}
\usepackage{xr}
\usepackage{hyperref}
\usepackage{algpseudocode}
\usepackage{algorithm}

\allowdisplaybreaks
\DeclareFontFamily{OT1}{pzc}{}
\DeclareFontShape{OT1}{pzc}{m}{it}{<-> s * [1.10] pzcmi7t}{}
\DeclareMathAlphabet{\mathpzc}{OT1}{pzc}{m}{it}

\newtheoremstyle{exampstyle}
{4pt} 
{4pt} 
{\itshape} 
{} 
{\bfseries} 
{.} 
{.75em} 
{} 
\theoremstyle{exampstyle}

\numberwithin{table}{section}
\numberwithin{figure}{section}
\newtheorem{theorem}{Theorem}[section]
\newtheorem{lemma}{Lemma}[section]

\newtheorem{assumption}{Assumption}[section]

\def\beq{\begin{equation}}
\def\eeq{\end{equation}}
\def\bals{\begin{align*}}
\def\eals{\end{align*}}
\def\bal{\begin{align}}
\def\eal{\end{align}}

\numberwithin{equation}{section}
\numberwithin{theorem}{section}
\numberwithin{corollary}{section}
\allowdisplaybreaks
\begin{document}
\title[ Monitoring volatility ]{Sequential monitoring for explosive
volatility regimes}
\author{Lajos Horv\'ath}
\address{Lajos Horv\'ath, Department of Mathematics, University of Utah, US}
\author{Lorenzo Trapani}
\address{Lorenzo Trapani, University of Leicester, UK, and Universita' di
Pavia, Italy}
\author{Shixuan Wang}
\address{Shixuan Wang, Department of Economics, University of Reading, UK}

\begin{abstract}
In this paper, we develop two families of sequential monitoring procedure to
(timely) detect changes in a GARCH(1,1) model. Whilst our methodologies can
be applied for the general analysis of changepoints in GARCH(1,1) sequences,
they are in particular designed to detect changes from stationarity to
explosivity or vice versa, thus allowing to check for \textquotedblleft
volatility bubbles\textquotedblright . Our statistics can be applied
irrespective of whether the historical sample is stationary or not, and
indeed without prior knowledge of the regime of the observations before and
after the break. In particular, we construct our detectors as the CUSUM
process of the quasi-Fisher scores of the log likelihood function. In order
to ensure timely detection, we then construct our boundary function
(exceeding which would indicate a break) by including a weighting sequence
which is designed to shorten the detection delay in the presence of a
changepoint. We consider two types of weights: a lighter set of weights,
which ensures timely detection in the presence of changes occurring
\textquotedblleft early, but not too early\textquotedblright\ after the end
of the historical sample; and a heavier set of weights, called
\textquotedblleft R\'{e}nyi weights\textquotedblright \thinspace\ which is
designed to ensure timely detection in the presence of changepoints
occurring very early in the monitoring horizon. In both cases, we derive the
limiting distribution of the detection delays, indicating the expected delay
for each set of weights. Our theoretical results are validated via a
comprehensive set of simulations, and an empirical application to daily
returns of individual stocks.
\end{abstract}

\maketitle

\doublespacing

\section{Introduction\label{intro}}

In recent years, developing tools for the (ex-ante or ex-post) detection of
the onset or the collapse of a bubble in financial markets has been one of
the most active research areas in financial econometrics; we refer the
reader, in particular, to the seminal articles on ex-post detection by %
\citet{phillips2011}, and \citet{phillips2015testing2}, and also to %
\citet{skrobotov2023testing} for a review. As far as ex-ante detection -
that is, finding the onset or collapse of a bubble in real time, as new data
come in - is concerned, this important issue has also been studied in
numerous recent contributions; although a comprehensive literature review
goes beyond the scope of this paper, we refer to the articles by %
\citet{homm2012testing}, \citet{phillips2015testing}, and %
\citet{whitehouse2023real}, \textit{inter alia}; in particular, the paper by %
\citet{whitehouse2023real} also contains a comprehensive literature review
of in-sample and online bubble detection methods. A common trait to the vast
majority of the existing literature is its reliance on a linear
specification, usually an AutoRegressive (AR) model, to capture regime
changes in the dynamics of log prices. Whilst such a modelling choice can be
justified from the theoretical point of view (see e.g. %
\citealp{phillips2011dating}), and whilst its analytical tractability offers
obvious advantages from a mathematical standpoint (see e.g. %
\citealp{phillips2007limit}, and \citealp{aue2007limit}), a major issue is
that using an AR model is fraught with difficulties when monitoring for
changes from an explosive towards a stationary regime. Several promising
solutions have been proposed such as the reverse regression approach by %
\citet{phillips2018financial}. \citet{horvath2023real} propose a different
model, based on a Random Coefficient Autoregressive (RCA) specification,
where inference is always standard normal irrespective of stationarity or
the lack thereof (\citealp{aue2011}), and develop a family of sequential
monitoring procedures based on the weighted CUSUM process, to check whether
the deterministic part of the autoregressive root changes over time. Such a
testing set-up also encompasses both the case of a switch from a stationary
to an explosive regime (thus indicating the start of a bubble phenomenon),
and a change from an explosive to a stationary regime (thus signalling the
collapse of a bubble).

The theory developed in \citet{horvath2023real} paves the way to a more
general research question, namely developing sequential monitoring
techniques which are robust to both the initial regime (i.e., which can be
employed irrespective of whether the observations in the training/historical
sample are stationary or explosive), and to the type of change which occurs
after a changepoint (i.e., which can detect changes from stationarity to
another stationary regime, or to a explosive regime; or from an explosive
regime to another explosive regime or a stationary one). From a technical
viewpoint, this question is nontrivial for at least three\ reasons. First,
proposing a changepoint detection methodology whose asymptotics is the same
irrespective of stationarity or explosivity is not easy \textit{per se},
because the partial sum processes which constitute the building blocks of
e.g. CUSUM-based statistics require completely different approximations
depending on whether the observations are stationary or not. Second, in
order to ensure timely changepoint detection, weighted versions of the CUSUM
process need to be considered, with different sets of weights ensuring
optimal detection delays for different changepoint locations within the
monitoring horizon. Third, on account of the previous point, it is important
to offer, to the applied user, a battery of results on the limiting
distribution of the detection delays, so as to gauge the expected detection
delay depending on the location of the changepoint and the weighing scheme
employed.

Motivated by the questions above, in this paper we investigate, with an
emphasis on completeness, the issue of sequential detection for regime
changes in a GARCH(1,1) sequence%
\begin{equation}
y_{i}=\sigma _{i}\epsilon _{i}\quad \;\;\mbox{and}\quad \;\;\sigma
_{i}^{2}=\omega +\alpha y_{i-1}^{2}+\beta \sigma _{i-1}^{2},\;\;\;1\leq
i<\infty .  \label{arch}
\end{equation}%
In particular, we develop two families of detectors based on the weighted
CUSUM process of the quasi-Fisher scores associated with (\ref{arch}): one
with \textquotedblleft mild\textquotedblright\ weights, designed to detect
changes that may occur \textquotedblleft not too early\textquotedblright\
after the start of the monitoring period; and one with heavy weights,
designed instead to detect changes occurring \textquotedblleft very
early\textquotedblright\ after the start of the monitoring period. The
latter is based on applying to the CUSUM process a set of (heavy) weights,
resulting in a family of test statistics known as \textit{R\'{e}nyi
statistics} (see \citealp{horvath2021} and \citealp{horvathmiller}, for
in-sample tests, and \citealp{ghezzi2024fast}, for sequential monitoring).
Our methodologies can, in principle, be applied to detect any type of change
in the vector $\left( \omega ,\alpha ,\beta \right) $. However, seeing as:
(1) our main interest is in detecting changes between regimes (e.g., from
stationarity to nonstationarity, or vice versa); the stationarity or lack
thereof of (\ref{arch}) is determined solely by $\alpha $ and $\beta $; and
(3) in the presence of nonstationarity, $\omega $ is not identified (%
\citealp{jensen2004asymptotic}), we focus on monitoring for changes only in
the sub-vector $\left( \alpha ,\beta \right) $. Detecting shifts in the
behaviour of the (conditional) volatility process $\sigma _{i}$ is important
in general; as \citet{hillebrand2005neglecting} notes, when neglecting a
break inference is biased in finite samples, and the sum of the estimated
autoregressive parameters $\alpha $ and $\beta $ converges to one.
Furthermore, changes (and, occasionally, explosions) in the volatility of
time series are often observed in practice (see e.g. %
\citealp{bloom2007uncertainty}, and \citealp{jurado2015measuring}). Hence,
finding the start or the end of an explosive regime in $\sigma _{i}$ is of
practical relevance because, as \citet{richter2023testing} put it, one
\textquotedblleft often sees sudden, integrated, or mildly explosive
behaviour in the second moment of the process which bounces back after a
while\textquotedblright\ (p. 468). Changes between stationarity and
explosivity in (\ref{arch}) can be interpreted as \textit{volatility bubbles}%
, i.e. events in which the second moment of the data (rather than the data,
e.g. prices, themselves) experiences periods of exhuberance. The link
between a volatility phenomenon and a \textquotedblleft
proper\textquotedblright\ bubble has not been fully explored yet (see %
\citealp{jurado2015measuring}), and, empirically, explosive regimes in
volatility can be ascribed to various sources in addition to bubbles (see %
\citealp{sornette2018can}, and \citealp{richter2023testing}). Nevertheless,
as \citet{jarrow2023explosion} put it, \textquotedblleft price bubbles
result from excess speculative trading decoupled from the asset's
fundamentals (dividends and liquidation value), which increases the asset's
price volatility to extreme levels\textquotedblright\ (p. 478). Hence, an
analysis of the regimes of the volatility of financial variables is bound to
contribute to a better understanding of bubble phenomena.

Based on the discussion above, in this paper we propose a battery of tests
for the sequential monitoring of the volatility of financial variables,
which complements the existing tests for bubbles based on the conditional
mean. Specifically, we make at least three contributions. First, we study
online detection of changes between stationarity and explosivity and vice
versa, which - to the best of our knowledge - is a novel result in the
literature, and which complements the ex-post detection statistics studied
in \citet{richter2023testing}. Second, we develop the full-blown theory for R%
\'{e}nyi statistics in the context of sequential monitoring of a GARCH(1,1)
model. Third, we derive the limiting distribution of detection delays for
all our monitoring schemes, including those based on R\'{e}nyi statistics;
this is an entirely novel result, which complements the results in %
\citet{horvath2020sequential}.\newline

The remainder of the paper is organised as follows. We discuss our workhorse
model and the main assumptions, as well as the test statistics, in Section %
\ref{model}. The theory is reported in Section \ref{asymptotics}: in
particular, the asymptotics under the null is in Section \ref{asy-null}, and
the full-blown asymptotics of the detection delay in the presence of a
changepoint is in Section \ref{asy-alt}. We validate our theory through a
comprehensive set of simulations (Section \ref{simulations}), and an
empirical application to daily returns of individual stocks (Section \ref%
{empirics}). Section \ref{conclusion} concludes.

NOTATION. We define the Euclidean norm of a vector $a$ as $\Vert a\Vert $.
We denote the integer part of a number as $\left\lfloor \cdot \right\rfloor $%
. We use: \textquotedblleft a.s\textquotedblright\ for \textquotedblleft
almost sure(ly)\textquotedblright ; \textquotedblleft $\rightarrow $%
\textquotedblright\ for the ordinary limit; \textquotedblleft $\overset{%
\mathcal{D}}{\rightarrow }$\textquotedblright\ for convergence in
distribution; \textquotedblleft $\overset{\mathcal{P}}{\rightarrow }$%
\textquotedblright\ for convergence in probability; \textquotedblleft $%
\overset{{\mathcal{D}}}{=}$\textquotedblright\ for equality in distribution.
Positive, finite constants are denoted as $c_{0}$, $c_{1}$, ... and their
value may change from line to line. Other, relevant notation is introduced
later on in the paper.

\section{Model, assumptions, and hypothesis testing\label{model}}

\subsection{Model, assumptions and hypotheses of interest\label{assume}}

The time dependent GARCH (1,1) sequence is defined by the recursion 
\begin{equation}
y_{i}=\sigma _{i}\epsilon _{i}\quad \;\;\mbox{and}\quad \;\;\sigma
_{i}^{2}=\omega _{i}+\alpha _{i}y_{i-1}^{2}+\beta _{i}\sigma
_{i-1}^{2},\;\;\;1\leq i<\infty ,  \label{garch}
\end{equation}%
where $\sigma _{0}^{2}$, $y_{0}^{2}$ are initial values, and $\alpha _{i}$, $%
\beta _{i}$ and $\omega _{i}$ are positive parameters.

Whilst the hypothesis testing framework is spelt out below, our monitoring
schemes are all based on the maintained assumption that we have $m$
observations which form a stable period (this is also known as the \textit{%
non-contamination assumption }in \citealp{CSW96}), viz. 
\begin{equation}
(\omega _{1},\alpha _{1},\beta _{1})=(\omega _{2},\alpha _{2},\beta
_{2})=\ldots =(\omega _{m},\alpha _{m},\beta _{m}).  \label{train}
\end{equation}%
We denote the value of the common parameter in \eqref{train} as $%
\mbox{\boldmath${\theta}$}_{0}=(\alpha _{0},\beta _{0},\omega _{0})^{\top }$%
. Prior to spelling out the main assumptions, we review the conditions for
the stationarity of $y_{i}$. As \citet{nelson1991conditional} shows (see
also \citealp{bougerol1992strict}, and \citealp{francq2012strict})

\vspace{-3pt}

\begin{enumerate}
\item if $E\log \left\vert \alpha _{0}\epsilon _{0}^{2}+\beta
_{0}\right\vert <\infty $, then $\sigma _{i}$ converges exponentially fast
to a unique, strictly stationary and ergodic solution $\left\{ \overline{%
\sigma }_{i},-\infty <i<\infty \right\} $ for all initial values $\epsilon
_{0}$ and $\sigma _{0}$;

\item if $E\log \left\vert \alpha _{0}\epsilon _{0}^{2}+\beta
_{0}\right\vert >\infty $, then $\sigma _{i}$ is nonstationary with $\sigma
_{i}\overset{a.s.}{\rightarrow }\infty $ exponentially fast (%
\citealp{nelson1991conditional});

\item if $E\log \left\vert \alpha _{0}\epsilon _{0}^{2}+\beta
_{0}\right\vert =\infty $, then $\sigma _{i}$ is nonstationary, but this is
a much more delicate case; indeed, \citet{kluppelberg2004continuous} show
that $\sigma _{i}\overset{\mathcal{P}}{\rightarrow }\infty $, but a.s.
divergence to infinity cannot be established. \newline
\end{enumerate}

\vspace{-12pt} We will develop several monitoring schemes for the null
hypothesis that the parameter $\mbox{\boldmath${\theta}$}_{0}$ undergoes no
changes after the historical training period $1\leq i\leq m$, i.e.%
\begin{equation}
H_{0}:\left( \omega _{m+k},\alpha _{m+k},\beta _{m+k}\right) =\left( \alpha
_{0},\beta _{0},\omega _{0}\right) \text{, for all }k\geq 1.  \label{null-h}
\end{equation}%
Under the alternative, we assume that there is a change at time $m+k^{\ast }$%
; whilst this would correspond to having $\left( \omega _{m+k^{\ast
}-j},\alpha _{m+k^{\ast }-j},\beta _{m+k^{\ast }-j}\right) $ $\neq $ $\left(
\omega _{m+k^{\ast }+j},\alpha _{m+k^{\ast }+j},\beta _{m+k^{\ast
}+j}\right) $ for all $j\geq 0$, it is well know that the $\omega _{i}$'s
cannot be identified in explosive, nonstationary regimes (%
\citealp{francq2012strict}). Hence, we will test for%
\begin{align}
H_{A}:\;& (\alpha _{m},\beta _{m})=(\alpha _{m+1},\beta _{m+1})=(\alpha
_{m+2},\beta _{m+2})  \label{alt} \\
& =\ldots =(\alpha _{m+k^{\ast }-1},\beta _{m+k^{\ast }-1})\neq (\alpha
_{m+k^{\ast }},\beta _{m+k^{\ast }})=  \notag \\
& =(\alpha _{m+k^{\ast }+1},\beta _{m+k^{\ast }+1})=\ldots  \notag
\end{align}%
i.e., for the possible presence of changes in $\alpha _{i}$ and $\beta _{i}$
only. Note that these are anyway the parameters of interest, since the
stationarity (or lack thereof) of $y_{i}$ is not affected by $\omega _{i}$.

We require the following assumptions on $\mbox{\boldmath${\theta}$}_{0}$,
and on the innovations $\{\epsilon _{i},-\infty <i<\infty \}$.

\begin{assumption}
\label{aspos}It holds that: $\alpha _{0}>0$, $\beta _{0}>0$ and $\omega
_{0}>0$.
\end{assumption}

\begin{assumption}
\label{asin}It holds that: \textit{(i)} $\{\epsilon _{i},-\infty <i<\infty
\} $ are independent and identically distributed random variables; \textit{%
(ii)} $\epsilon _{0}^{2}$ is nondegenerate; \textit{(iii)} $E\epsilon _{0}=0$%
, $E\epsilon _{i}^{2}=1$, $0<\mbox{\rm var}(\epsilon _{0}^{2})$, and $%
E|\epsilon _{0}|^{\kappa }<\infty $ with some $\kappa >4$.
\end{assumption}

Assumptions \ref{aspos} and \ref{asin} are standard. In particular, it is
worth noting that, in Assumptions \ref{aspos}, there is no requirement as to
the stationarity properties of $\left\{ y_{i},1\leq i\leq m\right\} $: the
observations in the training sample can belong to a stationary or explosive
volatility regime.

\vspace{-12pt}

\subsection{Estimation and monitoring schemes\label{scheme}}

As stated in the Introduction, the main purpose of our analysis is to offer
a detection scheme which finds changes in the parameters of a GARCH(1,1)
model as soon as possible after the training period. In this section, we
propose several detectors, all based on the CUSUM process of the
quasi-Fisher scores.

As is typical, we estimate the parameter $\mbox{\boldmath${\theta}$}_{0}$,
using the training sample, by Quasi Maximum Likelihood (QML). The $\log $
likelihood function is given by 
\begin{equation*}
\bar{\ell}_{i}(\mbox{\boldmath${\theta}$})=\log \bar{\sigma}_{i}^{2}(%
\mbox{\boldmath${\theta}$})+\frac{y_{i}^{2}}{\bar{\sigma}_{i}^{2}(%
\mbox{\boldmath${\theta}$})},\;\;\;\;\;1\leq i\leq m,
\end{equation*}%
where $\mbox{\boldmath${\theta}$}=(\omega ,\alpha ,\beta )^{\top }$, $\omega
>0,\alpha >0$ and $\beta >0$, and the random functions $\bar{\sigma}_{i}^{2}(%
\mbox{\boldmath${\theta}$})$ given by the recursion 
\begin{equation*}
\bar{\sigma}_{i}^{2}(\mbox{\boldmath${\theta}$})=\omega +\alpha
y_{i-1}^{2}+\beta \bar{\sigma}_{i-1}^{2}(\mbox{\boldmath${\theta}$}%
),\;\;\;\;\;1\leq i\leq m;
\end{equation*}%
the recursion starts from the initial values $y_{0}$ and $\bar{\sigma}%
_{0}^{2}$. The QML estimator computed from the historical sample is denoted
as $\hat{\mbox{\boldmath${\theta}$}}_{m}$, with 
\begin{equation*}
\hat{\mbox{\boldmath${\theta}$}}_{m}=\mbox{\rm argmax}\left\{ \sum_{i=1}^{m}%
\bar{\ell}_{i}(\mbox{\boldmath${\theta}$}):\;\mbox{\boldmath${\theta}$}\in %
\mbox{\boldmath${\theta}$}\right\} ,
\end{equation*}%
where the (compact) space $\mbox{\boldmath${\theta}$}$ is defined as%
\begin{equation*}
\mbox{\boldmath${\theta}$}=\left\{ \mbox{\boldmath${\theta}$}:\underline{%
\omega }\leq \omega \leq \overline{\omega },\underline{\alpha }\leq \alpha
\leq \overline{\alpha },\underline{\beta }\leq \beta \leq \overline{\beta }%
\right\} ,
\end{equation*}%
$0<\underline{\omega },\overline{\omega },\underline{\alpha },\overline{%
\alpha },\underline{\beta },\overline{\beta }<\infty $.

Starting with the initial values ${y}_{m}$ and ${\sigma }_{m}^{2}$, we
define the random functions $\bar{\sigma}_{m+k}^{2}(\mbox{\boldmath${%
\theta}$})$ based on the observations after the historical sample by the
recursion 
\begin{equation*}
\bar{\sigma}_{m+k}^{2}(\mbox{\boldmath${\theta}$})=\omega +\alpha
y_{m+k-1}^{2}+\beta \bar{\sigma}_{m+k-1}^{2}(\mbox{\boldmath${\theta}$}%
),\;\;\;\;\;k\geq 1,
\end{equation*}%
with the log likelihood function given by 
\begin{equation*}
\bar{\ell}_{m+k}(\mbox{\boldmath${\theta}$})=\log \bar{\sigma}_{m+k}^{2}(%
\mbox{\boldmath${\theta}$})+\frac{y_{m+k}^{2}}{\bar{\sigma}_{m+k}^{2}(%
\mbox{\boldmath${\theta}$})},\;\;\;\;\;k\geq 1.
\end{equation*}%
Hence, we define the CUSUM\ process of the quasi-Fisher scores as 
\begin{equation}
{\mathcal{r}}_{m,k}(\mbox{\boldmath${\theta}$})=\sum_{i=m+1}^{m+k}\left( 
\frac{\partial \bar{\ell}_{i}(\mbox{\boldmath${\theta}$})}{\partial \alpha },%
\frac{\partial \bar{\ell}_{i}(\mbox{\boldmath${\theta}$})}{\partial \beta }%
\right) ^{\top }\text{, \ \ }k\geq 1.  \label{detector}
\end{equation}%
Heuristically, under the null of no change, the scores have zero mean;
hence, the partial sum process ${\mathcal{r}}_{m,k}(\mbox{\boldmath${%
\theta}$})$ calculated at $\hat{\mbox{\boldmath${\theta}$}}_{m}$\ should
also fluctuate around zero with increasing variance. Conversely, in the
presence of a break (at, say, $k^{\ast }$), $\hat{\mbox{\boldmath${\theta}$}}%
_{m}$ is a biased estimator for the \textquotedblleft new\textquotedblright\
parameter $\mbox{\boldmath${\theta}$}_{m+k^{\ast }}$; thus, ${\mathcal{r}}%
_{m,k}(\mbox{\boldmath${\theta}$})$, calculated at $\hat{\mbox{\boldmath${%
\theta}$}}_{m}$, should have a drift term. In the light of these heuristic
considerations, we propose the following \textit{detector}%
\begin{equation}
{\mathcal{D}}_{m}(k)={\mathcal{r}}_{m,k}^{\top }(\hat{\mbox{\boldmath${%
\theta}$}}_{m})\hat{\mathbf{D}}_{m}^{-1}{\mathcal{r}}_{m,k}(\hat{%
\mbox{\boldmath${\theta}$}}_{m}),  \label{detect}
\end{equation}%
where 
\begin{equation*}
\hat{\mathbf{D}}_{m}=\frac{1}{m}\sum_{i=1}^{m}\left( \frac{\partial \bar{\ell%
}_{i}(\hat{\mbox{\boldmath${\theta}$}}_{m})}{\partial \alpha },\frac{%
\partial \bar{\ell}_{i}(\hat{\mbox{\boldmath${\theta}$}}_{m})}{\partial
\beta }\right) \left( \frac{\partial \bar{\ell}_{i}(\hat{\mbox{\boldmath${%
\theta}$}}_{m})}{\partial \alpha },\frac{\partial \bar{\ell}_{i}(\hat{%
\mbox{\boldmath${\theta}$}}_{m})}{\partial \beta }\right) ^{\top }.
\end{equation*}%
Based on (\ref{detect}), a break is flagged as soon as the detector ${%
\mathcal{D}}_{m}(k)$ exceeds a threshold. We call such a threshold the 
\textit{boundary function}. Similarly to \citet{CSW96}, \citet{lajos04}, %
\citet{horvath2020sequential}, \citet{homm2012testing}, we use the boundary
function, designed for a \textit{closed-ended procedure} - i.e., for a
procedure which terminates at a certain time, say ${\mathcal{n}}$, if there
is no change 
\begin{equation}
g_{m}(k)={{\mathcal{c}}{\mathcal{n}}(k/{\mathcal{n}})^{\eta }},\quad \;\;%
\mbox{with}\;\;0\leq \eta <1.  \label{boundary}
\end{equation}%
On account of (\ref{detect}) and (\ref{boundary}), a changepoint is found at
a stopping time $\tau _{m}$ defined as 
\begin{equation}
\tau _{m}=\left\{ 
\begin{array}{ll}
\min \left\{ k:\in \lbrack 1,2,\ldots ,{\mathcal{n}}-1],\;\;\;{\mathcal{D}}%
_{m}(k)\geq g_{m}(k)\right\} \vspace{0.3cm} &  \\ 
{\mathcal{n}},\;\mbox{if}\;\;{\mathcal{D}}_{m}(k)<g_{m}(k)\;\;\;%
\mbox{for
all}\;\;1\leq k\leq {\mathcal{n}}-1. & 
\end{array}%
\right.  \label{tau}
\end{equation}

The (user-chosen) parameter $\eta $ in (\ref{boundary}) determines the
timeliness of changepoint detection of our sequential monitoring procedure. %
\citet{aue2004delay} and \citet{aue2008monitoring} show that, as $\eta $
approaches $1$, changepoints are detected with a smaller and smaller delay
depending on their location; on the other hand, different values of $\eta $
work better for different changepoint locations, as also pointed out in a
recent contribution by \citet{kirch2022sequential}. In particular, values of 
$0\leq \eta <1$ are able to offer short detection delays for breaks that do
not occur \textquotedblleft too early\textquotedblright\ after $m$.

In order to detect earlier changes, \citet{ghezzi2024fast} (building on
previous contributions by \citealp{horvath2021}, and \citealp{horvathmiller}%
, developed for in-sample changepoint detection) suggest using R\'{e}nyi
type statistics, with stopping rule 
\begin{equation}
\bar{\tau}_{m}=\left\{ 
\begin{array}{ll}
\min \left\{ k:\in \lbrack r,2,\ldots ,{\mathcal{n}}-1],:\;\;{\mathcal{D}}%
_{m}(k)\geq \bar{g}_{m}(k)\right\} \vspace{0.3cm} &  \\ 
{\mathcal{n}},\;\mbox{if}\;\;{\mathcal{D}}_{m}(k)<\bar{g}_{m}(k)\;\;\;%
\mbox{for all}\;\;r\leq k\leq {\mathcal{n}}-1, & 
\end{array}%
\right.   \label{tau-renyi}
\end{equation}%
where $r$ is a trimming sequence specified in Assumption \ref{asrr}, and 
\begin{equation}
\bar{g}_{m}(k)={{\mathcal{c}}r(k/r)^{\eta }},\quad \;\;\mbox{with}\;\;\eta
>1.  \label{renyiboundaryw}
\end{equation}

We note that (\ref{tau}) and (\ref{tau-renyi}) exclude the case $\eta =1$.
Indeed, \citet{aue2004delay} show that, as far as stopping times under the
alternative are concerned, using $\eta =1$ would produce the shortest
detection time. The case $\eta =1$ requires to be studied separately.
Following \citet{lajos07}, we modify the boundary function. Let 
\begin{equation*}
a(x)=(2\log x)^{1/2}\;\quad \;\mbox{and}\;\quad b_{2}(x)=2\log x+\log \log x.
\end{equation*}%
We use the boundary functions 
\begin{equation}
g_{m}^{\ast }(k)=k\left( \frac{{\mathcal{c}}+b_{2}(\log {\mathcal{n}})}{%
a(\log {\mathcal{n}})}\right) ^{2},  \label{boundary-de1}
\end{equation}%
and%
\begin{equation}
\bar{g}_{m}^{\ast }(k)=k\left( \frac{{\mathcal{c}}+b_{2}(\log ({\mathcal{n}}%
/r))}{a(\log ({\mathcal{n}}/r))}\right) ^{2}.  \label{boundary-de2}
\end{equation}%
The corresponding stopping times - denoted as $\tau _{m}^{\ast }$ and $\bar{%
\tau}_{m}^{\ast }$ - are defined exactly in the same way as $\tau _{m}$ and $%
\bar{\tau}_{m}$ in (\ref{tau}) and (\ref{tau-renyi}), respectively, using
now the boundaries $g_{m}^{\ast }(k)$ and $\bar{g}_{m}^{\ast }(k)$.\newline

%\medskip

As mentioned above, we consider a closed-ended scheme, which terminates ${%
\mathcal{n}}$ periods after $m$. The following assumptions characterise the
length of the monitoring horizon and of the trimming sequence $r$ defined in
(\ref{tau-renyi}); in particular, Assumption \ref{assh2} is designed in
order to consider only early changepoint detection.

\begin{assumption}
\label{assh2}It holds that ${\mathcal{n}}\rightarrow \infty $, and ${%
\mathcal{n}}/m\rightarrow 0$.
\end{assumption}

\begin{assumption}
\label{asrr}It holds that, in equation (\ref{tau-renyi}), $r\rightarrow
\infty $ and $r/{\mathcal{n}}\rightarrow 0$.
\end{assumption}

\section{Asymptotics\label{asymptotics}}

In this section, we investigate the asymptotic behaviour of our monitoring
schemes under the null and under the alternative hypotheses.

\subsection{Asymptotics under the null\label{asy-null}}

Let $\mathbf{W}(t)=(W_{1}(t),W_{2}(t)),t\geq 0$ be a two dimensional
standard Wiener process - i.e.,\ $\{W_{1}(t),t\geq 0\}$ and $%
\{W_{2}(t),t\geq 0\}$ are two independent Gaussian processes with $%
EW_{1}(t)=EW_{2}(t)=0$, and covariance kernel $EW_{1}(t)W_{1}(s)$ $=$ $%
EW_{2}(t)W_{2}(s)$ $=\min (t,s)$.

\begin{theorem}
\label{ma1} We assume that $H_{0}$ of (\ref{null-h}) and Assumptions \ref%
{aspos}--\ref{assh2} hold, and that $E\log (\alpha _{0}\epsilon
_{0}^{2}+\beta _{0})\neq 0$.\newline
(i) If $0\leq \eta <1$, then we have 
\begin{equation*}
\lim_{m\rightarrow \infty }P\{\tau _{m}={\mathcal{n}}\}=P\left\{
\sup_{0<t\leq 1}\frac{1}{t^{\eta }}\left\Vert \mathbf{W}(t)\right\Vert
^{2}\leq {\mathcal{c}}\right\} .
\end{equation*}%
(ii) If in addition, Assumption \ref{asrr} also holds and $\eta >1$, then we
have 
\begin{equation*}
\lim_{m\rightarrow \infty }P\{\bar{\tau}_{m}={\mathcal{n}}\}=P\left\{
\sup_{1\leq t<\infty }\frac{1}{t^{\eta }}\left\Vert \mathbf{W}(t)\right\Vert
^{2}\leq {\mathcal{c}}\right\} .
\end{equation*}
\end{theorem}

Theorem \ref{ma1} offers a rule to calculate the asymptotic critical values;
for a given nominal level $\alpha $, the critical value ${\mathcal{c}}%
_{\alpha }$ is defined as%
\begin{equation*}
P\left\{ \sup_{0<t\leq 1}\frac{1}{t^{\eta }}\left\Vert \mathbf{W}%
(t)\right\Vert ^{2}\leq {\mathcal{c}}_{\alpha }\right\} =1-\alpha ,
\end{equation*}%
for all $0\leq \eta <1$, and 
\begin{equation*}
P\left\{ \sup_{1<t<\infty }\frac{1}{t^{\eta }}\left\Vert \mathbf{W}%
(t)\right\Vert ^{2}\leq {\mathcal{c}}_{\alpha }\right\} =1-\alpha ,
\end{equation*}%
for $\eta >1$. Using the scale transformation of the Wiener process, it
immediately follows that $\{\mathbf{W}(t),t>0\}$ $\overset{{\mathcal{D}}}{=}$
$\{t\mathbf{W}(1/t),t>0\}$; hence, for all $\eta >1$ 
\begin{equation*}
\sup_{1\leq t<\infty }\frac{1}{t^{\eta }}\Vert \mathbf{W}(t)\Vert ^{2}%
\overset{{\mathcal{D}}}{=}\sup_{0<t\leq 1}\frac{1}{t^{1-\eta }}\Vert \mathbf{%
W}(t)\Vert ^{2}.
\end{equation*}%
Theorem \ref{ma1} rules out the boundary case $E\log (\alpha _{0}\epsilon
_{0}^{2}+\beta _{0})=0$; this is because we would need an exact (and large
enough) rate of divergence for $\left\vert \sigma _{i}\right\vert $ as $%
i\rightarrow \infty $, but this result is not available in the case $E\log
(\alpha _{0}\epsilon _{0}^{2}+\beta _{0})=0$ (see \citealp{francq2012strict}%
, p. 823; and also Theorem 4 in \citealp{HT2019}, where a similar problem is
encountered, and the discussion thereafter).\newline

%\medskip

Theorem \ref{ma1} does not consider the case $\eta =1$, which corresponds to
the stopping times $\tau _{m}^{\ast }$ and $\bar{\tau}_{m}^{\ast }$ based on
the boundaries defined in (\ref{boundary-de1}) and (\ref{boundary-de2})
respectively. The case $\eta =1$ is studied separately in the following
theorem.

\begin{theorem}
\label{ma2} We assume that $H_{0}$ of (\ref{null-h}) and Assumptions \ref%
{aspos}--\ref{assh2} hold, and that $E\log (\alpha _{0}\epsilon
_{0}^{2}+\beta _{0})\neq 0$. Then, for all $-\infty <{\mathcal{c}}<\infty $,
it holds that\newline
(i) $\lim_{m\rightarrow \infty }P\{\tau _{m}^{\ast }={\mathcal{n}}\}=\exp
\left( -e^{-{\mathcal{c}}}\right) $.

(ii) If in addition Assumption \ref{asrr} also holds, then we have $%
\lim_{m\rightarrow \infty }P\{\bar{\tau}_{m}^{\ast }={\mathcal{n}}\}=\exp
\left( -e^{-{\mathcal{c}}}\right) .$
\end{theorem}

Theorem \ref{ma2} stipulates that the asymptotic critical values, for a
given nominal level $\alpha $, can be calculated as%
\begin{equation}
{\mathcal{c}}_{\alpha }=\overline{{\mathcal{c}}}_{\alpha }=-\log \left(
-\log \left( 1-\alpha \right) \right) ,  \label{asy-crv-de-2}
\end{equation}%
using ${\mathcal{c}}_{\alpha }$\ or $\overline{{\mathcal{c}}}_{\alpha }$\
according as (\ref{boundary-de1}) or (\ref{boundary-de2}) is employed.
Although the theorem offers an explicit formula to compute asymptotic
critical values, these are bound to be inaccurate due to the slow
convergence to the Extreme Value distribution. In particular, simulations
show that, in finite samples, asymptotic critical values overstate the true
values thus leading to low power.

\subsection{Asymptotics under the alternative\label{asy-alt}}

We now study the behaviour of our monitoring schemes under the alternative,
focussing, in particular, on the limiting distribution of the detection
delay. We report the limiting distribution of the detection delay when using 
$0\leq \eta <1$ in Section \ref{delay-1}; in Section \ref{delay-2}, we
report the limiting distribution of the detection delay when using $\eta >1$.

In both cases, under the alternative $H_{A}$, of (\ref{alt}), the parameter $%
\mbox{\boldmath${\theta}$}_{0}=(\alpha _{0},\beta _{0},\omega _{0})^{\top }$
changes to $\mbox{\boldmath${\theta}$}_{A}=(\alpha _{A},\beta _{A},\omega
_{A})^{\top }$ satisfying

\begin{assumption}
\label{asalt}\;$\alpha_A>0, \beta_A>0, \omega_0>0$, $\bar{%
\mbox{\boldmath${\theta}$}}_0 \neq \bar{\mbox{\boldmath${\theta}$}}_A$, $%
\bar{\mbox{\boldmath${\theta}$}}_0=(\alpha_0,\beta_0)^\top$ and $\bar{%
\mbox{\boldmath${\theta}$}}_A=(\alpha_A, \beta_A)^\top$.
\end{assumption}

\subsubsection{Detection delays with $0\leq \protect\eta <1$\label{delay-1}}

We begin by investigating the asymptotic behaviour of the stopping time $%
\tau _{m}$ defined in (\ref{tau}). Whilst the result in Theorem \ref{thcons}
below is valid for all cases, we need to introduce some preliminary
notation, separately, for the two cases: (1) when the sequence is \textit{%
stationary} after the change and (2) when the sequence is \textit{explosive}
after the change.\newline

%\medskip

\qquad \textit{Preliminary notation}\newline

%\medskip

We begin by introducing some preliminary notation for the former case, i.e. 
\begin{equation}
E\log (\alpha _{A}\epsilon _{0}^{2}+\beta _{A})<0.  \label{stalt}
\end{equation}%
Under (\ref{stalt}), after the change the observations are exponentially
close to $\{\hat{x}_{i},-\infty <i<\infty \}$, a stationary GARCH (1,1)
sequence given by 
\begin{equation}
\hat{x}_{i}=\hat{h}_{i}\epsilon _{i}\;\;\;\mbox{and}\;\;\;\hat{h}%
_{i}^{2}=\omega +\alpha _{A}\hat{x}_{i-1}^{2}+\beta _{A}\hat{h}_{i-1}^{2}.
\label{hatx}
\end{equation}%
We also define the log likelihood function%
\begin{equation*}
\hat{\ell}_{i}(\mbox{\boldmath${\theta}$})=\log \hat{h}_{i}^{2}(%
\mbox{\boldmath${\theta}$})+\frac{\hat{x}_{i}^{2}}{\hat{h}_{i}^{2}(%
\mbox{\boldmath${\theta}$})},
\end{equation*}%
where $\hat{h}_{i}^{2}(\mbox{\boldmath${\theta}$})=\omega +\alpha \hat{x}%
_{i-1}^{2}+\beta \hat{h}_{i-1}^{2}(\mbox{\boldmath${\theta}$})$. Let 
\begin{equation*}
{{\mathcal{r}}}^{(1)}(\mbox{\boldmath${\theta}$})=E\left( \frac{\partial 
\hat{\ell}_{i}(\mbox{\boldmath${\theta}$})}{\alpha },\frac{\partial \hat{\ell%
}_{i}(\mbox{\boldmath${\theta}$})}{\beta }\right) ^{\top },
\end{equation*}%
and define the size of the change as 
\begin{equation}
\mbox{\boldmath${\Delta}$}={{\mathcal{r}}}^{(1)}(\mbox{\boldmath${\theta}$}%
_{0})\neq \mathbf{0}.  \label{deldef}
\end{equation}%
We define the covariance matrix 
\begin{equation}
\mbox{\boldmath${\Sigma}$}_{1}=E\left[ \left( \frac{\partial \hat{\ell}_{i}(%
\mbox{\boldmath${\theta}$}_{0})}{\alpha },\frac{\partial \hat{\ell}_{i}(%
\mbox{\boldmath${\theta}$}_{0})}{\beta }\right) -\mbox{\boldmath${\Delta}$}%
\right] ^{\top }\left[ \left( \frac{\partial \hat{\ell}_{i}(%
\mbox{\boldmath${\theta}$}_{0})}{\alpha },\frac{\partial \hat{\ell}_{i}(%
\mbox{\boldmath${\theta}$}_{0})}{\beta }\right) -\mbox{\boldmath${\Delta}$}%
\right] ,  \label{Sig1}
\end{equation}%
and 
\begin{equation}
t^{\ast }=\lim_{m\rightarrow \infty }k^{\ast }\Biggl/\left( {\mathcal{n}}%
^{1-\eta }\frac{{\mathcal{c}}}{A_{m}}\right) ^{1/(2-\eta )},  \label{tstar}
\end{equation}%
where $A_{m}=\mbox{\boldmath${\Delta}$}^{\top }\hat{\mathbf{D}}_{m}^{-1}%
\mbox{\boldmath${\Delta}$}$. Finally (as far as preliminary notation is
concerned), we define $u_{\mathcal{n}}>0$ as the unique solution of the
equation 
\begin{equation}
u_{\mathcal{n}}^{2}=\left( u_{\mathcal{n}}+k^{\ast }\Biggl/\left( {\mathcal{n%
}}^{1-\eta }\frac{{\mathcal{c}}}{A_{m}}\right) ^{1/(2-\eta )}\right) ^{\eta
},  \label{umdef}
\end{equation}%
and $u^{\ast }>0$ as the solution of 
\begin{equation}
u^{\ast }=\left( u^{\ast }+t^{\ast }\right) ^{\eta /2}.  \label{ustar}
\end{equation}%
It is easy to see that $u_{\mathcal{n}}\rightarrow u^{\ast }$, and that $%
u^{\ast }=1$, if $t^{\ast }=0$. We are now ready to introduce the main
notation to spell out the properties of the stopping time $\tau _{m}$ when
the observations change into a stationary sequence.\newline

%\medskip

We now introduce the preliminary notation for the case when the observations
turn into an explosive sequence after the time of change, i.e. 
\begin{equation}
E\log (\alpha _{A}\epsilon _{0}^{2}+\beta _{A})>0.  \label{nostalst}
\end{equation}%
\citet{jensen2004asymptotic} proved that 
\begin{equation*}
\lim_{k\rightarrow \infty }\frac{1}{k}\sum_{i=m+k^{\ast }+1}^{m+k^{\ast
}+k}E\left( \frac{\partial \bar{\ell}_{i}(\mbox{\boldmath${\theta}$})}{%
\alpha },\frac{\partial \bar{\ell}_{i}(\mbox{\boldmath${\theta}$})}{\beta }%
\right) ^{\top }={\mathcal{r}}^{(2)}(\mbox{\boldmath${\theta}$}),
\end{equation*}%
exists. Similarly to $\mbox{\boldmath${\Delta}$}$ in equation (\ref{deldef}%
), we define the size of the change under $H_{A}$ as $\mbox{\boldmath${%
\Upsilon}$}={\mathcal{r}}^{(2)}(\mbox{\boldmath${\theta}$}_{0})\neq \mathbf{0%
}$. Similarly to (\ref{tstar}), we define%
\begin{equation*}
\tilde{t}^{\ast }=\lim_{m\rightarrow \infty }k^{\ast }\Biggl/\left( {%
\mathcal{n}}^{1-\eta }\frac{{\mathcal{c}}}{B_{m}}\right) ^{1/(2-\eta )},
\end{equation*}%
where $B_{m}=\mbox{\boldmath${\Upsilon}$}^{\top }\hat{\mathbf{D}}_{m}^{-1}%
\mbox{\boldmath${\Upsilon}$}$. Finally, similarly to $u_{\mathcal{n}}$ and $%
u^{\ast }$ we define $\tilde{u}_{\mathcal{n}}$ and $\tilde{u}^{\ast }$ as
the solutions of the equations 
\begin{equation*}
\tilde{u}_{\mathcal{n}}^{2}=\left( \tilde{u}_{\mathcal{n}}+k^{\ast }\Biggl/%
\left( {\mathcal{n}}^{1-\eta }\frac{{\mathcal{c}}}{B_{m}}\right) \right)
^{\eta },\text{ \ \ and \ \ }\tilde{u}^{\ast }=\left( \tilde{u}^{\ast }+%
\tilde{t}^{\ast }\right) ^{\eta /2}.
\end{equation*}%
After the change in the parameters, the gradient of the likelihood function
is approximated with the sequences 
\begin{equation*}
{\mathcal{v}}_{m+k^{\ast }+i,1}=\sum_{j=1}^{\infty }\epsilon _{m+k^{\ast
}+i-j}^{2}\frac{1}{\beta _{A}}\sum_{k=1}^{j}\frac{\beta _{A}}{\alpha
_{A}\epsilon _{m+k^{\ast }+i-k}^{2}+\beta _{A}}\quad \;
\end{equation*}%
and 
\begin{equation*}
{\mathcal{v}}_{m+k^{\ast }+i,2}=\sum_{j=1}^{\infty }\frac{1}{\beta _{A}}%
\sum_{k=1}^{j}\frac{\beta _{A}}{\alpha _{A}\epsilon _{m+k^{\ast
}+i-k}^{2}+\beta _{A}}.
\end{equation*}%
Analogously to $\mbox{\boldmath${\Sigma}$}_{1}$ in (\ref{Sig1}), we finally
introduce 
\begin{equation}
\mbox{\boldmath${\Sigma}$}_{2}=E\left( 1-\epsilon _{m+k^{\ast
}+1}^{2}\right) ^{2}E\left( 
\begin{array}{ll}
\displaystyle{\mathcal{v}}_{m+k^{\ast }+i,1}\vspace{0.2cm} &  \\ 
\displaystyle{\mathcal{v}}_{m+k^{\ast }+i,2} & 
\end{array}%
\right) \left( {\mathcal{v}}_{m+k^{\ast }+i,1},{\mathcal{v}}_{m+k^{\ast
}+i,2}\right) .  \label{Sig2}
\end{equation}

%\medskip

\qquad \textit{Main notation}\newline

%\medskip

After spelling out the preliminary notation for the two cases of the
observations begin stationary and nonstationary, we now introduce the main
notation. As we will see in Theorem \ref{thcons} below, in several cases the
delay $\tau _{m}-k^{\ast }$ converges to a standard normal random variable
after being centered and rescaled; the centering and rescaling for $\tau
_{m}-k^{\ast }$ depend on whether $t^{\ast }<\infty $ or $t^{\ast }=\infty $
($\tilde{t}^{\ast }<\infty $ or $\tilde{t}^{\ast }=\infty $, equivalently).
In particular

\vspace{-12pt}

\begin{enumerate}
\item when $t^{\ast }<\infty $ ($\tilde{t}^{\ast }<\infty $, respectively), $%
\tau _{m}-k^{\ast }$ will be centered around 
\begin{equation*}
v_{1,{\mathcal{n}}}=\left\{ 
\begin{array}{ll}
u_{\mathcal{n}}\left( \displaystyle{\mathcal{n}}^{1-\eta }\frac{\displaystyle%
{\mathcal{c}}}{\displaystyle A_{m}}\right) ^{1/(2-\eta )},\;\;\;\mbox{if}%
\;\;\;E\log (\alpha _{A}\epsilon _{0}^{2}+\beta _{A})<0,\vspace{0.3cm} &  \\ 
\tilde{u}_{\mathcal{n}}\left( \displaystyle{\mathcal{n}}^{1-\eta }\frac{%
\displaystyle{\mathcal{c}}}{\displaystyle B_{m}}\right) ^{1/(2-\eta )},\;\;\;%
\mbox{if}\;\;\;E\log (\alpha _{A}\epsilon _{0}^{2}+\beta _{A})>0, & 
\end{array}%
\right.
\end{equation*}%
and rescaled by%
\begin{equation*}
v_{2,{\mathcal{n}}}=\left\{ 
\begin{array}{ll}
\left( t^{\ast }\mbox{\boldmath${\Delta}$}^{\top }\mathbf{D}^{-1}%
\mbox{\boldmath${\Delta}$}+\left( 1-\displaystyle\frac{\displaystyle t^{\ast
}}{\displaystyle t^{\ast }+u^{\ast }}\right) \mbox{\boldmath${\Delta}$}%
^{\top }\mathbf{D}^{-1}\mbox{\boldmath${\Sigma}$}_{1}\mathbf{D}^{-1}%
\mbox{\boldmath${\Delta}$}\right) ^{1/2}(t^{\ast }+u^{\ast })^{3/2}\vspace{%
0.3cm} &  \\ 
\hspace{1.5cm}\;\;\;\;\times \left( \displaystyle{\mathcal{n}}^{1-\eta }%
\frac{\displaystyle{\mathcal{c}}}{\displaystyle A_{m}}\right) ^{1/(4-2\eta
)},\;\;\;\mbox{if}\;\;\;E\log (\alpha _{A}\epsilon _{0}^{2}+\beta _{A})<0, & 
\\ 
\vspace{0.3cm}\left( \tilde{t}^{\ast }\mbox{\boldmath${\Upsilon}$}^{\top }%
\mathbf{D}^{-1}\mbox{\boldmath${\Upsilon}$}+\left( 1-\displaystyle\frac{%
\displaystyle\tilde{t}^{\ast }}{\displaystyle\tilde{t}^{\ast }+\tilde{u}%
^{\ast }}\right) \mbox{\boldmath${\Upsilon}$}^{\top }\mathbf{D}^{-1}%
\mbox{\boldmath${\Sigma}$}_{2}\mathbf{D}^{-1}\mbox{\boldmath${\Upsilon}$}%
\right) ^{1/2}(\tilde{t}^{\ast }+\tilde{u}^{\ast })^{3/2} &  \\ 
\vspace{0.3cm}\hspace{1.5cm}\;\;\;\;\times \displaystyle\left( \displaystyle{%
\mathcal{n}}^{1-\eta }\frac{\displaystyle{\mathcal{c}}}{\displaystyle B_{m}}%
\right) ^{1/(4-2\eta )},\;\;\;\mbox{if}\;\;\;E\log (\alpha _{A}\epsilon
_{0}^{2}+\beta _{A})>0. & 
\end{array}%
\right.
\end{equation*}

\item when $t^{\ast }=\infty $ ($\tilde{t}^{\ast }<\infty $, respectively), $%
\tau _{m}-k^{\ast }$ will be centered around%
\begin{equation*}
v_{3,{\mathcal{n}}}=\left\{ 
\begin{array}{ll}
\displaystyle\left( \frac{\displaystyle{\mathcal{c}}}{\displaystyle A_{m}}%
\displaystyle{\mathcal{n}}^{1-\eta }(k^{\ast })^{\eta }\right) ^{1/2},\;\;\;%
\mbox{if}\;\;\;E\log (\alpha _{A}\epsilon _{0}^{2}+\beta _{A})<0,\vspace{%
0.3cm} &  \\ 
\displaystyle\left( \displaystyle\frac{\displaystyle{\mathcal{c}}}{%
\displaystyle B_{m}}\displaystyle{\mathcal{n}}^{1-\eta }(k^{\ast })^{\eta
}\right) ^{1/2},\;\;\;\mbox{if}\;\;\;E\log (\alpha _{A}\epsilon
_{0}^{2}+\beta _{A})>0, & 
\end{array}%
\right.
\end{equation*}%
and rescaled by 
\begin{equation*}
v_{4,{\mathcal{n}}}=\left\{ 
\begin{array}{ll}
\displaystyle\frac{\displaystyle(\mbox{\boldmath${\Delta}$}^{\top }\mathbf{D}%
^{-1}\mbox{\boldmath${\Delta}$})^{1/2}}{\displaystyle A_{m}}(k^{\ast
})^{1/2},\;\;\;\mbox{if}\;\;\;E\log (\alpha _{A}\epsilon _{0}^{2}+\beta
_{A})<0, &  \\ 
\vspace{0.5cm} &  \\ 
\displaystyle\frac{\displaystyle(\mbox{\boldmath${\Upsilon}$}^{\top }\mathbf{%
D}^{-1}\mbox{\boldmath${\Upsilon}$})^{1/2}}{\displaystyle B_{m}}(k^{\ast
})^{1/2},\;\;\;\mbox{if}\;\;\;E\log (\alpha _{A}\epsilon _{0}^{2}+\beta
_{A})>0. & 
\end{array}%
\right.
\end{equation*}
\end{enumerate}

In order to present the limiting distribution for both cases, we define: the
Gaussian process $\{\mbox{\boldmath${\Gamma}$}(t),t\geq 0\}$, with $E%
\mbox{\boldmath${\Gamma}$}(t)=\mathbf{0}$ and $E\mbox{\boldmath${\Gamma}$}(t)%
\mbox{\boldmath${\Gamma}$}^{\top }(s)=\min (t,s)\mathbf{D}$; the random
variables 
\begin{equation}
{\mathcal{A}}_{1}=\sup_{0<t\leq 1}\frac{1}{t^{\eta }}\mbox{\boldmath${%
\Gamma}$}(t)\mathbf{D}^{-1}\mbox{\boldmath${\Gamma}$}(t),  \label{a1}
\end{equation}%
\begin{equation}
{\mathcal{A}}_{2}(x)=\left\{ 
\begin{array}{ll}
\displaystyle\sup_{0<t\leq x}(t\mbox{\boldmath${\Delta}$}+%
\mbox{\boldmath${\Gamma}$}(1))^{\top }\mathbf{D}^{-1}(t\mbox{\boldmath${%
\Delta}$}+\mbox{\boldmath${\Gamma}$}(1)),\;\;\;\mbox{if}\;\;\;E\log (\alpha
_{A}\epsilon _{0}^{2}+\beta _{A})<0,\vspace{0.3cm} &  \\ 
\displaystyle\sup_{0<t\leq x}(t\mbox{\boldmath${\Upsilon}$}+%
\mbox{\boldmath${\Gamma}$}(1))^{\top }\mathbf{D}^{-1}(t\mbox{\boldmath${%
\Upsilon}$}+\mbox{\boldmath${\Gamma}$}(1)),\;\;\;\mbox{if}\;\;\;E\log
(\alpha _{A}\epsilon _{0}^{2}+\beta _{A})>0; & 
\end{array}%
\right.  \label{a2}
\end{equation}%
and, finally, the asymptotic variances 
\begin{equation}
\bar{s}_{1}^{2}=\left\{ 
\begin{array}{ll}
\begin{array}{l}
t^{\ast }\mbox{\boldmath${\Delta}$}^{\top }\mathbf{D}^{-1}%
\mbox{\boldmath${\Delta}$}+(1-t^{\ast }/(t^{\ast }+u^{\ast }))%
\mbox{\boldmath${\Delta}$}^{\top }\mathbf{D}^{-1}\mbox{\boldmath${\Sigma}$}%
_{1}\mathbf{D}^{-1}\mbox{\boldmath${\Delta}$},\;\;\;\; \\ 
\;\ \ \ \ \ \ \ \ \ \ \ \ \ \ \ \ \ \ \ \ \ \ \ \ \ \text{if \ \ \ }\;E\log
(\alpha _{A}\epsilon _{0}^{2}+\beta _{A})<0,\vspace{0.3cm}%
\end{array}
&  \\ 
\begin{array}{l}
t^{\ast }\mbox{\boldmath${\Upsilon}$}^{\top }\mathbf{D}^{-1}%
\mbox{\boldmath${\Upsilon}$}+(1-t^{\ast }/(t^{\ast }+u^{\ast }))%
\mbox{\boldmath${\Upsilon}$}^{\top }\mathbf{D}^{-1}\mbox{\boldmath${\Sigma}$}%
_{2}\mathbf{D}^{-1}\mbox{\boldmath${\Upsilon}$},\;\;\;\; \\ 
\;\ \ \ \ \ \ \ \ \ \ \ \ \ \ \ \ \ \ \ \ \ \ \ \ \ \text{if \ \ }\;E\log
(\alpha _{A}\epsilon _{0}^{2}+\beta _{A})>0,%
\end{array}
& 
\end{array}%
\right.  \label{s1}
\end{equation}%
\begin{equation}
\bar{s}_{2}^{2}=\left\{ 
\begin{array}{ll}
\mbox{\boldmath${\Delta}$}^{\top }\mathbf{D}^{-1}\mbox{\boldmath${\Delta}$}%
,\;\;\;\mbox{if}\;\;\;E\log (\alpha _{A}\epsilon _{0}^{2}+\beta _{A})<0,%
\vspace{0.3cm} &  \\ 
\mbox{\boldmath${\Upsilon}$}^{\top }\mathbf{D}^{-1}\mbox{\boldmath${%
\Upsilon}$},\;\;\;\mbox{if}\;\;\;E\log (\alpha _{A}\epsilon _{0}^{2}+\beta
_{A})>0. & 
\end{array}%
\right.  \label{s2}
\end{equation}%
\newline

%\medskip

Let ${\mathcal{N}}$ denote a standard normal random variable, and let 
\begin{equation*}
\lim_{m\rightarrow \infty }\frac{k^{\ast }}{{\mathcal{n}}}=\bar{u}.
\end{equation*}

\begin{theorem}
\label{thcons}We assume that $H_{A}$ of (\ref{alt}) and Assumptions \ref%
{aspos}--\ref{assh2} and \ref{asalt} hold, and that $E\log (\alpha
_{0}\epsilon _{0}^{2}+\beta _{0})\neq 0$. Then, for all $0\leq \eta <1$%
\newline
(i) If %\beq\label{thcons1}
\begin{equation}
\lim_{m\rightarrow \infty }\frac{k^{\ast }}{{\mathcal{n}}^{(1-\eta )/(2-\eta
)}}<\infty ,  \label{v-early}
\end{equation}%
then 
\begin{equation*}
\frac{\tau _{m}-k^{\ast }-v_{1,{\mathcal{n}}}}{v_{2,{\mathcal{n}}}}\overset{{%
\mathcal{D}}}{\rightarrow }\bar{s}_{1}{\mathcal{N}}(0,1).
\end{equation*}%
(ii) If %\beq\label{thcons2}
\begin{equation}
\lim_{m\rightarrow \infty }\frac{k^{\ast }}{{\mathcal{n}}^{(1-\eta )/(2-\eta
)}}=\infty ,  \label{m-early}
\end{equation}%
%
%
%
%
%
%
%
%
%
%
%
%
%\eeq
and $\bar{u}=0$, then 
\begin{equation*}
\frac{\tau _{m}-k^{\ast }-v_{3,{\mathcal{n}}}}{v_{4,{\mathcal{n}}}}\overset{{%
\mathcal{D}}}{\rightarrow }\bar{s}_{2}{\mathcal{N}}(0,1).
\end{equation*}%
(iii) If $0<\bar{u}<1$, then 
\begin{equation*}
\lim_{m\rightarrow \infty }P\left\{ \frac{\tau _{m}-k^{\ast }}{(k^{\ast
})^{1/2}}>x\right\} =P\left\{ \bar{u}^{1-\eta }\max \left( {\mathcal{A}}_{1},%
{\mathcal{A}}_{2}(x)\right) \leq {\mathcal{c}}\right\} .
\end{equation*}
\end{theorem}

In order to understand the practical implications of Theorem \ref{thcons},
note that (up to some positive and finite constant) 
\begin{equation*}
v_{1,{\mathcal{n}}}\approx \left( m^{1-\eta }\right) ^{1/(2-\eta )},\quad
\;\;v_{2,{\mathcal{n}}}\approx \left( m^{1-\eta }\right) ^{1/(4-2\eta
)}\approx v_{1,{\mathcal{n}}}^{1/2},
\end{equation*}%
\begin{equation*}
v_{3,{\mathcal{n}}}\approx \left[ {\mathcal{n}}^{1-\eta }(k^{\ast })^{\eta }%
\right] ^{1/2}\;\;\;\mbox{and}\;\;\;v_{4,{\mathcal{n}}}\approx (k^{\ast
})^{1/2}.
\end{equation*}%
The case (\ref{v-early}) corresponds to a \textquotedblleft very
early\textquotedblright\ break; in this case, Theorem \ref{thcons} states
that the expected delay is approximately $v_{1,{\mathcal{n}}}$, i.e. that it
is approximately equal to $\left( m^{1-\eta }\right) ^{1/(2-\eta )}$.
Clearly, as $\eta $ increases, $v_{1,{\mathcal{n}}}$ decreases; the
dispersion around the expected delay, measured by $v_{2,{\mathcal{n}}}$,
also decreases, indicating that the choice of $\eta $ plays a role in
determining the delay in detecting (very early) changepoints, and that
larger values of $\eta $ reduce such a delay. In the presence of an
\textquotedblleft early, but not so early\textquotedblright\ break -
corresponding to case \textit{(ii)} of the theorem, where recall that $%
k^{\ast }=o\left( {\mathcal{n}}\right) $, the expected delay $v_{3,{\mathcal{%
n}}}$ still decreases as $\eta $ increases, as long as $k^{\ast }={\mathcal{n%
}}^{\gamma }$, for $\gamma >1/(2-\eta )$, but the dispersion around the
expected delay - given by the standardization $v_{4,{\mathcal{n}}}$ - does
not depend on $\eta $. Finally, the case of a late(r) change is studied in
part \textit{(iii)} of the theorem: in such a case, $\eta $ - and therefore
the weight function in the definition of the detector - does not play any
role.

\subsubsection{Detection delays when $\protect\eta >1$\label{delay-2}}

We now investigate the asymptotic behaviour of the stopping time $\tilde{\tau%
}_{m}$ defined in (\ref{tau-renyi}) - that is, when the detector is a R\'{e}%
nyi type statistic with $\eta >1$. In such a case, the asymptotic behaviour
of the detection delay uses the same notation irrrespective of whether $%
E\log (\alpha _{A}\epsilon _{0}^{2}+\beta _{A})<0$ or $>0$.

Let 
\begin{equation*}
{\mathcal{a}}=\lim_{m\rightarrow \infty }\frac{k^{\ast }}{r}\in \lbrack
0,\infty ].
\end{equation*}%
We define two independent normal random vectors $\mathbf{N}_{1}$ and $%
\mathbf{N}_{2}$ such that $E\mathbf{N}_{1}=\mathbf{0},E\mathbf{N}_{2}=%
\mathbf{0},E\mathbf{N}_{1}\mathbf{N}_{1}^{\top }=\mathbf{D}$ and 
\begin{equation*}
E\mathbf{N}_{2}\mathbf{N}_{2}^{\top }=\left\{ 
\begin{array}{ll}
\mbox{\boldmath${\Sigma}$}_{1},\;\;\;\;\mbox{if}\;\;\;E\log (\alpha
_{A}\epsilon _{0}^{2}+\beta _{A})<0,\vspace{0.3cm} &  \\ 
\mbox{\boldmath${\Sigma}$}_{2},\;\;\;\;\mbox{if}\;\;\;E\log (\alpha
_{A}\epsilon _{0}^{2}+\beta _{A})>0. & 
\end{array}%
\right.
\end{equation*}%
Similarly to ${\mathcal{A}}_{1}$ and ${\mathcal{A}}_{2}(x)$ in (\ref{a1})
and (\ref{a2}), we define 
\begin{equation}
{\mathcal{B}}_{1}=\sup_{1\leq t\leq {\mathcal{a}}}\frac{1}{t^{\eta }}%
\mbox{\boldmath${\Gamma}$}(t)\mathbf{D}^{-1}\mbox{\boldmath${\Gamma}$}(t),
\label{b1}
\end{equation}%
\begin{equation}
{\mathcal{B}}_{2}(x)=\left\{ 
\begin{array}{ll}
\begin{array}{l}
\displaystyle\frac{1}{{\mathcal{a}}^{\eta }}\sup_{0<t\leq x}(t%
\mbox{\boldmath${\Delta}$}+\mbox{\boldmath${\Gamma}$}({\mathcal{a}}))^{\top }%
\mathbf{D}^{-1}(t\mbox{\boldmath${\Delta}$}+\mbox{\boldmath${\Gamma}$}({%
\mathcal{a}})),\; \\ 
\;\ \ \ \ \ \ \ \ \ \ \ \ \ \ \ \ \ \ \ \ \ \ \ \text{ \ \ }\;\mbox{if}%
\;\;\;E\log (\alpha _{A}\epsilon _{0}^{2}+\beta _{A})<0,\vspace{0.3cm}%
\end{array}
&  \\ 
\begin{array}{l}
\displaystyle\frac{1}{{\mathcal{a}}^{\eta }}\sup_{0<t\leq x}(t%
\mbox{\boldmath${\Upsilon}$}+\mbox{\boldmath${\Gamma}$}({\mathcal{a}}%
))^{\top }\mathbf{D}^{-1}(t\mbox{\boldmath${\Upsilon}$}+\mbox{\boldmath${%
\Gamma}$}({\mathcal{a}})),\; \\ 
\;\text{\ \ \ \ \ \ \ \ \ \ \ \ \ \ \ \ \ \ \ \ \ \ \ \ \ }\;\mbox{if}%
\;\;\;E\log (\alpha _{A}\epsilon _{0}^{2}+\beta _{A})>0,%
\end{array}
& 
\end{array}%
\right.  \label{b2}
\end{equation}%
and the centering sequence

\begin{equation*}
v_{5,{\mathcal{n}}}=\left\{ 
\begin{array}{ll}
\displaystyle\left( \displaystyle\frac{\displaystyle{\mathcal{c}}}{%
\displaystyle A_{m}}\displaystyle\frac{\displaystyle(k^{\ast })^{\eta }}{%
\displaystyle r^{1-\eta }}\right) ^{1/2},\;\;\;\mbox{if}\;\;\;E\log (\alpha
_{A}\epsilon _{0}^{2}+\beta _{A})<0,\vspace{0.3cm} &  \\ 
\displaystyle\left( \displaystyle\frac{\displaystyle{\mathcal{c}}}{B_{m}}%
\displaystyle\frac{\displaystyle(k^{\ast })^{\eta }}{\displaystyle r^{1-\eta
}}\right) ^{1/2},\;\;\;\mbox{if}\;\;\;E\log (\alpha _{A}\epsilon
_{0}^{2}+\beta _{A})>0, & 
\end{array}%
\right.
\end{equation*}%
and the rescaling sequence $v_{6,{\mathcal{n}}}=(k^{\ast })^{1/2}$.

\begin{theorem}
\label{altre}We assume that $H_{A}$ of (\ref{alt}) and Assumptions \ref%
{aspos}--\ref{assh2} and \ref{asalt} hold, and that $E\log (\alpha
_{0}\epsilon _{0}^{2}+\beta _{0})\neq 0$. Then, for all $1<\eta <2$.\newline
(i) If $k^{\ast }\leq r$ and ${\mathcal{a}}=0$ hold, then 
\begin{equation*}
\lim_{m\rightarrow \infty }P\{\bar{\tau}_{m}=r\}=1.
\end{equation*}%
(ii) If $k^{\ast }\leq r$ and ${\mathcal{a}}>0$ hold, then 
\begin{equation*}
\lim_{m\rightarrow \infty }P\{\bar{\tau}_{m}=r\}=P\left\{ ({\mathcal{a}}%
^{1/2}\mathbf{N}_{1}+{\mathcal{a}}\mbox{\boldmath${\Delta}$}+\mathbf{N}%
_{2})^{\top }\mathbf{D}^{-1}({\mathcal{a}}^{1/2}\mathbf{N}_{1}+{\mathcal{a}}%
\mbox{\boldmath${\Delta}$}+\mathbf{N}_{2})>{\mathcal{c}}\right\} .
\end{equation*}%
(iii) If $k^{\ast }>r$ and ${\mathcal{a}}<\infty $, then 
\begin{equation*}
\lim_{m\rightarrow \infty }P\left\{ \bar{\tau}_{m}>k^{\ast
}+xr^{1/2}\right\} =P\left\{ \max \left( {\mathcal{B}}_{1},{\mathcal{B}}%
_{2}(x)\right) \leq {\mathcal{c}}\right\} .
\end{equation*}%
(iv) If $k^{\ast }>r,\bar{u}<1$ and ${\mathcal{a}}=\infty $ hold, then 
\begin{equation*}
\frac{\bar{\tau}_{m}-k^{\ast }-v_{5,{\mathcal{n}}}}{v_{6,{\mathcal{n}}}}%
\overset{{\mathcal{D}}}{\rightarrow }\bar{s}_{2}{\mathcal{N}}.
\end{equation*}
\end{theorem}

Similarly to Theorem \ref{thcons}, Theorem \ref{altre} describes the
detection delay when using R\'{e}nyi type statistics depending on the
location of the break; to the best of our knowledge, this is the first time
such a result has ever been derived. Part \textit{(i)} of the theorem is
also derived in \citet{ghezzi2024fast}, and, in essence, it states that if
the break occurs prior to the trimming sequence $r$ in the R\'{e}nyi type
statistics, then it is identified straight at $r$ - that is, as soon as the R%
\'{e}nyi type statistics starts the monitoring. Parts \textit{(ii)} and 
\textit{(iii) }of the theorem refine and extend the results in %
\citet{ghezzi2024fast}. Finally, part \textit{(iv)} states that, in the case
of a break occurring late - or, better, much later than $r$ - a large value
of $\eta $ could even be detrimental because the centering sequence $v_{5,{%
\mathcal{n}}}$ diverges with $k^{\ast }$, at a faster rate as $\tau $
increases. This confirms the common wisdom (see \citealp{kirch2022asymptotic}
and \citealp{kirch2022sequential}), and the findings in %
\citet{ghezzi2024fast}, that R\'{e}nyi type statistics are designed for the
fast detection of very early occurring breaks, whereas they may yield
suboptimal results for later breaks.

\section{Simulations\label{simulations}}

In this section, we assess the finite sample performance of our monitoring
procedures via Monte Carlo simulations. According to the theory in Section %
\ref{asymptotics}, we can have two classes of monitoring schemes, based on $%
\eta \neq 1$ (covered by Theorem \ref{ma1}) and $\eta=1$ (covered by Theorem %
\ref{ma2}). For the sake of brevity, here we only focus on the case $\eta
\neq 1$. We consider several data generating processes (DGP). We use three
lengths of the historical training sample $m=500,1000,5000$ and two lengths
of the monitoring $\mathcal{n}=250,500$. The sequential procedure is
performed $5,000$ times with independently generated samples, and the
percentage of simulations for which the detector crosses the boundary
functions is reported for several values of $\eta $. For the R\'{e}nyi type
statistic based on Theorem \ref{ma1}(\textit{ii}), we follow \cite%
{horvath2021} and set $r=\sqrt{\mathcal{n}}$. Guidelines on implementation
are provided in Section \ref{sec:implementation} of the Supplement. To
obtain critical values, we simulate two independent standard Wiener
processes $W_{1}(t)$ and $W_{2}(t)$ on a grid of $100,000$ equally spaced
points in the unit interval $\left[ 0,1\right] $ and compute $\sup_{0<t\leq
1}\frac{1}{t^{\eta }}\left( W_{1}^{2}(t)+W_{2}^{2}(t)\right) $ and $%
\sup_{0<t\leq 1}\frac{1}{t^{1-\eta }}\left( W_{1}^{2}(t)+W_{2}^{2}(t)\right) 
$. We repeat this by $100,000$ times and obtain the empirical $90\%$, $95\%$%
, and $99\%$ percentiles of the above two quantities, corresponding to the
critical values at $10\%$, $5\%$, and $1\%$ levels based on Theorem \ref{ma1}%
(\textit{i}) and (\textit{ii}). Critical values are in Table \ref%
{tab:cv_i_ii}.

\begin{table}[tbph]
\caption{Critical values}
\label{tab:cv_i_ii}\centering
{\scriptsize {\ 
\begin{tabular}{llcccp{1cm}llccc}
\multicolumn{5}{c}{Based on Theorem 3.1 (\textit{i})} &  & 
\multicolumn{5}{c}{Based on Theorem 3.1 (\textit{ii})} \\ 
\cmidrule{1-5}\cmidrule{7-11} $\eta/\alpha$ &  & 10\% & 5\% & 1\% &  & $%
\eta/\alpha$ &  & 10\% & 5\% & 1\% \\ 
\cmidrule{1-5}\cmidrule{7-11} $\eta=0.0$ &  & 5.838 & 7.215 & 10.474 &  & $%
\eta=1.3$ &  & 5.609 & 7.024 & 10.235 \\ 
$\eta=0.3$ &  & 6.173 & 7.556 & 10.819 &  & $\eta=1.5$ &  & 5.516 & 6.909 & 
10.090 \\ 
$\eta=0.5$ &  & 6.537 & 7.934 & 11.188 &  & $\eta=1.7$ &  & 5.436 & 6.822 & 
10.014 \\ 
$\eta=0.7$ &  & 7.191 & 8.622 & 11.861 &  & $\eta=2.0$ &  & 5.340 & 6.715 & 
9.913 \\ 
\bottomrule \bottomrule &  &  &  &  &  &  &  &  &  & 
\end{tabular}
} }
\end{table}

The boundary functions in Section \ref{model} are designed for the case $%
m\rightarrow \infty $. However, preliminary simulations show that the
empirical sizes based on those boundary functions tend to be larger than the
nominal levels in finite samples, in particular for DGPs with the Student's $%
t$ errors. To make our monitoring schemes more practical under small finite
samples, we suggest to \textquotedblleft tune\textquotedblright\ the
boundary functions as 
\begin{align}
\mathfrak{g}_{m}(k)& =\mathcal{c}\mathcal{n}\left( 1+\dfrac{1}{\log (m)}%
\right) ^{2}\left( 1+\dfrac{k}{m}\right) ^{2}\left( \dfrac{k}{\mathcal{n}}%
\right) ^{\eta },\quad \;\; & \mbox{with}\;\;0& \leq \eta <1,
\label{eq:tune3} \\
\bar{\mathfrak{g}}_{m}(k)& =\mathcal{c}r\left( 1+\dfrac{1}{\log (m)}\right)
^{2}\left( 1+\dfrac{k}{m}\right) ^{2}\left( \frac{k}{r}\right) ^{\eta
},\quad \;\; & \mbox{with}\;\;\eta & >1,  \label{eq:tune4}
\end{align}%
for (\ref{boundary}) and (\ref{renyiboundaryw}) respectively. The intuition
underpinning the term $\left( 1+1/\log (m)\right) ^{2}$ is to boost the
boundary function in small samples. The term $\left( 1+k/m\right) ^{2}$ as
is typically employed when the monitoring horizon is \textquotedblleft
long\textquotedblright\ (see \citealp{horvath2020sequential}, %
\citealp{horvath2021monitoring}, and \citealp{zenhya}). Although this term
is inconsequential for the asymptotic theory in our set-up, we find that it
can further improve the empirical size. Both tuning terms are asymptotically
negligible. and only play a role in finite samples to achieve better size
control at no expense for power. The proposed tuning is tailored to DGPs
with Student's $t$ errors, rather than Gaussian errors; indeed, heavy tails
are a well-known stylised fact of financial returns.

\subsection{Empirical size under the null}

Under the null hypothesis, the realisation of GARCH(1,1) in the historical
training period ($1\leqslant i\leqslant m$) and in the monitoring period ($%
m+1\leqslant i\leqslant m+\mathcal{n}$) is 
\begin{equation*}
y_{i}=\sigma _{i}\epsilon _{i},\qquad \mbox{and}\qquad \sigma _{i}=\omega
_{0}+\alpha _{0}y_{i-1}^{2}+\beta _{0}\sigma _{i-1}^{2},\qquad \ 1\leq i\leq
m+\mathcal{n},
\end{equation*}%
where $\epsilon _{i}$ follows a standard normal distribution or the
Student's $t$ distribution with $7$ degrees of freedom. Since our monitoring
procedure does not require the historical sample to be stationary or not, we
choose the following two set of GARCH(1,1) parameters, taken from %
\citet{francq2012strict}: (\textit{i}) $(\omega _{0},\alpha _{0},\beta
_{0})=(0.10,0.18,0.80)$, which represents the stationary case; (\textit{ii}) 
$(\omega _{0},\alpha _{0},\beta _{0})=(0.10,0.30,0.80)$, corresponding to
the nonstationary case since $E\log (\alpha \epsilon _{0}^{2}+\beta )>0$
under errors following either the standard normal or the Student's $t$
distributions. \newline

Table \ref{tab:size_H0_5pct} reports the empirical sizes at $5\%$
significance level for the monitoring scheme based on Theorem \ref{ma1}(%
\textit{i}) for different values of $\eta $. A noticeable feature is that a
larger $\eta $ results in a higher rejection rates, and a smaller $\eta $ is
more conservative in rejection. Under the Student's $t$ errors, $\eta =0.3$
is a good choice because the monitoring procedure has reasonably good
empirical sizes when $m=1,000$, and the empirical sizes for $m=5,000$ are
closer to the theoretical level of $5\%$. Under Gaussian errors, the
monitoring procedure is slightly under-sized, which is mainly due to the
additionally tuning we imposed in \eqref{eq:tune3}. For practical use, the
tendency to under-reject with Gaussian errors may not necessarily be a
concern, because the empirical power does not seem to be affected, as shown
in Section \ref{sec:HA}. Lastly, the simulation results show that our
monitoring schemed works reasonably well for both stationary and
nonstationary GARCH(1,1) models.\newline

\begin{table}[htbp]
\caption{Empirical size based on Theorem 3.1(\textit{i})}
\label{tab:size_H0_5pct}\centering
{\scriptsize {\ }}
\par
{\scriptsize 
\begin{tabular}{llrrrp{1cm}rrr}
\toprule \toprule &  & \multicolumn{3}{c}{$\epsilon_i \sim \mathcal{N} (0,1)$%
} &  & \multicolumn{3}{c}{$\epsilon_i \sim $ Student's $t$} \\ 
&  & \multicolumn{1}{l}{$m= 500$} & \multicolumn{1}{l}{$m= 1000$} & 
\multicolumn{1}{l}{$m= 5000$} &  & \multicolumn{1}{l}{$m= 500$} & 
\multicolumn{1}{l}{$m= 1000$} & \multicolumn{1}{l}{$m= 5000$} \\ 
\midrule &  & \multicolumn{7}{c}{\underline{Stationary GARCH(1,1)}} \\ 
&  & \multicolumn{7}{c}{$\mathcal{n}=250$} \\ 
\cmidrule{3-9} $\eta=0.0$ &  & 6.5\% & 4.8\% & 3.5\% &  & 8.3\% & 8.0\% & 
6.2\% \\ 
$\eta=0.3$ &  & 7.1\% & 5.5\% & 3.8\% &  & 10.1\% & 8.9\% & 7.5\% \\ 
$\eta=0.5$ &  & 8.5\% & 6.2\% & 4.5\% &  & 11.3\% & 10.8\% & 9.1\% \\ 
$\eta=0.7$ &  & 10.1\% & 8.8\% & 6.6\% &  & 14.5\% & 13.2\% & 10.9\% \\ 
&  & \multicolumn{7}{c}{$\mathcal{n}=500$} \\ 
\cmidrule{3-9} $\eta=0.0$ &  & 5.2\% & 4.3\% & 2.7\% &  & 8.3\% & 5.4\% & 
4.6\% \\ 
$\eta=0.3$ &  & 5.8\% & 4.8\% & 2.9\% &  & 9.3\% & 6.1\% & 5.2\% \\ 
$\eta=0.5$ &  & 6.6\% & 5.4\% & 3.7\% &  & 11.4\% & 8.2\% & 6.9\% \\ 
$\eta=0.7$ &  & 8.7\% & 7.4\% & 5.6\% &  & 14.3\% & 10.5\% & 8.8\% \\ 
\midrule &  & \multicolumn{7}{c}{\underline{Nonstationary GARCH(1,1)}} \\ 
&  & \multicolumn{7}{c}{$\mathcal{n}=250$} \\ 
\cmidrule{3-9} $\eta=0.0$ &  & 5.4\% & 4.3\% & 3.7\% &  & 8.8\% & 5.2\% & 
5.0\% \\ 
$\eta=0.3$ &  & 6.2\% & 5.2\% & 4.2\% &  & 10.7\% & 6.5\% & 5.5\% \\ 
$\eta=0.5$ &  & 7.5\% & 6.7\% & 5.1\% &  & 12.3\% & 8.5\% & 6.8\% \\ 
$\eta=0.7$ &  & 9.9\% & 9.4\% & 5.9\% &  & 14.7\% & 12.2\% & 9.1\% \\ 
&  & \multicolumn{7}{c}{$\mathcal{n}=500$} \\ 
\cmidrule{3-9} $\eta=0.0$ &  & 2.6\% & 2.9\% & 2.9\% &  & 5.7\% & 4.4\% & 
4.0\% \\ 
$\eta=0.3$ &  & 3.8\% & 3.2\% & 3.0\% &  & 7.2\% & 5.3\% & 4.1\% \\ 
$\eta=0.5$ &  & 4.5\% & 3.9\% & 3.6\% &  & 8.7\% & 7.7\% & 5.3\% \\ 
$\eta=0.7$ &  & 7.4\% & 6.0\% & 5.1\% &  & 13.0\% & 9.6\% & 8.1\% \\ 
\bottomrule \bottomrule &  &  &  &  &  &  &  & 
\end{tabular}
}
\end{table}

Table \ref{tab:size_renyi_5pct} in the Supplement contains the empirical
sizes for R\'{e}nyi type statistics. The rejection rates are slightly higher
than the $5\%$ nominal level. We note that, in principle, it would be
possible to design a different tuning for R\'{e}nyi type statistics.

\subsection{Empirical power under $H_{A}$\label{sec:HA}}

We now turn to the analysis of the empirical power. Under the alternative,
the data is generated by 
\begin{equation*}
y_{i}=\sigma _{i}\epsilon _{i},
\end{equation*}%
and 
\begin{equation*}
\sigma _{i}=%
\begin{cases}
\omega _{0}+\alpha _{0}y_{i-1}^{2}+\beta _{0}\sigma _{i-1}^{2},\qquad  & %
\mbox{ if }1\leq i<m+k^{\ast }, \\ 
\omega _{A}+\alpha _{A}y_{i-1}^{2}+\beta _{A}\sigma _{i-1}^{2},\qquad  & %
\mbox{ if }m+k^{\ast }\leq i\leq m+\mathcal{n},%
\end{cases}%
\end{equation*}%
where the parameter $\boldmath{\theta }_{0}$$=(\alpha _{0},\beta _{0},\omega
_{0})^{\top }$ changes to $\boldmath{\theta }_{A}$$=(\alpha _{A},\beta
_{A},\omega _{A})^{\top }$ at time $m+k^{\ast }$. We consider two scenarios
for the time of change: (\textit{a}) $k^{\ast }=\lfloor \sqrt{\mathcal{n}}%
\rfloor $ corresponds to a change occurring \textquotedblleft early, but not
too early\textquotedblright\ after the historical sample; (\textit{b}) $%
k^{\ast }=0.5{\mathcal{n}}$ indicates a change happening much later than $r$%
. \newline

There are many possible ways of changes under the alternative. To keep our
results clean, we set $\omega _{0}=\omega _{A}=0.1$ and $\alpha _{0}=\alpha
_{A}=0.18$, and concentrate on a change in $\beta $ under the following four
representative alternatives:

\begin{itemize}
\item[$H_{A,1}$:] $\beta_0=0.8, \beta_A=0.6$, i.e. a change from a
stationary to another stationary regime,

\item[$H_{A,2}$:] $\beta_0=0.8, \beta_A=0.9$, i.e. a change from a
stationary to an explosive regime,

\item[$H_{A,3}$:] $\beta_0=0.9, \beta_A=0.8$, i.e. a change from an
explosive to a stationary regime,

\item[$H_{A,4}$:] $\beta_0=0.9, \beta_A=1.0$, i.e. a change from an
explosive to another explosive regime.\newline
\end{itemize}

Tables \ref{tab:power_standard_early} and \ref{tab:power_standard_middle}
show the empirical power of the monitoring scheme based on Theorem \ref{ma1}(%
\textit{i}) at $5\%$ significance level when $\mathcal{n}=500$ for a change
at $k^{\ast }=\lfloor \sqrt{\mathcal{n}}\rfloor $ and $k^{\ast }=\lfloor 0.5%
\mathcal{n}\rfloor $, respectively.\footnote{%
The empirical power of $\mathcal{n}=250$ (not reported) is marginally lower
than the empirical power of $\mathcal{n}=500$.} There are five major
observations. First, our monitoring scheme is highly effective in detecting
changes under $H_{A,2}$ and $H_{A,3}$ for both early and late changes. These
alternatives result in a change between a stationary regime and an explosive
regime, which is relatively easy to detect. Second, the monitoring scheme
exhibits high power in detecting early changes under $H_{A,1}$ and $H_{A,4}$%
. These alternatives represent a change within either a stationary or an
explosive regime. Third, there is a deterioration in power when detecting
late changes under $H_{A,1}$ and $H_{A,4}$, although satisfactory levels can
be achieved by using a large(r) training sample size of $m=5,000$. Fourth,
the power is relatively lower when using the Student's $t$ distribution
errors compared to normal errors. Lastly, there is only a marginal decline
observed in the power with a larger value of $\eta $.\newline

% Table generated by Excel2LaTeX from sheet 'Power - New'
\begin{table}[htbp]
\caption{Empirical power based on Theorem 3.1(\textit{i}) for a change at $%
k^*=\lfloor\protect\sqrt{\mathcal{n}}\rfloor$}
\label{tab:power_standard_early}\centering
{\scriptsize {\ 
\begin{tabular}{llrrrrrrr}
\toprule \toprule $\mathcal{n}=500$ &  & \multicolumn{3}{c}{$H_{A,1}$} &  & 
\multicolumn{3}{c}{$H_{A,2}$} \\ 
\cmidrule{3-5}\cmidrule{7-9} before $k^*=\lfloor\sqrt{\mathcal{n}}\rfloor$ & 
& \multicolumn{3}{c}{$\beta_0=0.80$} &  & \multicolumn{3}{c}{$\beta_0=0.80$}
\\ 
after $k^*=\lfloor\sqrt{\mathcal{n}}\rfloor$ &  & \multicolumn{3}{c}{$%
\beta_1=0.60$} &  & \multicolumn{3}{c}{$\beta_1=0.90$} \\ 
\midrule $\epsilon_i \sim \mathcal{N} (0,1)$ &  & \multicolumn{1}{c}{$m=500$}
& \multicolumn{1}{c}{$m=1000$} & \multicolumn{1}{c}{$m=5000$} &  & 
\multicolumn{1}{c}{$m=500$} & \multicolumn{1}{c}{$m=1000$} & 
\multicolumn{1}{c}{$m=5000$} \\ 
\midrule $\eta=0.0$ &  & 96.76\% & 99.94\% & 100.00\% &  & 99.76\% & 100.00\%
& 100.00\% \\ 
$\eta=0.3$ &  & 96.38\% & 99.94\% & 100.00\% &  & 99.76\% & 100.00\% & 
100.00\% \\ 
$\eta=0.5$ &  & 96.06\% & 99.94\% & 100.00\% &  & 99.74\% & 100.00\% & 
100.00\% \\ 
$\eta=0.7$ &  & 95.44\% & 99.84\% & 100.00\% &  & 99.72\% & 100.00\% & 
100.00\% \\ 
\midrule $\epsilon_i \sim $ Student's $t$ &  &  &  &  &  &  &  &  \\ 
\midrule $\eta=0.0$ &  & 80.26\% & 96.22\% & 99.92\% &  & 96.28\% & 99.44\%
& 100.00\% \\ 
$\eta=0.3$ &  & 79.16\% & 95.80\% & 99.92\% &  & 96.04\% & 99.36\% & 100.00\%
\\ 
$\eta=0.5$ &  & 78.20\% & 95.40\% & 99.92\% &  & 95.94\% & 99.34\% & 100.00\%
\\ 
$\eta=0.7$ &  & 76.88\% & 94.48\% & 99.92\% &  & 95.76\% & 99.18\% & 99.98\%
\\ 
\midrule \midrule $\mathcal{n}=500$ &  & \multicolumn{3}{c}{$H_{A,3}$} &  & 
\multicolumn{3}{c}{$H_{A,4}$} \\ 
\cmidrule{3-5}\cmidrule{7-9} before $k^*=\lfloor\sqrt{\mathcal{n}}\rfloor$ & 
& \multicolumn{3}{c}{$\beta_0=0.90$} &  & \multicolumn{3}{c}{$\beta_0=0.90$}
\\ 
after $k^*=\lfloor\sqrt{\mathcal{n}}\rfloor$ &  & \multicolumn{3}{c}{$%
\beta_1=0.80$} &  & \multicolumn{3}{c}{$\beta_1=1.00$} \\ 
\midrule $\epsilon_i \sim \mathcal{N} (0,1)$ &  & \multicolumn{1}{c}{$m=500$}
& \multicolumn{1}{c}{$m=1000$} & \multicolumn{1}{c}{$m=5000$} &  & 
\multicolumn{1}{c}{$m=500$} & \multicolumn{1}{c}{$m=1000$} & 
\multicolumn{1}{c}{$m=5000$} \\ 
\midrule $\eta=0.0$ &  & 100.00\% & 100.00\% & 100.00\% &  & 99.82\% & 
100.00\% & 100.00\% \\ 
$\eta=0.3$ &  & 100.00\% & 100.00\% & 100.00\% &  & 99.74\% & 100.00\% & 
100.00\% \\ 
$\eta=0.5$ &  & 100.00\% & 100.00\% & 100.00\% &  & 99.70\% & 100.00\% & 
100.00\% \\ 
$\eta=0.7$ &  & 100.00\% & 100.00\% & 100.00\% &  & 99.60\% & 100.00\% & 
100.00\% \\ 
\midrule $\epsilon_i \sim $ Student's $t$ &  &  &  &  &  &  &  &  \\ 
\midrule $\eta=0.0$ &  & 99.96\% & 100.00\% & 100.00\% &  & 94.00\% & 98.90\%
& 100.00\% \\ 
$\eta=0.3$ &  & 99.94\% & 100.00\% & 100.00\% &  & 93.76\% & 98.84\% & 
100.00\% \\ 
$\eta=0.5$ &  & 99.94\% & 100.00\% & 100.00\% &  & 93.40\% & 98.70\% & 
100.00\% \\ 
$\eta=0.7$ &  & 99.94\% & 100.00\% & 100.00\% &  & 92.98\% & 98.42\% & 
99.98\% \\ 
\bottomrule \bottomrule &  &  &  &  &  &  &  & 
\end{tabular}
} }
\end{table}

% Table generated by Excel2LaTeX from sheet 'Power - New'
\begin{table}[htbp]
\caption{Empirical power based on Theorem 3.1(\textit{i}) for a change at $%
k^*=\lfloor0.5\mathcal{n}\rfloor$}
\label{tab:power_standard_middle}\centering
{\scriptsize {\ 
\begin{tabular}{llrrrrrrr}
\toprule \toprule &  & \multicolumn{3}{c}{$H_{A,1}$} &  & \multicolumn{3}{c}{%
$H_{A,2}$} \\ 
\cmidrule{3-5}\cmidrule{7-9} before $k^*=\lfloor0.5\mathcal{n}\rfloor$ &  & 
\multicolumn{3}{c}{$\beta_0=0.80$} &  & \multicolumn{3}{c}{$\beta_0=0.80$}
\\ 
after $k^*=\lfloor0.5\mathcal{n}\rfloor$ &  & \multicolumn{3}{c}{$%
\beta_1=0.60$} &  & \multicolumn{3}{c}{$\beta_1=0.90$} \\ 
\midrule $\epsilon_i \sim \mathcal{N} (0,1)$ &  & \multicolumn{1}{c}{$m=500$}
& \multicolumn{1}{c}{$m=1000$} & \multicolumn{1}{c}{$m=5000$} &  & 
\multicolumn{1}{c}{$m=500$} & \multicolumn{1}{c}{$m=1000$} & 
\multicolumn{1}{c}{$m=5000$} \\ 
\midrule $\eta=0.0$ &  & 50.26\% & 77.88\% & 98.46\% &  & 99.76\% & 100.00\%
& 100.00\% \\ 
$\eta=0.3$ &  & 47.84\% & 76.16\% & 98.12\% &  & 99.76\% & 100.00\% & 
100.00\% \\ 
$\eta=0.5$ &  & 46.38\% & 74.32\% & 97.56\% &  & 99.74\% & 100.00\% & 
100.00\% \\ 
$\eta=0.7$ &  & 43.26\% & 70.56\% & 96.60\% &  & 99.72\% & 100.00\% & 
100.00\% \\ 
\midrule $\epsilon_i \sim $ Student's $t$ &  &  &  &  &  &  &  &  \\ 
\midrule $\eta=0.0$ &  & 25.68\% & 43.10\% & 72.88\% &  & 96.28\% & 99.44\%
& 100.00\% \\ 
$\eta=0.3$ &  & 25.02\% & 41.18\% & 70.48\% &  & 96.04\% & 99.36\% & 100.00\%
\\ 
$\eta=0.5$ &  & 24.80\% & 39.46\% & 68.18\% &  & 95.94\% & 99.34\% & 100.00\%
\\ 
$\eta=0.7$ &  & 25.18\% & 37.68\% & 64.14\% &  & 95.76\% & 99.18\% & 99.98\%
\\ 
\midrule \midrule &  & \multicolumn{3}{c}{$H_{A,3}$} &  & \multicolumn{3}{c}{%
$H_{A,4}$} \\ 
\cmidrule{3-5}\cmidrule{7-9} before $k^*=\lfloor0.5\mathcal{n}\rfloor$ &  & 
\multicolumn{3}{c}{$\beta_0=0.90$} &  & \multicolumn{3}{c}{$\beta_0=0.90$}
\\ 
after $k^*=\lfloor0.5\mathcal{n}\rfloor$ &  & \multicolumn{3}{c}{$%
\beta_1=0.80$} &  & \multicolumn{3}{c}{$\beta_1=1.00$} \\ 
\midrule $\epsilon_i \sim \mathcal{N} (0,1)$ &  & \multicolumn{1}{c}{$m=500$}
& \multicolumn{1}{c}{$m=1000$} & \multicolumn{1}{c}{$m=5000$} &  & 
\multicolumn{1}{c}{$m=500$} & \multicolumn{1}{c}{$m=1000$} & 
\multicolumn{1}{c}{$m=5000$} \\ 
\midrule $\eta=0.0$ &  & 100.00\% & 100.00\% & 100.00\% &  & 76.78\% & 
93.14\% & 99.90\% \\ 
$\eta=0.3$ &  & 100.00\% & 100.00\% & 100.00\% &  & 75.64\% & 92.36\% & 
99.88\% \\ 
$\eta=0.5$ &  & 100.00\% & 100.00\% & 100.00\% &  & 74.40\% & 91.58\% & 
99.86\% \\ 
$\eta=0.7$ &  & 100.00\% & 100.00\% & 100.00\% &  & 71.98\% & 89.88\% & 
99.66\% \\ 
\midrule $\epsilon_i \sim $ Student's $t$ &  &  &  &  &  &  &  &  \\ 
\midrule $\eta=0.0$ &  & 99.96\% & 100.00\% & 100.00\% &  & 64.16\% & 78.82\%
& 94.30\% \\ 
$\eta=0.3$ &  & 99.94\% & 100.00\% & 100.00\% &  & 63.22\% & 78.14\% & 
93.76\% \\ 
$\eta=0.5$ &  & 99.94\% & 100.00\% & 100.00\% &  & 62.38\% & 77.00\% & 
93.26\% \\ 
$\eta=0.7$ &  & 99.94\% & 100.00\% & 100.00\% &  & 61.28\% & 75.30\% & 
92.16\% \\ 
\bottomrule \bottomrule &  &  &  &  &  &  &  & 
\end{tabular}
} }
\end{table}

Tables \ref{tab:power_renyi_early} and \ref{tab:power_renyi_middle} in the
Supplement provide the empirical power for the R\'{e}nyi type statistics
based on Theorem \ref{ma1}(\textit{ii}) under the same setting. When
detecting early changes at $k^{\ast }=\lfloor \sqrt{\mathcal{n}}\rfloor $,
similar observations as above apply; the monitoring schemes with $\eta =1.3$
and $1.5$ proves to be effective. However, one noticeable difference is that
a larger value of $\eta $ is detrimental in the power. In particular, $\eta
=1.7$ and $2$ suffer a remarkable loss of power under $H_{A,1}$. As far as
late changes ($k^{\ast }=\lfloor 0.5\mathcal{n}\rfloor $) are concerned, the
R\'{e}nyi type statistics become much less effective, as predicted by the
theory. This is because R\'{e}nyi type statistics are devised for the fast
detection of very early changes, whilst being suboptimal for late changes. 
\newline

It is also worthwhile to examine the stopping time $\tau _{m}$ and $\bar{\tau%
}_{m}$ in order to investigate the detection delays of our monitoring
procedures. Figure \ref{fig:delay_m500} shows the boxplot of the detection
delays of $\tau _{m}$ and $\bar{\tau}_{m}$ for a change at $k^{\ast
}=\lfloor \sqrt{\mathcal{n}}\rfloor $ under $H_{A,2}$ when $m=500$, $%
\mathcal{n}=500$. For the monitoring procedure based on Theorem \ref{ma1}(%
\textit{i}), it is consistent with our theory that larger values of $\eta $
reduce the detection delay. Considering the R\'{e}nyi type statistics based
on Theorem \ref{ma1}(\textit{ii}), there is only a marginal difference in
using various values of $\eta $. Comparing the detection delay between the
monitoring procedures based on Theorem \ref{ma1}(\textit{i}) and (\textit{ii}%
), we can clearly see the merit of the R\'{e}nyi type statistics for the
fast detection of early changes, as evidenced by shorter detection delays.

\begin{figure}[tbp]
\centering
\includegraphics[width=0.45\linewidth]{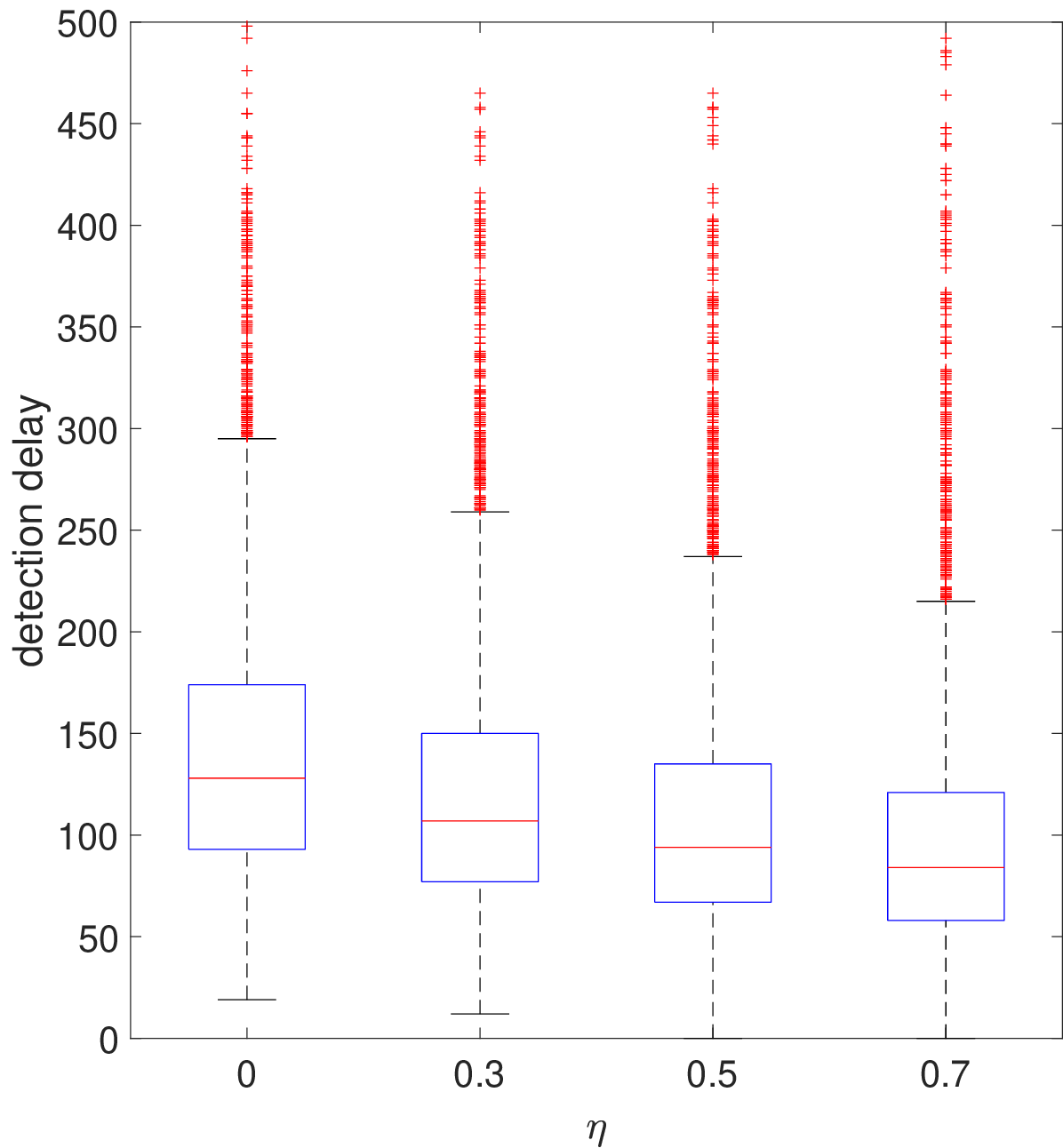} %
\includegraphics[width=0.45\linewidth]{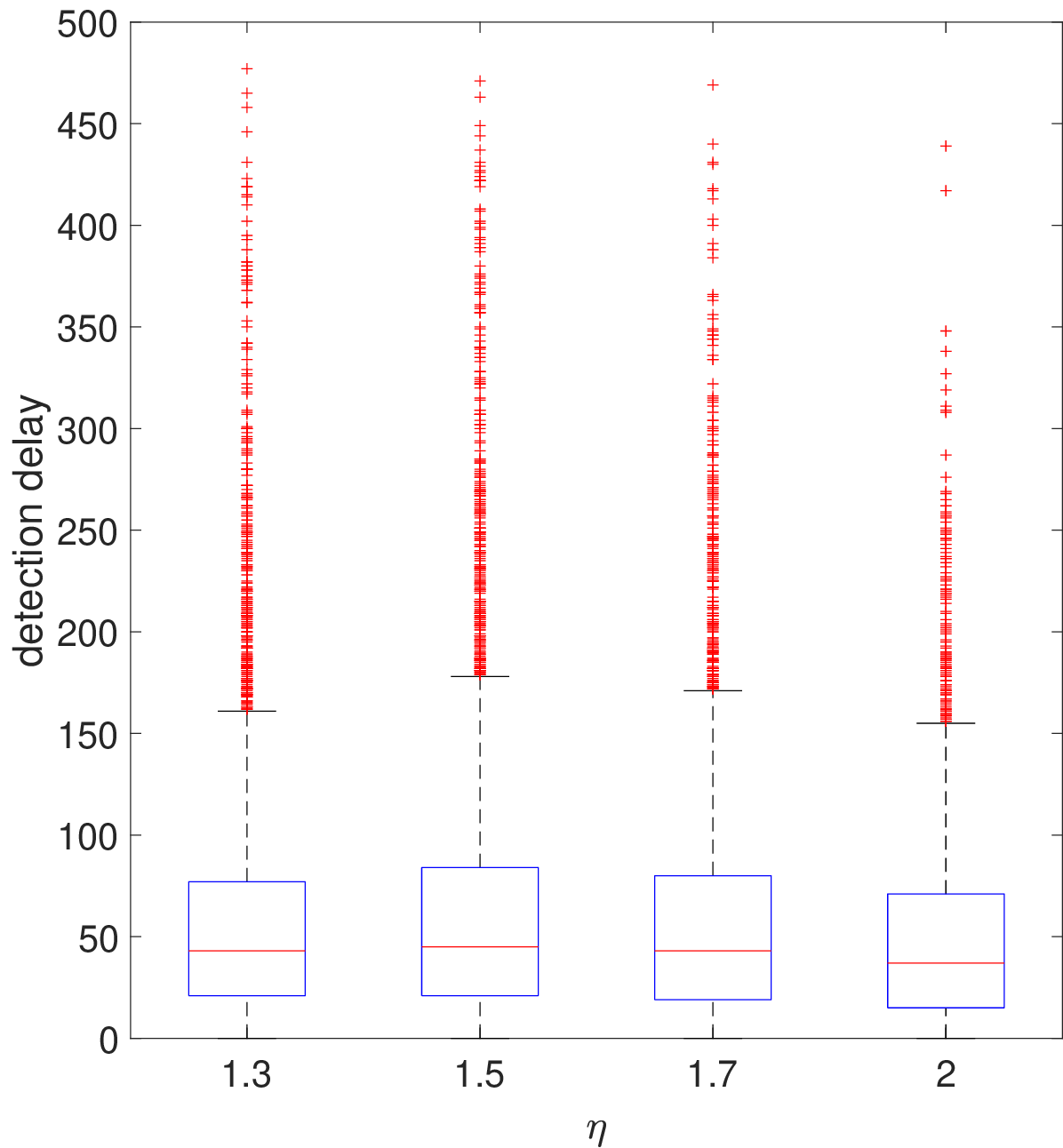}
\caption{Boxplot of detection Delays when $m=500$, $\mathcal{n}=500$ for a
change at $k^*=\lfloor\protect\sqrt{\mathcal{n}}\rfloor$ under $H_{A,2}$.
Left Panel: the monitoring procedure based on Theorem 3.1(\textit{i}); Right
Panel: the monitoring procedure based on Theorem 3.1(\textit{ii}).}
\label{fig:delay_m500}
\end{figure}

%\begin{figure}
%	\centering
%	\includegraphics[width=0.45\linewidth]{delay_i_m1000}
%	\includegraphics[width=0.45\linewidth]{delay_ii_m1000}
%	\caption{Boxplot of detection Delays when $m=1000$, $\mathcal{n}=500$ for a change at $k^*=\lfloor\sqrt{\mathcal{n}}\rfloor$. Left Panel: the monitoring procedure based on Theorem 3.1(\textit{i}); Right Panel: the monitoring procedure based on Theorem (\textit{ii}).}
%	\label{fig:delayim1000}
%\end{figure}

%\clearpage

\section{Empirical illustration\label{empirics}}

We illustrate our monitoring procedures using daily returns of individual
stocks. We focus on four stocks: Apple Inc. (ticker: AAPL, Permno: 14593),
Middlefield Banc Corp. (ticker: MBCN, Permno: 14932), Genetic Technologies
Ltd (ticker: GENE, Permno: 90899), and NTS Realty Holdings LP (ticker: NLP,
Permno: 90508). We download daily returns (without dividend) from the CRSP
database.\footnote{%
We choose to use the daily returns without dividend, rather than log
difference of prices, to avoid the complication due to stock splits.} We
consider two periods in order to showcase the detection for four types of
changes (in the sense of the four different alternatives in our simulation).
Depending on the specific purpose of the researcher, one can choose between
the monitoring procedures based on Theorem \ref{ma1}(\textit{i}) and (%
\textit{ii}). Based on our simulations,if the aim is to quickly detect very
early changes, we suggest using the R\'{e}nyi type statistics based on
Theorem \ref{ma1}\textit{(ii)}, with some tolerance for the compromise in
size and power; conversely, if the purpose is to have good size control and
high power, it is recommended to use the monitoring procedure based on
Theorem \ref{ma1}(\textit{i}). In this application, our preference is to
have a good balance of size and power, and the procedure based on Theorem %
\ref{ma1}(\textit{i}) (with the choice of $\eta =0.3$) delivers a good
performance with sample sizes similar to the dataset used in this section.%
\footnote{%
We relegate the results using R\'{e}nyi weights in Section \ref%
{renyiempirics} of the Supplement.} Before applying our monitoring
procedure, it is necessary to ensure there is no change during the
historical training sample. To this end, we use the test developed by %
\citeauthor{horvath2024detecting} (\citeyear{horvath2024detecting}, labeled
as HW(2024) hereinafter). Their test is to detect changes in GARCH(1,1)
processes without assuming stationarity, which can accommodate either
stationary or nonstationary historical sample. A rejection of their test
indicates there is no change of $(\alpha ,\beta )$ in the GARCH(1,1) during
the historical sample. At the same time, we are keen to understand which the
type of four changes may occur. Thus, we firstly examine whether our
historical sample is stationary or not by employing the nonstationarity test
developed by \citeauthor{francq2012strict} (\citeyear{francq2012strict},
labeled as FZ(2012) hereinafter). At the end of our monitoring horizon, we
use the FZ(2012) test again to check the stationarity of the samples after
the change (if there is one).

\subsection{Change from a stationary regime}

To illustrate changepoint detection from a stationary regime, we choose the
training period of 2016--2019 (1007 trading days) and the monitoring period
of 2020--2021 (507 trading days). The training period is before the outbreak
of COVID-19, while the monitoring period is in the pandemic. We apply our
monitoring procedure for the stocks of AAPL and MBCN during this period.
Table \ref{tab:empircal} (Columns 1 and 2) reports the results of the
sequential monitoring procedure, as well as other information, including
HW(2024) test, FZ(2012) test, and parameter estimates. HW(2024) test
indicates that there is no parameter change during the training sample for
AAPL and MBCN. The nonstationarity test of FZ(2012) indicates that they are
both stationary during the training sample. Our sequential monitoring
detects a change of AAPL on July 31st, 2020 and a change of MBCN on May 8th,
2020. Based on the FZ(2012) test for the sample after the change, we can
conclude that AAPL experienced a change from a stationary to another
stationary regime, while MBCN shifted from a stationary regime to a
nonstationary one. Figure \ref{fig:AAPL_MBCN} contains returns series (upper
panel) during the monitoring period and the detector versus the boundary
function (lower panel).

% Table generated by Excel2LaTeX from sheet 'Prototype Analysis'
\begin{table}[htbp]
\caption{Monitoring results of the four stocks}
\label{tab:empircal}\centering
{\scriptsize {\ 
\begin{tabular}{lllll}
\toprule \toprule & AAPL & MBCN & GENE & NLP \\ 
\multicolumn{5}{l}{\textbf{\underline{Training Sample}}} \\ 
Start Date & 2016-01-04 & 2016-01-04 & 2007-01-03 & 2007-01-03 \\ 
End Date & 2019-12-31 & 2019-12-31 & 2010-12-31 & 2010-12-31 \\ 
Sample Size & 1007 & 1007 & 1011 & 1011 \\ 
&  &  &  &  \\ 
\underline{HW(2024) Test} &  &  &  &  \\ 
Test stat & 1.469 & 1.351 & 1.179 & 1.088 \\ 
Rejection & Not Rej. & Not Rej. & Not Rej. & Not Rej. \\ 
&  &  &  &  \\ 
\underline{FZ(2012) NS Test} &  &  &  &  \\ 
$p$-value & 0.00\% & 0.00\% & 35.29\% & 33.81\% \\ 
Stationary or not & Stationary & Stationary & Nonstationary & Nonstationary
\\ 
&  &  &  &  \\ 
\underline{Parameter estimates} &  &  &  &  \\ 
$\hat{\alpha}_0$ & 0.135 & 0.183 & 0.287 & 0.099 \\ 
$\hat{\beta}_0$ & 0.745 & 0.579 & 0.816 & 0.916 \\ 
\multicolumn{5}{l}{\textbf{\underline{Monitoring Sample}}} \\ 
Start Date & 2020-01-02 & 2020-01-02 & 2011-01-03 & 2011-01-03 \\ 
End Date & 2021-12-31 & 2021-12-31 & 2012-12-31 & 2012-12-31 \\ 
Sample Size & 507 & 507 & 505 & 505 \\ 
&  &  &  &  \\ 
\underline{Our Sequential Monitoring} &  &  &  &  \\ 
Rejection & Rej. & Rej. & Rej. & Rej. \\ 
Time of Change & 2020-07-31 & 2020-05-08 & 2011-04-27 & 2012-09-04 \\ 
&  &  &  &  \\ 
\multicolumn{5}{l}{\textit{\textbf{After the change}}} \\ 
\underline{FZ(2012) NS Test} &  &  &  &  \\ 
$p$-value & 0.00\% & 10.38\% & 0.00\% & 100.00\% \\ 
Stationary or not & Stationary & Nonstationary & Stationary & Nonstationary
\\ 
&  &  &  &  \\ 
\underline{Parameter estimates} &  &  &  &  \\ 
$\hat{\alpha}_A$ & 0.052 & 0.091 & 0.488 & 0.001 \\ 
$\hat{\beta}_A$ & 0.935 & 0.913 & 0.528 & 1.053 \\ 
\bottomrule \bottomrule &  &  &  & 
\end{tabular}
} }
\end{table}

\begin{figure}[htbp]
\centering
\includegraphics[width=0.48\linewidth]{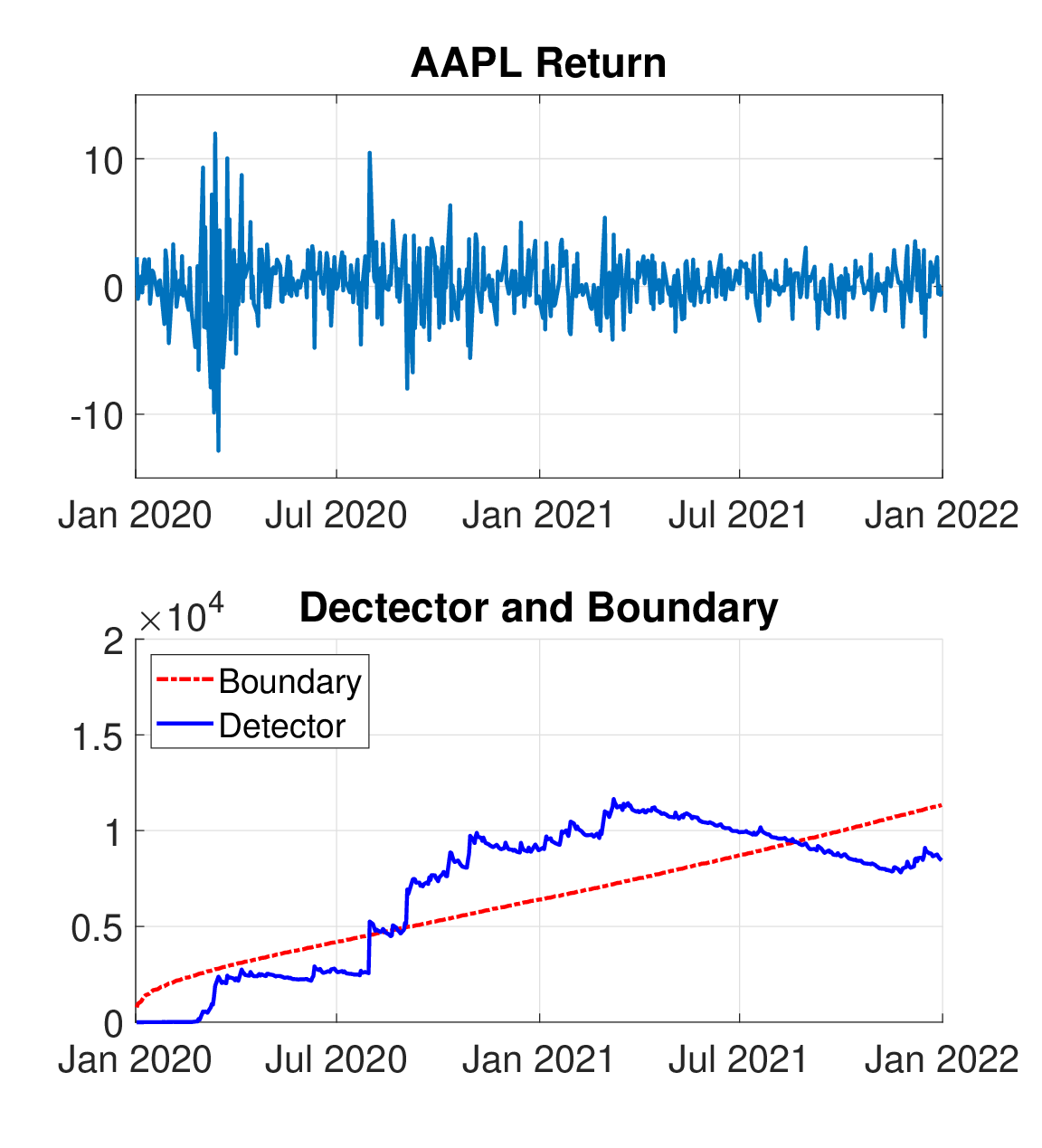} \includegraphics[width=0.48%
\linewidth]{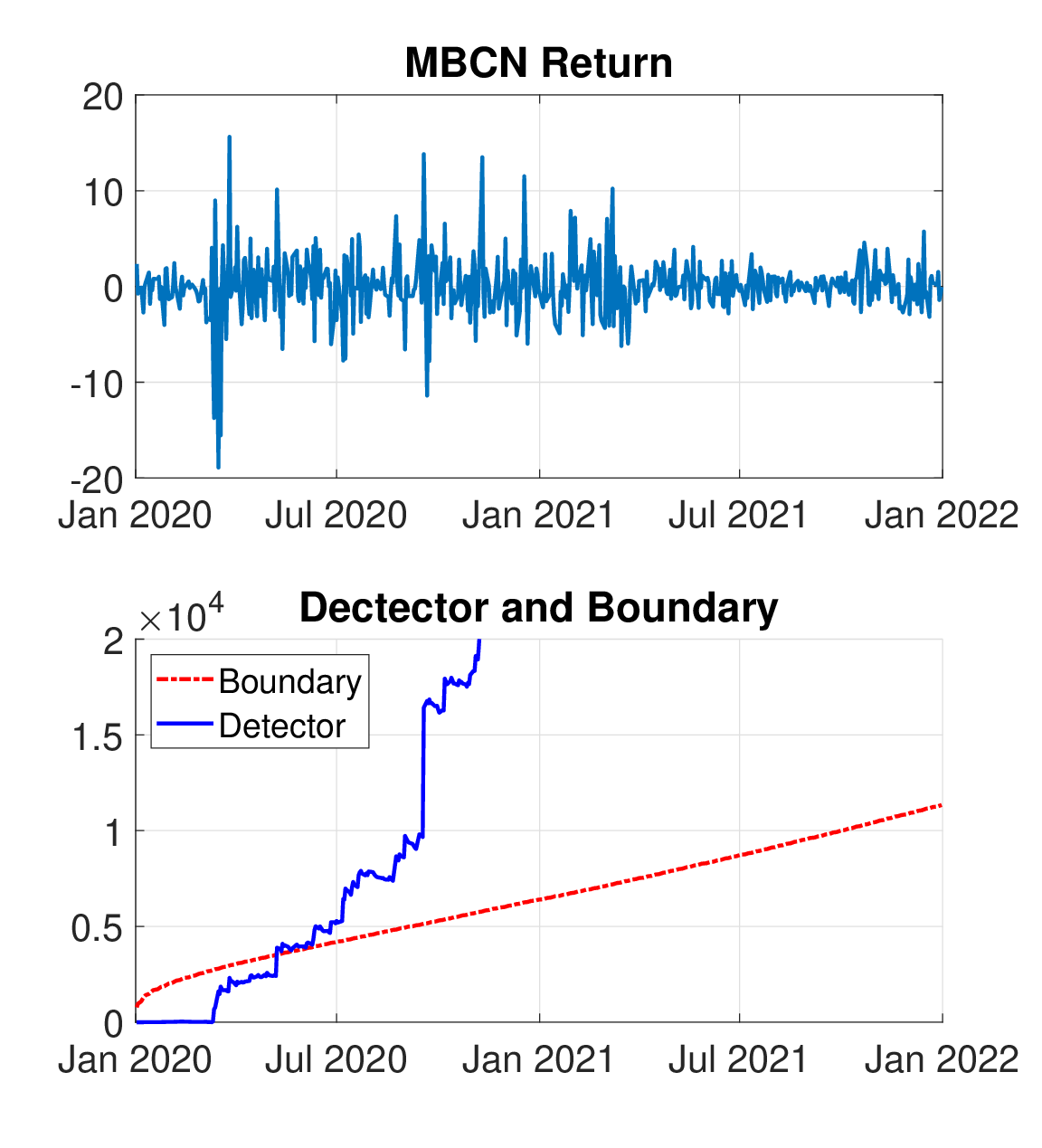}
\caption{Upper Panel: the return series of AAPL (left) and MBCN (right)
during the monitoring period; Lower Panel: the detector $\mathcal{D}_m(k)$
versus the boundary function $\mathfrak{g}_m(k)$ of AAPL (left) and MBCN
(right).}
\label{fig:AAPL_MBCN}
\end{figure}

\subsection{Change from a nonstationary regime}

We now consider detection from a nonstationary regime, and use 2007--2010
(1011 trading days) as the training period and 2011-2022 (505 trading days)
as the monitoring period. The training period covers the global financial
crisis (GFC), while the monitoring period follows the GFC but includes the
European debt crisis. In this period, we monitor GENE and NLP. The results
of the sequential monitoring procedure, alongside other supporting
information, are displayed in Columns 3 and 4 of Table \ref{tab:empircal}.
Based on the HW(2024) test, we cannot reject that the return series of GENE
and NLP have change in the training period. As evidenced by the
nonstationarity test of FZ(2012), both stocks are in the nonstationary
regime during the training period. Our sequential monitoring procedure
reveals a change of GENE on April 27th, 2011 and a change of NLP on
September 4th, 2012. After applying FZ(2012) nonstationarity test on the
sample after the change, it is found that the change of GENE is from a
nonstationary regime to a stationary regime, whilst the change of NLP is
from a nonstationary to another nonstationary regime. It is also interesting
to note that GENE after the change is in a strict stationary regime, but not
in a second-order stationary regime. Figure \ref{fig:GENE_NLP} shows their
returns series (upper panel) during the monitoring period and the detector
versus the boundary function (lower panel).

\begin{figure}[htbp]
\centering
\includegraphics[width=0.48\linewidth]{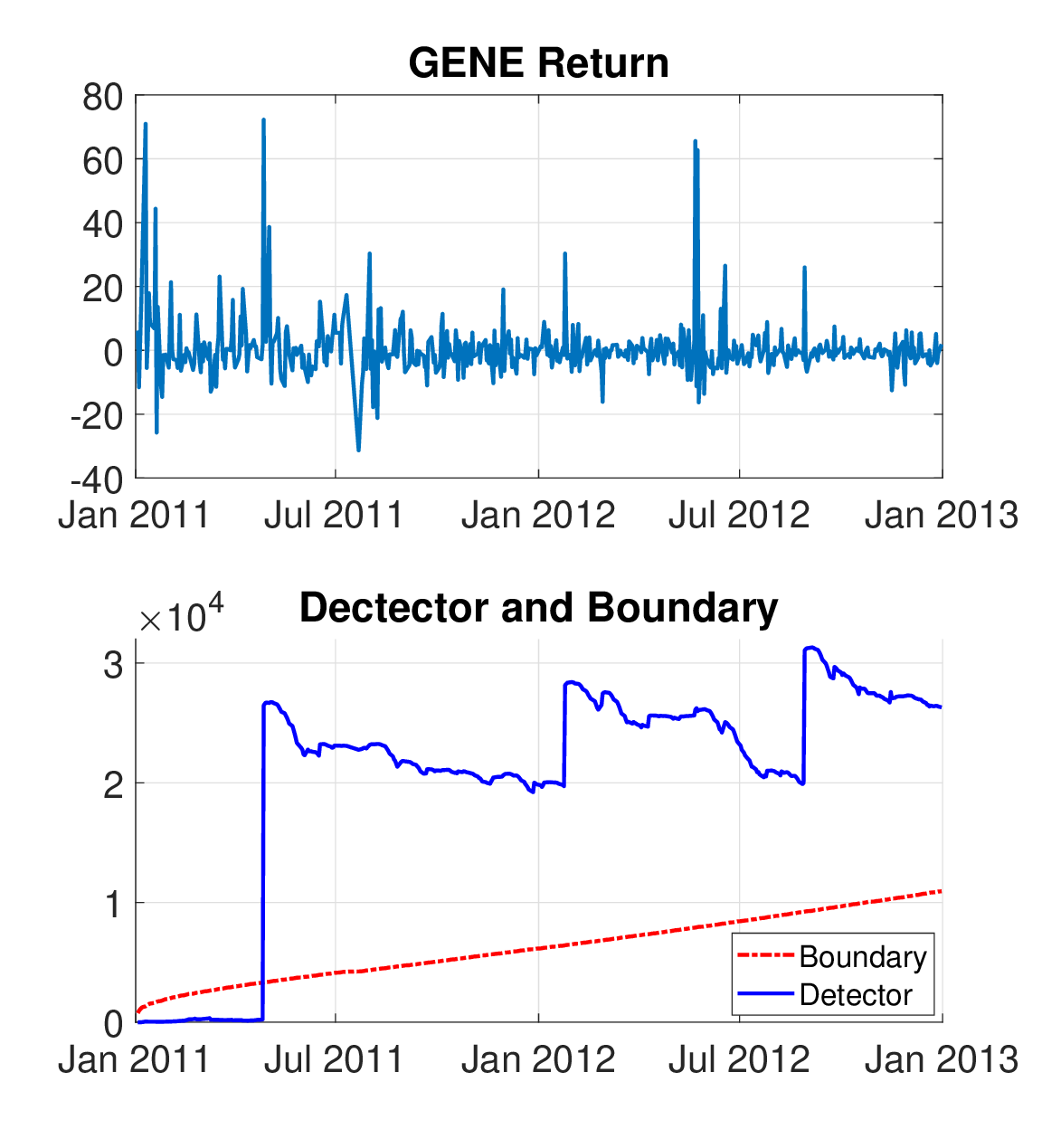} \includegraphics[width=0.48%
\linewidth]{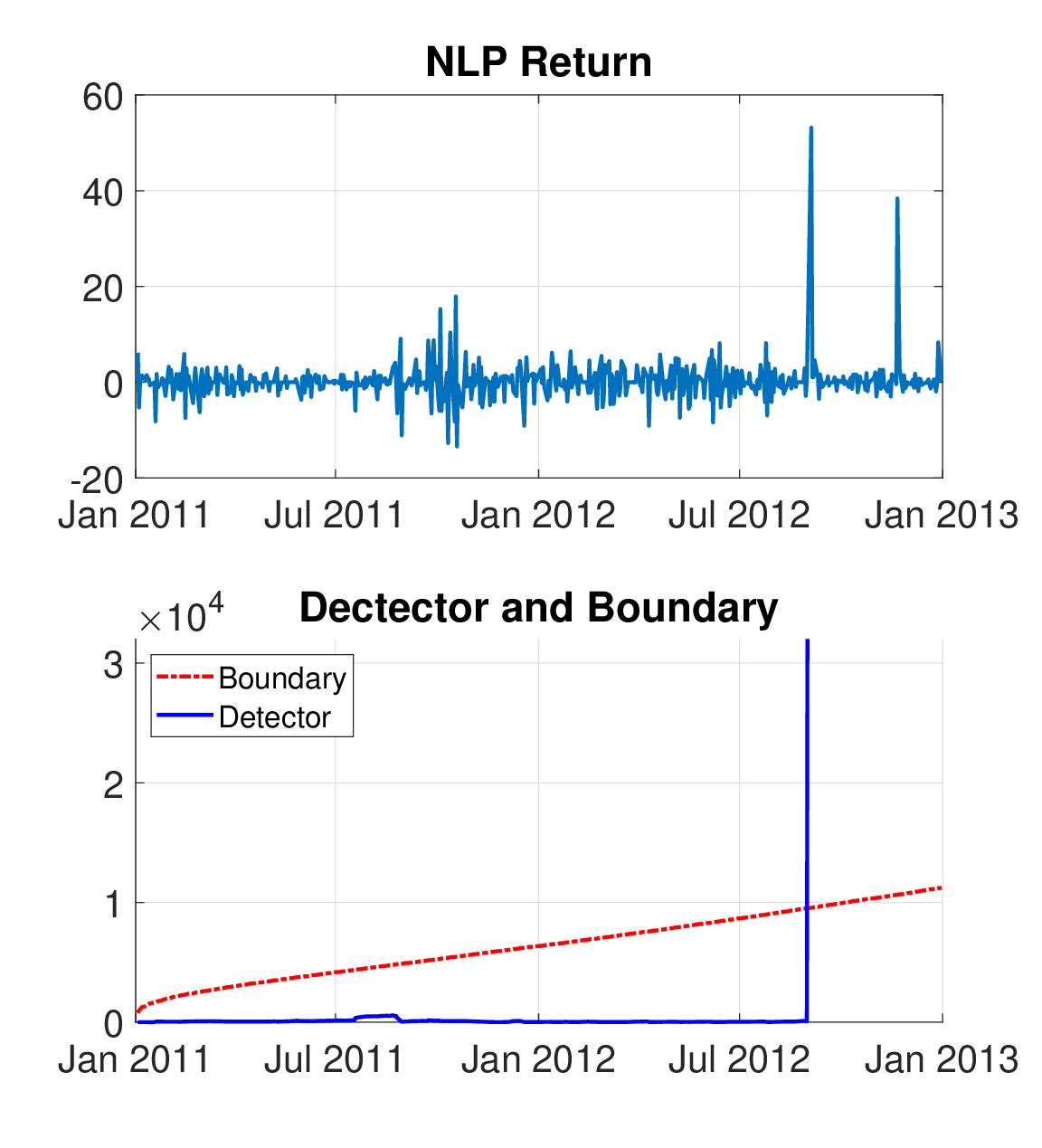}
\caption{Upper Panel: the return series of GENE (left) and NLP (right)
during the monitoring period; Lower Panel: the detector $\mathcal{D}_m(k)$
versus the boundary function $\mathfrak{g}_m(k)$ of GENE (left) and NLP
(right).}
\label{fig:GENE_NLP}
\end{figure}

\section{Conclusions and discussions \label{conclusion}}

In this paper, we complement the existing literature on (ex-ante) testing
for bubble phenomena by proposing a family of weighted, CUSUM-based
statistics to detect changes in the parameters of a GARCH(1,1) process. Our
monitoring procedure can be applied irrespective of whether, in the training
sample, the observations are stationary or explosive, and it is able to
detect all types of changes: (a) from a stationary to another stationary
regime (which is helpful in order to avoid the issues concerning the
consistent estimation of a GARCH(1,1) process spelt out in %
\citealp{hillebrand2005neglecting}); (b) from a stationary to an explosive
regime (which contains information on the possible inception of a bubble);
(c) from an explosive to a stationary regime (which, in the light of the
previous point, could shed light on the cooling off the turbulence
associated with a bubble on a financial market); and (d) from an explosive
to another explosive regime (which, depending on the direction of the change
- towards a more or a less explosive regime - could indicate whether
exuberant volatility is heating up or cooling down). Technically, we propose
two families of statistics, both based on weighted versions of the CUSUM
process of the quasi-Fisher scores: one family uses lighter weights, and it
is designed to detect, optimally, changes occurring not immediately after
the start of the monitoring horizon; the other family uses heavier, R\'{e}%
nyi-type weights, which make it more sensitive to changepoints occurring
immediately after the end of the training period. For both cases, we study
the limiting distribution of the detection delays; to the best of our
knowledge, no such results exist for the case of a GARCH(1,1) models, and no
results in general exist for the case of R\'{e}nyi statistics. Given the
interest in the detection of bubble phenomena, and the scant amount of
contributions in the context of detection of changes in the volatility, we
believe that our paper should be a useful addition to the toolbox of the
financial econometrician.

%\begin{adjustwidth}{-5pt}{-5pt}

{\footnotesize {\ 
\bibliographystyle{chicago}
\bibliography{LTbiblio}

\begin{thebibliography}{}

\bibitem[\protect\citeauthoryear{Aue, H{\"o}rmann, Horv{\'a}th, and
  Hu{\v{s}}kov{\'a}}{Aue et~al.}{2014}]{aue2014dependent}
Aue, A., S.~H{\"o}rmann, L.~Horv{\'a}th, and M.~Hu{\v{s}}kov{\'a} (2014).
\newblock Dependent functional linear models with applications to monitoring
  structural change.
\newblock {\em Statistica Sinica\/}, 1043--1073.

\bibitem[\protect\citeauthoryear{Aue and Horv{\'a}th}{Aue and
  Horv{\'a}th}{2004}]{aue2004delay}
Aue, A. and L.~Horv{\'a}th (2004).
\newblock Delay time in sequential detection of change.
\newblock {\em Statistics \& Probability Letters\/}~{\em 67\/}(3), 221--231.

\bibitem[\protect\citeauthoryear{Aue and Horv{\'a}th}{Aue and
  Horv{\'a}th}{2007}]{aue2007limit}
Aue, A. and L.~Horv{\'a}th (2007).
\newblock A limit theorem for mildly explosive autoregression with stable
  errors.
\newblock {\em Econometric Theory\/}~{\em 23\/}(2), 201--220.

\bibitem[\protect\citeauthoryear{Aue and Horv{\'a}th}{Aue and
  Horv{\'a}th}{2011}]{aue2011}
Aue, A. and L.~Horv{\'a}th (2011).
\newblock Quasi-likelihood estimation in stationary and nonstationary
  autoregressive models with random coefficients.
\newblock {\em Statistica Sinica\/}~{\em 21\/}(3), 973--999.

\bibitem[\protect\citeauthoryear{Aue, Horv{\'a}th, Kokoszka, and
  Steinebach}{Aue et~al.}{2008}]{aue2008monitoring}
Aue, A., L.~Horv{\'a}th, P.~Kokoszka, and J.~Steinebach (2008).
\newblock Monitoring shifts in mean: asymptotic normality of stopping times.
\newblock {\em Test\/}~{\em 17}, 515--530.

\bibitem[\protect\citeauthoryear{Berkes, Horv{\'a}th, and Kokoszka}{Berkes
  et~al.}{2003}]{berkes2003garch}
Berkes, I., L.~Horv{\'a}th, and P.~Kokoszka (2003).
\newblock {GARCH} processes: structure and estimation.
\newblock {\em Bernoulli\/}, 201--227.

\bibitem[\protect\citeauthoryear{Bloom}{Bloom}{2007}]{bloom2007uncertainty}
Bloom, N. (2007).
\newblock Uncertainty and the dynamics of {R}\&{D}.
\newblock {\em American Economic Review\/}~{\em 97\/}(2), 250--255.

\bibitem[\protect\citeauthoryear{Bougerol and Picard}{Bougerol and
  Picard}{1992}]{bougerol1992strict}
Bougerol, P. and N.~Picard (1992).
\newblock Strict stationarity of generalized autoregressive processes.
\newblock {\em The Annals of Probability\/}~{\em 20\/}(4), 1714--1730.

\bibitem[\protect\citeauthoryear{Breiman}{Breiman}{1968}]{breiman1968}
Breiman, L. (1968).
\newblock {\em Probability}.
\newblock Classics in Applied Mathematics. Society for Industrial and Applied
  Mathematics.

\bibitem[\protect\citeauthoryear{Chu, Stinchcombe, and White}{Chu
  et~al.}{1996}]{CSW96}
Chu, C., M.~Stinchcombe, and H.~White (1996).
\newblock Monitoring structural change.
\newblock {\em Econometrica\/}~{\em 64\/}(5), 1045--1066.

\bibitem[\protect\citeauthoryear{Cs{\"o}rg{\H{o}} and
  Horv{\'a}th}{Cs{\"o}rg{\H{o}} and Horv{\'a}th}{1997}]{csorgo1997}
Cs{\"o}rg{\H{o}}, M. and L.~Horv{\'a}th (1997).
\newblock {\em Limit {T}heorems in {C}hange-{P}oint {A}nalysis}, Volume~18.
\newblock John Wiley \& Sons.

\bibitem[\protect\citeauthoryear{Fiorentini, Calzolari, and
  Panattoni}{Fiorentini et~al.}{1996}]{fiorentini1996analytic}
Fiorentini, G., G.~Calzolari, and L.~Panattoni (1996).
\newblock Analytic derivatives and the computation of garch estimates.
\newblock {\em Journal of applied econometrics\/}~{\em 11\/}(4), 399--417.

\bibitem[\protect\citeauthoryear{Francq and Zakoian}{Francq and
  Zakoian}{2004}]{francq2004maximum}
Francq, C. and J.-M. Zakoian (2004).
\newblock Maximum likelihood estimation of pure {GARCH} and {ARMA}-{GARCH}
  processes.
\newblock {\em Bernoulli\/}~{\em 10\/}(4), 605--637.

\bibitem[\protect\citeauthoryear{Francq and Zako{\"\i}an}{Francq and
  Zako{\"\i}an}{2012}]{francq2012strict}
Francq, C. and J.-M. Zako{\"\i}an (2012).
\newblock Strict stationarity testing and estimation of explosive and
  stationary generalized autoregressive conditional heteroscedasticity models.
\newblock {\em Econometrica\/}~{\em 80\/}(2), 821--861.

\bibitem[\protect\citeauthoryear{Francq and Zakoian}{Francq and
  Zakoian}{2019}]{francq2019garch}
Francq, C. and J.-M. Zakoian (2019).
\newblock {\em GARCH models: structure, statistical inference and financial
  applications}.
\newblock John Wiley \& Sons.

\bibitem[\protect\citeauthoryear{Ghezzi, Rossi, and Trapani}{Ghezzi
  et~al.}{2024}]{ghezzi2024fast}
Ghezzi, F., E.~Rossi, and L.~Trapani (2024).
\newblock Fast online changepoint detection.
\newblock {\em arXiv preprint arXiv:2402.04433\/}.

\bibitem[\protect\citeauthoryear{Hillebrand}{Hillebrand}{2005}]{hillebrand2005neglecting}
Hillebrand, E. (2005).
\newblock Neglecting parameter changes in {GARCH} models.
\newblock {\em Journal of Econometrics\/}~{\em 129\/}(1-2), 121--138.

\bibitem[\protect\citeauthoryear{Homm and Breitung}{Homm and
  Breitung}{2012}]{homm2012testing}
Homm, U. and J.~Breitung (2012).
\newblock Testing for speculative bubbles in stock markets: a comparison of
  alternative methods.
\newblock {\em Journal of Financial Econometrics\/}~{\em 10\/}(1), 198--231.

\bibitem[\protect\citeauthoryear{Horv\'{a}th, Hu\v{s}kov\'{a}, Kokoszka, and
  Steinebach}{Horv\'{a}th et~al.}{2004}]{lajos04}
Horv\'{a}th, L., M.~Hu\v{s}kov\'{a}, P.~Kokoszka, and J.~Steinebach (2004).
\newblock Monitoring changes in linear models.
\newblock {\em Journal of Statistical Planning and Inference\/}~{\em 126\/}(1),
  225--251.

\bibitem[\protect\citeauthoryear{Horv\'{a}th, Kokoszka, and
  Steinebach}{Horv\'{a}th et~al.}{2007}]{lajos07}
Horv\'{a}th, L., P.~Kokoszka, and J.~Steinebach (2007).
\newblock On sequential detection of parameter changes in linear regression.
\newblock {\em Statistics and Probability Letters\/}~{\em 80}, 1806--1813.

\bibitem[\protect\citeauthoryear{Horv{\'a}th, Kokoszka, and Wang}{Horv{\'a}th
  et~al.}{2021}]{horvath2021monitoring}
Horv{\'a}th, L., P.~Kokoszka, and S.~Wang (2021).
\newblock Monitoring for a change point in a sequence of distributions.
\newblock {\em The Annals of Statistics\/}~{\em 49\/}(4), 2271--2291.

\bibitem[\protect\citeauthoryear{Horv{\'a}th, Liu, and Lu}{Horv{\'a}th
  et~al.}{2022}]{zenhya}
Horv{\'a}th, L., Z.~Liu, and S.~Lu (2022).
\newblock Sequential monitoring of changes in dynamic linear models, applied to
  the {US} housing market.
\newblock {\em Econometric Theory\/}~{\em 38\/}(2), 209--272.

\bibitem[\protect\citeauthoryear{Horv{\'a}th, Liu, Rice, and Wang}{Horv{\'a}th
  et~al.}{2020}]{horvath2020sequential}
Horv{\'a}th, L., Z.~Liu, G.~Rice, and S.~Wang (2020).
\newblock Sequential monitoring for changes from stationarity to mild
  non-stationarity.
\newblock {\em Journal of Econometrics\/}~{\em 215\/}(1), 209--238.

\bibitem[\protect\citeauthoryear{Horv{\'a}th, Miller, and Rice}{Horv{\'a}th
  et~al.}{2020}]{horvathmiller}
Horv{\'a}th, L., C.~Miller, and G.~Rice (2020).
\newblock A new class of change point test statistics of r{\'e}nyi type.
\newblock {\em Journal of Business \& Economic Statistics\/}~{\em 38\/}(3),
  570--579.

\bibitem[\protect\citeauthoryear{Horv{\'a}th, Miller, and Rice}{Horv{\'a}th
  et~al.}{2021}]{horvath2021}
Horv{\'a}th, L., C.~Miller, and G.~Rice (2021).
\newblock Detecting early or late changes in linear models with heteroscedastic
  errors.
\newblock {\em Scandinavian Journal of Statistics\/}~{\em 48\/}(2), 577--609.

\bibitem[\protect\citeauthoryear{Horv{\'a}th and Trapani}{Horv{\'a}th and
  Trapani}{2019}]{HT2019}
Horv{\'a}th, L. and L.~Trapani (2019).
\newblock Testing for randomness in a random coefficient autoregression model.
\newblock {\em Journal of Econometrics\/}~{\em 209\/}(2), 338--352.

\bibitem[\protect\citeauthoryear{Horv{\'a}th and Trapani}{Horv{\'a}th and
  Trapani}{2023}]{horvath2023real}
Horv{\'a}th, L. and L.~Trapani (2023).
\newblock Real-time monitoring with {RCA} models.
\newblock {\em arXiv preprint arXiv:2312.11710\/}.

\bibitem[\protect\citeauthoryear{Horv{\'a}th and Wang}{Horv{\'a}th and
  Wang}{2024}]{horvath2024detecting}
Horv{\'a}th, L. and S.~Wang (2024).
\newblock Detecting changes in {GARCH}(1, 1) processes without assuming
  stationarity.
\newblock {\em Available at SSRN 4712255\/}.

\bibitem[\protect\citeauthoryear{Jarrow and Kwok}{Jarrow and
  Kwok}{2023}]{jarrow2023explosion}
Jarrow, R.~A. and S.~S. Kwok (2023).
\newblock An explosion time characterization of asset price bubbles.
\newblock {\em International Review of Finance\/}~{\em 23\/}(2), 469--479.

\bibitem[\protect\citeauthoryear{Jensen and Rahbek}{Jensen and
  Rahbek}{2004}]{jensen2004asymptotic}
Jensen, S.~T. and A.~Rahbek (2004).
\newblock Asymptotic inference for nonstationary {GARCH}.
\newblock {\em Econometric Theory\/}~{\em 20\/}(6), 1203--1226.

\bibitem[\protect\citeauthoryear{Jurado, Ludvigson, and Ng}{Jurado
  et~al.}{2015}]{jurado2015measuring}
Jurado, K., S.~C. Ludvigson, and S.~Ng (2015).
\newblock Measuring uncertainty.
\newblock {\em American Economic Review\/}~{\em 105\/}(3), 1177--1216.

\bibitem[\protect\citeauthoryear{Kirch and Stoehr}{Kirch and
  Stoehr}{2022a}]{kirch2022asymptotic}
Kirch, C. and C.~Stoehr (2022a).
\newblock Asymptotic delay times of sequential tests based on {U}-statistics
  for early and late change points.
\newblock {\em Journal of Statistical Planning and Inference\/}~{\em 221},
  114--135.

\bibitem[\protect\citeauthoryear{Kirch and Stoehr}{Kirch and
  Stoehr}{2022b}]{kirch2022sequential}
Kirch, C. and C.~Stoehr (2022b).
\newblock Sequential change point tests based on {U}-statistics.
\newblock {\em Scandinavian Journal of Statistics\/}~{\em 49\/}(3), 1184--1214.

\bibitem[\protect\citeauthoryear{Kl{\"u}ppelberg, Lindner, and
  Maller}{Kl{\"u}ppelberg et~al.}{2004}]{kluppelberg2004continuous}
Kl{\"u}ppelberg, C., A.~Lindner, and R.~Maller (2004).
\newblock A continuous-time {GARCH} process driven by a l{\'e}vy process:
  stationarity and second-order behaviour.
\newblock {\em Journal of Applied Probability\/}~{\em 41\/}(3), 601--622.

\bibitem[\protect\citeauthoryear{Nelson}{Nelson}{1991}]{nelson1991conditional}
Nelson, D.~B. (1991).
\newblock Conditional heteroskedasticity in asset returns: a new approach.
\newblock {\em Econometrica\/}, 347--370.

\bibitem[\protect\citeauthoryear{Pelagatti and Lisi}{Pelagatti and
  Lisi}{2009}]{pelagatti2009variance}
Pelagatti, M. and F.~Lisi (2009).
\newblock Variance initialisation in garch estimation.
\newblock In {\em S. Co. 2009. Sixth Conference. Complex Data Modeling and
  Computationally Intensive Statistical Methods for Estimation and Prediction},
  pp.\  323. Maggioli Editore.

\bibitem[\protect\citeauthoryear{Phillips and Magdalinos}{Phillips and
  Magdalinos}{2007}]{phillips2007limit}
Phillips, P.~C. and T.~Magdalinos (2007).
\newblock Limit theory for moderate deviations from a unit root.
\newblock {\em Journal of Econometrics\/}~{\em 136\/}(1), 115--130.

\bibitem[\protect\citeauthoryear{Phillips, Shi, and Yu}{Phillips
  et~al.}{2015a}]{phillips2015testing2}
Phillips, P.~C., S.~Shi, and J.~Yu (2015a).
\newblock Testing for multiple bubbles: historical episodes of exuberance and
  collapse in the {S}\&{P} 500.
\newblock {\em International Economic Review\/}~{\em 56\/}(4), 1043--1078.

\bibitem[\protect\citeauthoryear{Phillips, Shi, and Yu}{Phillips
  et~al.}{2015b}]{phillips2015testing}
Phillips, P.~C., S.~Shi, and J.~Yu (2015b).
\newblock Testing for multiple bubbles: limit theory of real-time detectors.
\newblock {\em International Economic Review\/}~{\em 56\/}(4), 1079--1134.

\bibitem[\protect\citeauthoryear{Phillips and Shi}{Phillips and
  Shi}{2018}]{phillips2018financial}
Phillips, P.~C. and S.-P. Shi (2018).
\newblock Financial bubble implosion and reverse regression.
\newblock {\em Econometric Theory\/}~{\em 34\/}(4), 705--753.

\bibitem[\protect\citeauthoryear{Phillips, Wu, and Yu}{Phillips
  et~al.}{2011}]{phillips2011}
Phillips, P.~C., Y.~Wu, and J.~Yu (2011).
\newblock Explosive behavior in the 1990s {Nasdaq}: when did exuberance
  escalate asset values?
\newblock {\em International Economic Review\/}~{\em 52\/}(1), 201--226.

\bibitem[\protect\citeauthoryear{Phillips and Yu}{Phillips and
  Yu}{2011}]{phillips2011dating}
Phillips, P.~C. and J.~Yu (2011).
\newblock Dating the timeline of financial bubbles during the subprime crisis.
\newblock {\em Quantitative Economics\/}~{\em 2\/}(3), 455--491.

\bibitem[\protect\citeauthoryear{Richter, Wang, and Wu}{Richter
  et~al.}{2023}]{richter2023testing}
Richter, S., W.~Wang, and W.~B. Wu (2023).
\newblock Testing for parameter change epochs in {GARCH} time series.
\newblock {\em The Econometrics Journal\/}~{\em 26\/}(3), 467--491.

\bibitem[\protect\citeauthoryear{Skrobotov}{Skrobotov}{2023}]{skrobotov2023testing}
Skrobotov, A. (2023).
\newblock Testing for explosive bubbles: a review.
\newblock {\em Dependence Modeling\/}~{\em 11\/}(1), 20220152.

\bibitem[\protect\citeauthoryear{Sornette, Cauwels, and Smilyanov}{Sornette
  et~al.}{2018}]{sornette2018can}
Sornette, D., P.~Cauwels, and G.~Smilyanov (2018).
\newblock Can we use volatility to diagnose financial bubbles? {L}essons from
  40 historical bubbles.
\newblock {\em Quantitative Finance and Economics\/}~{\em 2\/}(1), 486--594.

\bibitem[\protect\citeauthoryear{Whitehouse, Harvey, and Leybourne}{Whitehouse
  et~al.}{2023}]{whitehouse2023real}
Whitehouse, E.~J., D.~I. Harvey, and S.~J. Leybourne (2023).
\newblock Real-time monitoring of bubbles and crashes.
\newblock {\em Oxford Bulletin of Economics and Statistics\/}~{\em 85\/}(3),
  482--513.

\end{thebibliography}
} }

\newpage

\clearpage

\renewcommand*{\thesection}{\Alph{section}}

\setcounter{section}{0} \setcounter{subsection}{-1} %
\setcounter{subsubsection}{-1} \setcounter{equation}{0} \setcounter{lemma}{0}
\setcounter{theorem}{0} \renewcommand{\theassumption}{A.\arabic{assumption}} 
\renewcommand{\thetheorem}{A.\arabic{theorem}} \renewcommand{\thelemma}{A.%
\arabic{lemma}} \renewcommand{\thecorollary}{A.\arabic{corollary}} %
\renewcommand{\theequation}{A.\arabic{equation}}

\section{Implementation guidelines, further Monte Carlo evidence, and
further empirical results\label{sec:further}}

\subsection{Practical guidance on the implementation\label%
{sec:implementation}}

In this section, we provide a detailed step-by-step guidance on implementing
our sequential monitoring procedures (based on Theorem \ref{ma1}), which
could be useful for researchers who want to apply it without studying the
underlying theory. The steps are as follows:

\begin{enumerate}
\item Use QMLE to estimate the parameters $\boldsymbol{\theta }$ using the
training sample 
\begin{equation*}
\hat{\boldsymbol{\theta }}_{m}=\mbox{\rm argmax}\left\{ \sum_{i=1}^{m}\bar{%
\ell}_{i}(\boldsymbol{\theta }):\;\boldsymbol{\theta }\in \boldsymbol{\Theta 
}\right\} ,
\end{equation*}%
where 
\begin{equation*}
\bar{\ell}_{i}(\boldsymbol{\theta })=\log \bar{\sigma}_{i}^{2}+\dfrac{%
y_{i}^{2}}{\bar{\sigma}_{i}^{2}},
\end{equation*}%
and 
\begin{equation*}
\bar{\sigma}_{i}^{2}(\boldsymbol{\theta })=\omega +\alpha y_{i-1}+\beta \bar{%
\sigma}_{i-1}^{2}(\boldsymbol{\theta }),\;\;\;\;\;1\leq i\leq m.
\end{equation*}

\item Calculate the quasi-Fisher scores of the log likelihood function 
\begin{equation*}
\frac{\partial \bar{\ell}_{i}(\hat{\boldsymbol{\theta }}_{m})}{\partial
\alpha }=\frac{\partial \bar{\ell}_{i}(\hat{\boldsymbol{\theta }}_{m})}{%
\partial \bar{\sigma}_{i}^{2}(\hat{\boldsymbol{\theta }}_{m})}\frac{\partial 
\bar{\sigma}_{i}^{2}(\hat{\boldsymbol{\theta }}_{m})}{\partial \alpha }=%
\frac{1}{\bar{\sigma}_{i}^{2}(\hat{\boldsymbol{\theta }}_{m})}\left( 1-\frac{%
y_{i}^{2}}{\bar{\sigma}_{i}^{2}(\hat{\boldsymbol{\theta }}_{m})}\right) 
\frac{\partial \bar{\sigma}_{i}^{2}(\hat{\boldsymbol{\theta }}_{m})}{%
\partial \alpha },\;\;\;\;\;1\leq i\leq m+\mathcal{n},
\end{equation*}

\begin{equation*}
\frac{\partial \bar{\ell}_{i}(\hat{\boldsymbol{\theta }}_{m})}{\partial
\beta }=\frac{\partial \bar{\ell}_{i}(\hat{\boldsymbol{\theta }}_{m})}{%
\partial \bar{\sigma}_{i}^{2}(\hat{\boldsymbol{\theta }}_{m})}\frac{\partial 
\bar{\sigma}_{i}^{2}(\hat{\boldsymbol{\theta }}_{m})}{\partial \beta }=\frac{%
1}{\bar{\sigma}_{i}^{2}(\hat{\boldsymbol{\theta }}_{m})}\left( 1-\frac{%
y_{i}^{2}}{\bar{\sigma}_{i}^{2}(\hat{\boldsymbol{\theta }}_{m})}\right) 
\frac{\partial \bar{\sigma}_{i}^{2}(\hat{\boldsymbol{\theta }}_{m})}{%
\partial \beta },\;\;\;\;\;1\leq i\leq m+\mathcal{n}.
\end{equation*}%
To implement the above calculation, we need the following recursive
equations: 
\begin{equation*}
\frac{\partial \bar{\sigma}_{i}^{2}(\hat{\boldsymbol{\theta }}_{m})}{%
\partial \alpha }=y_{i-1}^{2}+\beta \frac{\partial \bar{\sigma}_{i-1}^{2}(%
\hat{\boldsymbol{\theta }}_{m})}{\partial \alpha },
\end{equation*}%
\begin{equation*}
\frac{\partial \bar{\sigma}_{i}^{2}(\hat{\boldsymbol{\theta }}_{m})}{%
\partial \beta }=\bar{\sigma}_{i-1}^{2}(\hat{\boldsymbol{\theta }}%
_{m})+\beta \frac{\partial \bar{\sigma}_{i-1}^{2}(\hat{\boldsymbol{\theta }}%
_{m})}{\partial \beta },
\end{equation*}%
with initial values set to 
\begin{equation*}
\frac{\partial \bar{\sigma}_{0}^{2}(\hat{\boldsymbol{\theta }}_{m})}{%
\partial \alpha }=\frac{\partial \bar{\sigma}_{0}^{2}(\hat{\boldsymbol{%
\theta }}_{m})}{\partial \beta }=0,
\end{equation*}%
which follows \cite{fiorentini1996analytic} and \cite{pelagatti2009variance}.

\item Estimate the covariance matrix of the derivatives in the training
sample 
\begin{equation*}
\hat{\mathbf{D}}_{m}=\dfrac{1}{m}\sum_{i=i}^{m}\left( \frac{\partial \bar{%
\ell}_{i}(\hat{\boldsymbol{\theta }}_{m})}{\partial \alpha },\frac{\partial 
\bar{\ell}_{i}(\hat{\boldsymbol{\theta }}_{m})}{\partial \beta }\right)
^{\top }\left( \frac{\partial \bar{\ell}_{i}(\hat{\boldsymbol{\theta }}_{m})%
}{\partial \alpha },\frac{\partial \bar{\ell}_{i}(\hat{\boldsymbol{\theta }}%
_{m})}{\partial \beta }\right) .
\end{equation*}

\item Obtain the detector 
\begin{equation*}
\mathcal{D}_{m}(k)=\mathcal{r}_{m,k}^{\top }(\hat{\boldsymbol{\theta }}_{m})%
\hat{\mathbf{D}}_{m}^{-1}\mathcal{r}_{m,k}(\hat{\boldsymbol{\theta }}%
_{m}),\;\;\;\;\;m\leq k<m+\mathcal{n},
\end{equation*}%
where 
\begin{equation*}
\mathcal{r}_{m,k}(\boldsymbol{\theta })=\sum_{i=m+1}^{m+k}\left( \frac{%
\partial \bar{\ell}_{i}(\boldsymbol{\theta })}{\partial \alpha },\frac{%
\partial \bar{\ell}_{i}(\boldsymbol{\theta })}{\partial \beta }\right)
^{\top }.
\end{equation*}

\item Based on Theorem \ref{ma1}(\textit{i}), a break is flagged as soon as
the detector $\mathcal{D}_{m}(k)$ exceeds the tuned boundary function 
\begin{equation*}
\mathfrak{g}_{m}(k)=\mathcal{c}\mathcal{n}\left( 1+\dfrac{1}{\log (m)}%
\right) ^{2}\left( 1+\dfrac{k}{m}\right) ^{2}\left( \dfrac{k}{\mathcal{n}}%
\right) ^{\eta },\quad \;\;\mbox{where}\;\;1\leq k\leqslant \mathcal{n}-1;
\end{equation*}%
based on Theorem \ref{ma1}(\textit{ii}), a break is indicated as soon as the
detector $\mathcal{D}_{m}(k)$ goes above the tuned boundary function 
\begin{equation*}
\bar{\mathfrak{g}}_{m}(k)=\mathcal{c}r\left( 1+\dfrac{1}{\log (m)}\right)
^{2}\left( 1+\dfrac{k}{m}\right) ^{2}\left( \frac{k}{r}\right) ^{\eta
},\quad \;\;\mbox{where}\;\;r\leq k\leqslant \mathcal{n}-1;
\end{equation*}%
otherwise, no break is reported.
\end{enumerate}

\newpage \clearpage

\subsection{Additional simulation results\label{sec:additional}}

\begin{table}[b!]
\caption{Empirical size based on Theorem 3.1(\textit{ii})}
\label{tab:size_renyi_5pct}\centering
{\scriptsize {\ 
\begin{tabular}{llrrrrrrr}
\toprule \toprule &  & \multicolumn{3}{c}{$\epsilon_i \sim \mathcal{N} (0,1)$%
} &  & \multicolumn{3}{c}{$\epsilon_i \sim $ Student's $t$} \\ 
&  & \multicolumn{1}{l}{$m= 500$} & \multicolumn{1}{l}{$m= 1000$} & 
\multicolumn{1}{l}{$m= 5000$} &  & \multicolumn{1}{l}{$m= 500$} & 
\multicolumn{1}{l}{$m= 1000$} & \multicolumn{1}{l}{$m= 5000$} \\ 
\midrule &  & \multicolumn{7}{c}{Stationary GARCH(1,1)} \\ 
&  & \multicolumn{7}{c}{$\mathcal{n}=250$} \\ 
\cmidrule{3-9} $\tau=1.3$ &  & 14.2\% & 11.9\% & 9.4\% &  & 18.5\% & 13.8\%
& 11.1\% \\ 
$\tau=1.5$ &  & 12.4\% & 10.0\% & 7.7\% &  & 15.6\% & 11.4\% & 9.1\% \\ 
$\tau=1.7$ &  & 11.0\% & 8.7\% & 6.9\% &  & 14.0\% & 10.0\% & 7.8\% \\ 
$\tau=2.0$ &  & 9.4\% & 7.7\% & 5.8\% &  & 11.9\% & 8.7\% & 6.9\% \\ 
&  & \multicolumn{7}{c}{$\mathcal{n}=500$} \\ 
\cmidrule{3-9} $\tau=1.3$ &  & 14.8\% & 10.9\% & 8.8\% &  & 18.9\% & 15.2\%
& 12.0\% \\ 
$\tau=1.5$ &  & 12.8\% & 9.4\% & 7.1\% &  & 15.7\% & 12.9\% & 9.8\% \\ 
$\tau=1.7$ &  & 11.2\% & 8.3\% & 6.3\% &  & 13.9\% & 11.6\% & 8.6\% \\ 
$\tau=2.0$ &  & 10.0\% & 7.2\% & 5.4\% &  & 11.9\% & 10.0\% & 7.7\% \\ 
\midrule &  & \multicolumn{7}{c}{NonStationary GARCH(1,1)} \\ 
&  & \multicolumn{7}{c}{$\mathcal{n}=250$} \\ 
\cmidrule{3-9} $\tau=1.3$ &  & 13.8\% & 11.6\% & 10.1\% &  & 19.1\% & 14.5\%
& 11.2\% \\ 
$\tau=1.5$ &  & 11.4\% & 9.5\% & 8.5\% &  & 15.9\% & 12.1\% & 9.7\% \\ 
$\tau=1.7$ &  & 10.4\% & 8.3\% & 7.6\% &  & 14.0\% & 10.9\% & 8.3\% \\ 
$\tau=2.0$ &  & 8.7\% & 7.3\% & 6.6\% &  & 12.0\% & 9.5\% & 7.1\% \\ 
&  & \multicolumn{7}{c}{$\mathcal{n}=500$} \\ 
\cmidrule{3-9} $\tau=1.3$ &  & 13.3\% & 10.8\% & 9.7\% &  & 18.3\% & 14.2\%
& 11.0\% \\ 
$\tau=1.5$ &  & 11.8\% & 9.0\% & 8.3\% &  & 15.6\% & 12.4\% & 9.1\% \\ 
$\tau=1.7$ &  & 10.5\% & 7.9\% & 7.3\% &  & 14.0\% & 10.7\% & 8.0\% \\ 
$\tau=2.0$ &  & 9.1\% & 7.2\% & 6.6\% &  & 12.4\% & 9.5\% & 7.1\% \\ 
\bottomrule \bottomrule &  &  &  &  &  &  &  & 
\end{tabular}
} }
\end{table}

\begin{table}[tbph]
\caption{Empirical power based on Theorem 3.1(\textit{ii}) for a change at $%
k^{\ast }=\lfloor \protect\sqrt{\mathcal{n}}\rfloor $}
\label{tab:power_renyi_early}\centering
{\scriptsize {\ 
\begin{tabular}{llrrrrrrr}
\toprule \toprule $\mathcal{n}=500$ &  & \multicolumn{3}{c}{$H_{A,1}$} &  & 
\multicolumn{3}{c}{$H_{A,2}$} \\ 
\cmidrule{3-5}\cmidrule{7-9} before $k^*=\lfloor\sqrt{\mathcal{n}}\rfloor$ & 
& \multicolumn{3}{c}{$\beta_0=0.80$} &  & \multicolumn{3}{c}{$\beta_0=0.80$}
\\ 
after $k^*=\lfloor\sqrt{\mathcal{n}}\rfloor$ &  & \multicolumn{3}{c}{$%
\beta_1=0.60$} &  & \multicolumn{3}{c}{$\beta_1=0.90$} \\ 
\midrule $\epsilon_i \sim \mathcal{N} (0,1)$ &  & \multicolumn{1}{c}{$m=500$}
& \multicolumn{1}{c}{$m=1000$} & \multicolumn{1}{c}{$m=5000$} &  & 
\multicolumn{1}{c}{$m=500$} & \multicolumn{1}{c}{$m=1000$} & 
\multicolumn{1}{c}{$m=5000$} \\ 
\midrule $\eta=1.3$ &  & 83.18\% & 96.92\% & 100.00\% &  & 97.38\% & 100.00\%
& 100.00\% \\ 
$\eta=1.5$ &  & 61.78\% & 81.22\% & 99.62\% &  & 90.36\% & 98.68\% & 100.00\%
\\ 
$\eta=1.7$ &  & 41.70\% & 50.02\% & 83.42\% &  & 78.28\% & 92.08\% & 99.94\%
\\ 
$\eta=2.0$ &  & 23.40\% & 21.54\% & 21.56\% &  & 57.94\% & 65.18\% & 86.94\%
\\ 
\midrule $\epsilon_i \sim $ Student's $t$ &  &  &  &  &  &  &  &  \\ 
\midrule $\eta=1.3$ &  & 52.46\% & 71.82\% & 96.60\% &  & 90.46\% & 97.16\%
& 99.88\% \\ 
$\eta=1.5$ &  & 29.94\% & 37.48\% & 64.74\% &  & 82.34\% & 91.20\% & 99.06\%
\\ 
$\eta=1.7$ &  & 19.36\% & 16.00\% & 16.12\% &  & 72.08\% & 79.40\% & 93.50\%
\\ 
$\eta=2.0$ &  & 12.90\% & 8.98\% & 6.62\% &  & 55.80\% & 57.54\% & 62.86\%
\\ 
\midrule \midrule $\mathcal{n}=500$ &  & \multicolumn{3}{c}{$H_{A,3}$} &  & 
\multicolumn{3}{c}{$H_{A,4}$} \\ 
\cmidrule{3-5}\cmidrule{7-9} before $k^*=\lfloor\sqrt{\mathcal{n}}\rfloor$ & 
& \multicolumn{3}{c}{$\beta_0=0.90$} &  & \multicolumn{3}{c}{$\beta_0=0.90$}
\\ 
after $k^*=\lfloor\sqrt{\mathcal{n}}\rfloor$ &  & \multicolumn{3}{c}{$%
\beta_1=0.80$} &  & \multicolumn{3}{c}{$\beta_1=1.00$} \\ 
\midrule $\epsilon_i \sim \mathcal{N} (0,1)$ &  & \multicolumn{1}{c}{$m=500$}
& \multicolumn{1}{c}{$m=1000$} & \multicolumn{1}{c}{$m=5000$} &  & 
\multicolumn{1}{c}{$m=500$} & \multicolumn{1}{c}{$m=1000$} & 
\multicolumn{1}{c}{$m=5000$} \\ 
\midrule $\eta=1.3$ &  & 100.00\% & 100.00\% & 100.00\% &  & 96.02\% & 
99.78\% & 100.00\% \\ 
$\eta=1.5$ &  & 99.88\% & 100.00\% & 100.00\% &  & 87.06\% & 97.58\% & 
100.00\% \\ 
$\eta=1.7$ &  & 99.48\% & 99.98\% & 100.00\% &  & 73.94\% & 86.94\% & 99.68\%
\\ 
$\eta=2.0$ &  & 97.04\% & 99.70\% & 100.00\% &  & 53.06\% & 57.54\% & 77.02\%
\\ 
\midrule $\epsilon_i \sim $ Student's $t$ &  &  &  &  &  &  &  &  \\ 
\midrule $\eta=1.3$ &  & 99.64\% & 99.98\% & 100.00\% &  & 86.36\% & 94.92\%
& 99.70\% \\ 
$\eta=1.5$ &  & 99.26\% & 99.86\% & 100.00\% &  & 75.78\% & 86.50\% & 97.92\%
\\ 
$\eta=1.7$ &  & 98.02\% & 99.76\% & 100.00\% &  & 64.62\% & 72.14\% & 87.64\%
\\ 
$\eta=2.0$ &  & 94.98\% & 98.28\% & 99.98\% &  & 48.70\% & 47.94\% & 50.70\%
\\ 
\bottomrule \bottomrule &  &  &  &  &  &  &  & 
\end{tabular}
} }
\end{table}

\begin{table}[tbph]
\caption{Empirical power based on Theorem 3.1(\textit{ii}) for a change at $%
k^{\ast }=\lfloor 0.5\mathcal{n}\rfloor $}
\label{tab:power_renyi_middle}\centering
{\scriptsize {\ 
\begin{tabular}{llrrrrrrr}
\toprule \toprule $\mathcal{n}=500$ &  & \multicolumn{3}{c}{$H_{A,1}$} &  & 
\multicolumn{3}{c}{$H_{A,2}$} \\ 
\cmidrule{3-5}\cmidrule{7-9} before $k^*=\lfloor0.5\mathcal{n}\rfloor$ &  & 
\multicolumn{3}{c}{$\beta_0=0.80$} &  & \multicolumn{3}{c}{$\beta_0=0.80$}
\\ 
after $k^*=\lfloor0.5\mathcal{n}\rfloor$ &  & \multicolumn{3}{c}{$%
\beta_1=0.60$} &  & \multicolumn{3}{c}{$\beta_1=0.90$} \\ 
\midrule $\epsilon_i \sim \mathcal{N} (0,1)$ &  & \multicolumn{1}{c}{$m=500$}
& \multicolumn{1}{c}{$m=1000$} & \multicolumn{1}{c}{$m=5000$} &  & 
\multicolumn{1}{c}{$m=500$} & \multicolumn{1}{c}{$m=1000$} & 
\multicolumn{1}{c}{$m=5000$} \\ 
\midrule $\eta=1.3$ &  & 21.96\% & 31.82\% & 68.10\% &  & 55.94\% & 74.74\%
& 97.18\% \\ 
$\eta=1.5$ &  & 13.38\% & 11.70\% & 20.02\% &  & 36.16\% & 48.14\% & 80.00\%
\\ 
$\eta=1.7$ &  & 11.08\% & 7.88\% & 7.00\% &  & 22.76\% & 24.90\% & 38.76\%
\\ 
$\eta=2.0$ &  & 9.34\% & 6.60\% & 5.90\% &  & 12.82\% & 9.52\% & 7.22\% \\ 
\midrule $\epsilon_i \sim $ Student's $t$ &  &  &  &  &  &  &  &  \\ 
\midrule $\eta=1.3$ &  & 19.90\% & 17.22\% & 22.84\% &  & 54.20\% & 62.88\%
& 79.52\% \\ 
$\eta=1.5$ &  & 15.70\% & 11.72\% & 9.72\% &  & 39.02\% & 41.70\% & 52.72\%
\\ 
$\eta=1.7$ &  & 13.78\% & 9.82\% & 8.06\% &  & 27.32\% & 25.60\% & 24.78\%
\\ 
$\eta=2.0$ &  & 12.16\% & 8.30\% & 6.84\% &  & 16.74\% & 12.46\% & 8.14\% \\ 
\midrule \midrule $\mathcal{n}=500$ &  & \multicolumn{3}{c}{$H_{A,3}$} &  & 
\multicolumn{3}{c}{$H_{A,4}$} \\ 
\cmidrule{3-5}\cmidrule{7-9} before $k^*=\lfloor0.5\mathcal{n}\rfloor$ &  & 
\multicolumn{3}{c}{$\beta_0=0.90$} &  & \multicolumn{3}{c}{$\beta_0=0.90$}
\\ 
after $k^*=\lfloor0.5\mathcal{n}\rfloor$ &  & \multicolumn{3}{c}{$%
\beta_1=0.80$} &  & \multicolumn{3}{c}{$\beta_1=1.00$} \\ 
\midrule $\epsilon_i \sim \mathcal{N} (0,1)$ &  & \multicolumn{1}{c}{$m=500$}
& \multicolumn{1}{c}{$m=1000$} & \multicolumn{1}{c}{$m=5000$} &  & 
\multicolumn{1}{c}{$m=500$} & \multicolumn{1}{c}{$m=1000$} & 
\multicolumn{1}{c}{$m=5000$} \\ 
\midrule $\eta=1.3$ &  & 96.20\% & 99.36\% & 99.98\% &  & 49.88\% & 66.64\%
& 94.26\% \\ 
$\eta=1.5$ &  & 91.46\% & 97.62\% & 99.78\% &  & 31.74\% & 38.80\% & 68.92\%
\\ 
$\eta=1.7$ &  & 82.48\% & 92.58\% & 98.96\% &  & 19.56\% & 18.66\% & 26.42\%
\\ 
$\eta=2.0$ &  & 64.50\% & 75.26\% & 91.88\% &  & 12.50\% & 9.26\% & 7.16\%
\\ 
\midrule $\epsilon_i \sim $ Student's $t$ &  &  &  &  &  &  &  &  \\ 
\midrule $\eta=1.3$ &  & 92.50\% & 97.56\% & 99.84\% &  & 44.16\% & 51.12\%
& 71.38\% \\ 
$\eta=1.5$ &  & 85.94\% & 93.78\% & 98.94\% &  & 30.78\% & 31.32\% & 39.32\%
\\ 
$\eta=1.7$ &  & 77.12\% & 86.66\% & 96.38\% &  & 21.04\% & 17.84\% & 15.50\%
\\ 
$\eta=2.0$ &  & 61.24\% & 68.16\% & 83.06\% &  & 14.72\% & 10.76\% & 7.32\%
\\ 
\bottomrule \bottomrule &  &  &  &  &  &  &  & 
\end{tabular}
} }
\end{table}

\clearpage
\newpage

\subsection{Further results - empirical application using R\'{e}nyi weights\label{renyiempirics}}

Figure \ref{fig:renyi_emp} presents the detector $\mathcal{D}_m(k)$ (with
the choice of $\eta=1.5$) versus the boundary function $\bar{\mathfrak{g}}%
_m(k)$ based on Theorem 3.1 (\textit{ii}) for the four stocks during the
same periods analysed in Section 5. The monitoring procedure based on the
R\'enyi type statistics detects a change of AAPL on March 2nd, 2020, a
change of MBCN on March 13th, 2020, a change of GENE on April 27th, 2011,
and a change of NLP on September 4th, 2012. Such results illustrate the
merit of R\'enyi type statistics for the fast detection of very early
changes, in particular for AAPL (nearly 5 months earlier) and MBCN (about 2
months earlier), compared to the procedure based on Theorem 3.1(\textit{i}).

\begin{figure}[h]
\centering
\includegraphics[width=0.45\linewidth]{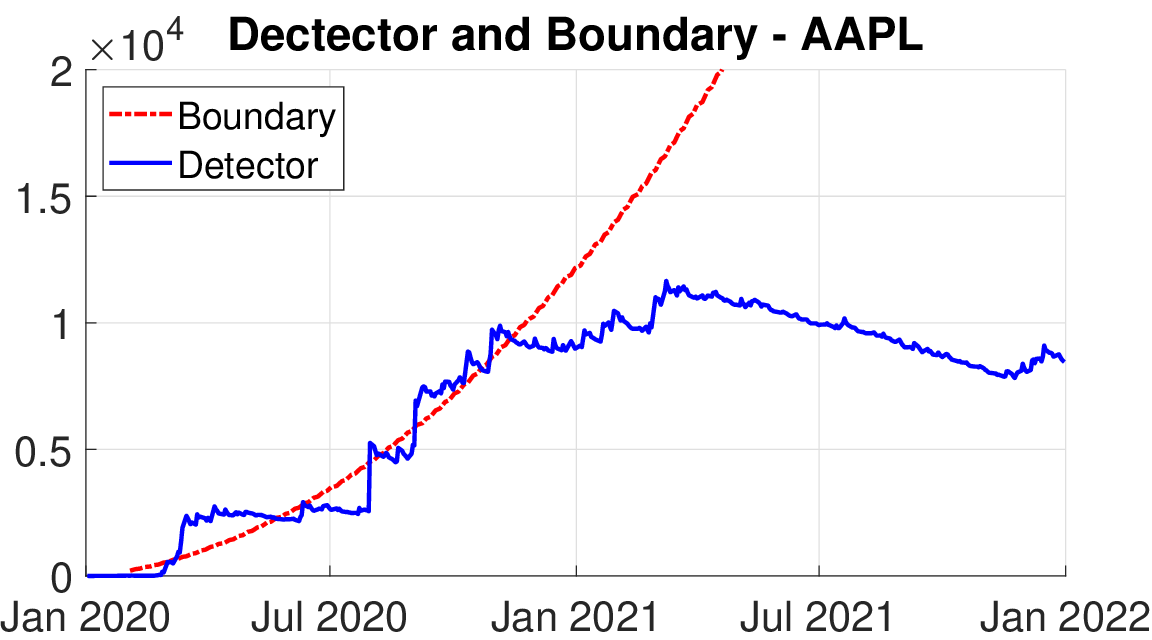}  %
\includegraphics[width=0.45\linewidth]{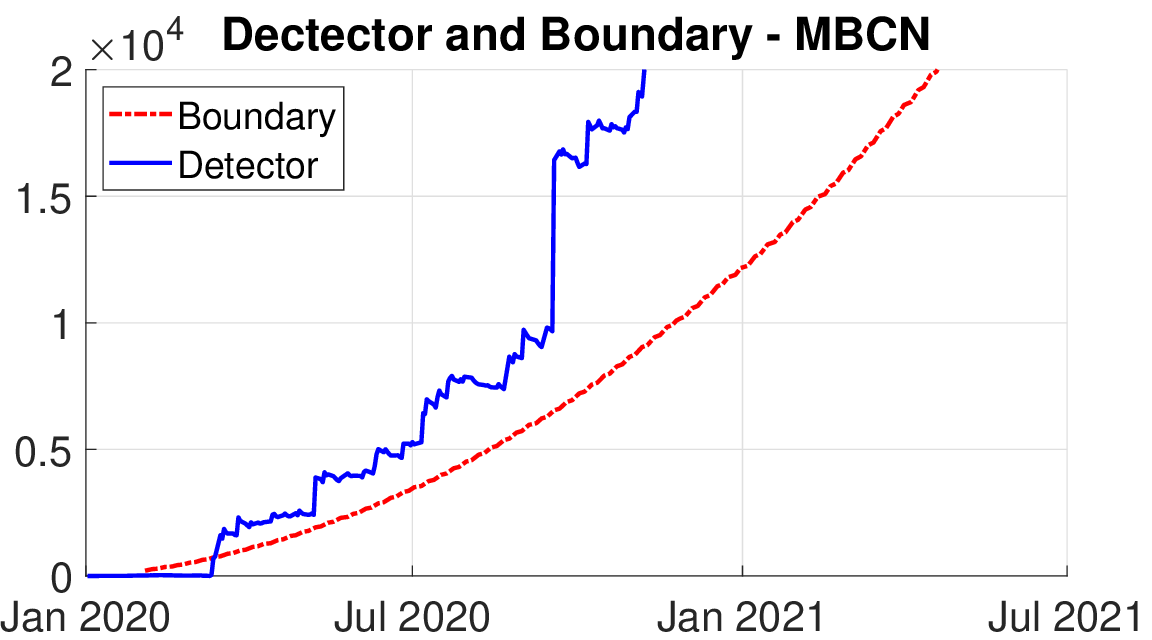}  \newline
\includegraphics[width=0.45\linewidth]{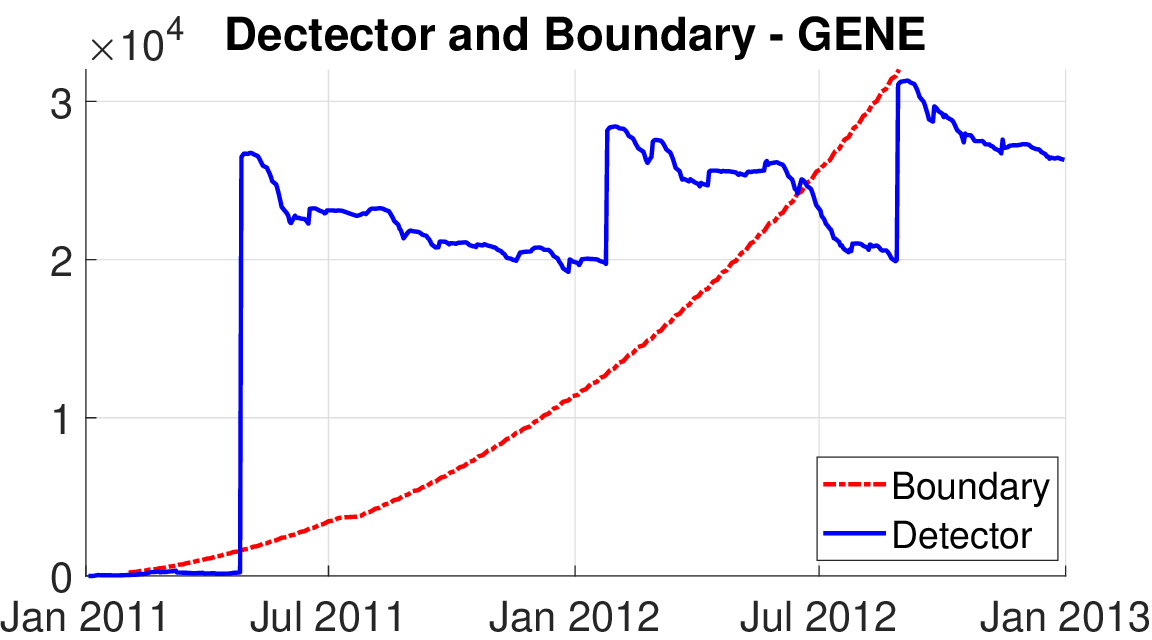}  %
\includegraphics[width=0.45\linewidth]{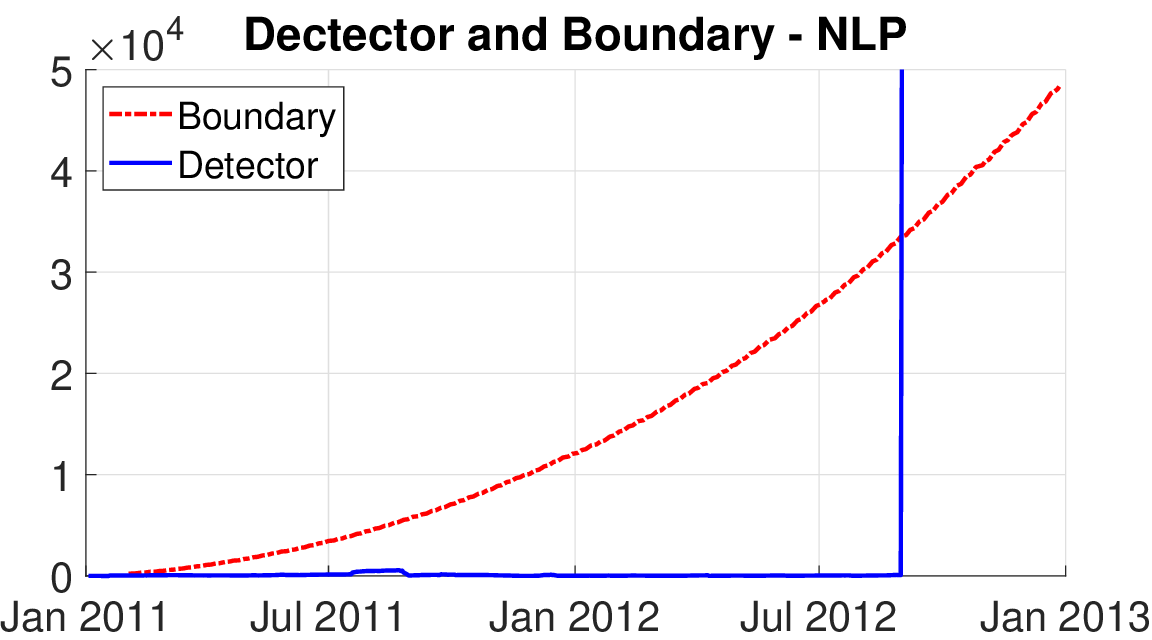}  
\caption{The detector $\mathcal{D}_m(k)$ versus the boundary function $\bar{%
\mathfrak{g}}_m(k)$ based on Theorem 3.1 (\textit{ii})}
\label{fig:renyi_emp}
\end{figure}

\newpage \clearpage

\renewcommand*{\thesection}{\Alph{section}}

\setcounter{subsection}{-1} %
\setcounter{subsubsection}{-1} \setcounter{equation}{0} \setcounter{lemma}{0}
\setcounter{theorem}{0} \renewcommand{\theassumption}{B.\arabic{assumption}} 
\renewcommand{\thetheorem}{B.\arabic{theorem}} \renewcommand{\thelemma}{B.%
\arabic{lemma}} \renewcommand{\thecorollary}{B.\arabic{corollary}} %
\renewcommand{\theequation}{B.\arabic{equation}}

\section{Technical lemmas \label{lemmas}}

Henceforth, we use $\nabla $ and $\nabla ^{2}$ for the gradient vector and
the Hessian matrix. We begin with some facts and notation which we will use
throughout the proofs.\newline

%\medskip

If \eqref{stalt} holds, then there is a unique, stationary, non anticipative
sequence $\{x_{i},-\infty <i<\infty \}$ satisfying 
\begin{equation*}
x_{i}=h_{i}\epsilon _{i},\quad \;\;h_{i}^{2}=\omega _{0}+\alpha
_{0}x_{i-1}^{2}+\beta _{0}h_{i-1}^{2},\quad \;\;-\infty <i<\infty
\end{equation*}%
(cf. \citealp{berkes2003garch}, and \citealp{francq2019garch}). Next we
define 
\begin{equation*}
\bar{h}_{i}^{2}(\mbox{\boldmath${\theta}$})=\omega +\alpha x_{i-1}^{2}+\beta 
\bar{h}_{i-1}^{2}(\mbox{\boldmath${\theta}$}),
\end{equation*}%
$\mbox{\boldmath${\theta}$}=(\omega ,\alpha ,\beta )^{\top }$. The
stationary version of the likelihood function is 
\begin{equation*}
\bar{f}_{i}(\mbox{\boldmath${\theta}$})=\log \bar{h}_{i}^{2}(%
\mbox{\boldmath${\theta}$})+\frac{x_{i}^{2}}{\bar{h}_{i}^{2}(%
\mbox{\boldmath${\theta}$})}
\end{equation*}%
and we define 
\begin{equation*}
{{\mathcal{f}}}_{k}(\mbox{\boldmath${\theta}$})=\sum_{i=1}^{k}\bar{f}_{i}(%
\mbox{\boldmath${\theta}$}).
\end{equation*}%
Similarly, the stationary version of ${\mathcal{r}}_{m,k}(%
\mbox{\boldmath${\theta}$})$ is 
\begin{equation}
\bar{{\mathcal{r}}}_{m,k}(\mbox{\boldmath${\theta}$})=\sum_{i=m+1}^{m+k}%
\left( \frac{\partial \bar{f}_{i}(\mbox{\boldmath${\theta}$})}{\partial
\alpha },\frac{\partial \bar{f}_{i}(\mbox{\boldmath${\theta}$})}{\partial
\beta }\right) ^{\top }.  \label{mcrsta}
\end{equation}%
\citet{berkes2003garch} and \citet{francq2019garch} showed that the matrix 
\begin{equation*}
{\mathbf{H}}=E\nabla ^{2}\bar{f}_{0}(\mbox{\boldmath${\theta}$}_{0}),
\end{equation*}%
exists and it is non singular.

After the change the observations are given by 
\begin{equation}
y_{i}=\sigma _{i}\epsilon _{i}\;\;\;\mbox{and}\;\;\;\sigma _{i}^{2}=\omega
_{A}+\alpha _{A}y_{i-1}^{2}+\beta _{A}\sigma _{i-1}^{2},\;\;\;i>m+k^{\ast }.
\label{alty}
\end{equation}

\begin{lemma}
\label{leone} If Assumptions \ref{aspos} and \ref{asin} hold, then we have 
\begin{equation*}
m\left( \hat{\mbox{\boldmath${\theta}$}}_{m}-\mbox{\boldmath${\theta}$}%
_{0}\right) =-{\mathbf{H}}^{-1}\nabla {{\mathcal{f}}}_{m}(%
\mbox{\boldmath${\theta}$}_{0})+O_{P}\left( \frac{1}{m}\right) .
\end{equation*}

\begin{proof}
The lemma is taken from \citet{berkes2003garch} and \citealp{francq2019garch}
(cf.\ also Lemma 5.3 in \citealp{horvath2024detecting}).
\end{proof}
\end{lemma}

\begin{lemma}
\label{letwo} If Assumptions \ref{aspos}--\ref{assh2} hold, then we have%
\begin{equation}
\max_{1\leq k\leq {\mathcal{n}}}\frac{1}{k^{\zeta }}\sup_{%
\mbox{\boldmath${\theta}$}\in \mbox{\boldmath${\theta}$}}\left\Vert {%
\mathcal{r}}_{m,k}(\mbox{\boldmath${\theta}$})-\bar{{\mathcal{r}}}_{m,k}(%
\mbox{\boldmath${\theta}$})\right\Vert =O_{P}(1),  \label{letwo1}
\end{equation}%
with some $\zeta <1/2$ and 
\begin{equation}
\max_{1\leq k\leq {\mathcal{n}}}\frac{1}{{\mathcal{n}}^{1/2}(k/{\mathcal{n}}%
)^{\gamma }}\sup_{\mbox{\boldmath${\theta}$}\in \mbox{\boldmath${\theta}$}%
}\left\Vert {\mathcal{r}}_{m,k}(\mbox{\boldmath${\theta}$})-\bar{{\mathcal{r}%
}}_{m,k}(\mbox{\boldmath${\theta}$})\right\Vert =o_{P}(1),  \label{letwo2}
\end{equation}%
for all $0\leq \gamma <1/2$.

\begin{proof}
It is shown in \citet{berkes2003garch}, \citet{francq2004maximum}, and %
\citet{francq2019garch} that there is a constant $0<\rho <1$ such that 
\begin{equation*}
\sup_{\mbox{\boldmath${\theta}$}\in \mbox{\boldmath${\theta}$}}\left\vert 
\bar{\sigma}_{m+i}^{2}(\mbox{\boldmath${\theta}$})-\bar{h}_{m+i}^{2}(%
\mbox{\boldmath${\theta}$})\right\vert =O\left( \rho ^{i}\right) \quad \quad
a.s.
\end{equation*}%
\begin{equation*}
\sup_{\mbox{\boldmath${\theta}$}\in \mbox{\boldmath${\theta}$}}\left\Vert
\nabla \bar{\sigma}_{m+i}^{2}(\mbox{\boldmath${\theta}$})-\nabla \bar{h}%
_{m+i}^{2}(\mbox{\boldmath${\theta}$})\right\Vert =O\left( \rho ^{i}\right)
\quad \quad a.s.
\end{equation*}%
and 
\begin{equation*}
\left\vert y_{m+i}^{2}-x_{m+i}^{2}\right\vert =O\left( \rho ^{i}\right)
\quad \quad a.s.
\end{equation*}%
Since $\underline{\omega }\leq \bar{\sigma}_{h}^{2}(\mbox{\boldmath${%
\theta}$})$ and $\underline{\omega }\leq \bar{h}_{j}^{2}(\mbox{\boldmath${%
\theta}$})$ for all $j\geq 1$, elementary calculations yield \eqref{letwo1}.
Now \eqref{letwo1} implies 
\begin{equation*}
\max_{1\leq k\leq {\mathcal{n}}}\frac{1}{{\mathcal{n}}^{1/2}(k/{\mathcal{n}}%
)^{\gamma }}\sup_{\mbox{\boldmath${\theta}$}\in \mbox{\boldmath${\theta}$}%
}\left\Vert {\mathcal{r}}_{m,k}(\mbox{\boldmath${\theta}$})-\bar{{\mathcal{r}%
}}_{m,k}(\mbox{\boldmath${\theta}$})\right\Vert =O_{P}(1){\mathcal{n}}%
^{-1/2+\gamma }\max_{1\leq k\leq {\mathcal{n}}}k^{-\gamma +\zeta }=o_{P}(1).
\end{equation*}
\end{proof}
\end{lemma}

\begin{lemma}
\label{lethree} If Assumptions \ref{aspos}--\ref{assh2} hold, then we have 
\begin{equation*}
\max_{1\leq k\leq {\mathcal{n}}}\frac{1}{{\mathcal{n}}^{1/2}(k/{\mathcal{n}}%
)^{\gamma }}\left\Vert {\mathcal{r}}_{m,k}(\hat{\mbox{\boldmath${\theta}$}}%
_{m})-\bar{{\mathcal{r}}}_{m,k}(\mbox{\boldmath${\theta}$}_{0})\right\Vert
=o_{P}(1),
\end{equation*}%
for all $0\leq \gamma <1/2$.

\begin{proof}
Lemma \ref{letwo} yields 
\begin{equation*}
\max_{1\leq k\leq {\mathcal{n}}}\frac{1}{{\mathcal{n}}^{1/2}(k/{\mathcal{n}}%
)^{\gamma }}\sup_{\mbox{\boldmath${\theta}$}\in \mbox{\boldmath${\theta}$}%
}\left\Vert {\mathcal{r}}_{m,k}(\hat{\mbox{\boldmath${\theta}$}}_{m})-\bar{{%
\mathcal{r}}}_{m,k}(\hat{\mbox{\boldmath${\theta}$}}_{m})\right\Vert
=o_{P}(1).
\end{equation*}%
It is proven in \citet{berkes2003garch} that there is a neighbourhood of $%
\mbox{\boldmath${\theta}$}_{0}$, say $\mbox{\boldmath${\theta}$}_{0}$, such
that 
\begin{equation*}
\max_{1\leq k\leq {\mathcal{n}}}\sup_{\mbox{\boldmath${\theta}$}\in %
\mbox{\boldmath${\theta}$}_{0}}\frac{1}{k}\left\Vert \nabla ^{3}{\mathcal{f}}%
_{m,k}(\mbox{\boldmath${\theta}$})\right\Vert =O_{P}(1).
\end{equation*}%
Using a two term Taylor expansion coordinate-wise, we conclude that 
\begin{equation}
\nabla {\mathcal{f}}_{m,k}(\hat{\mbox{\boldmath${\theta}$}}_{m})-\nabla {%
\mathcal{f}}_{m,k}(\mbox{\boldmath${\theta}$}_{0})=\nabla ^{2}{\mathcal{f}}%
_{m,k}(\mbox{\boldmath${\theta}$}_{0})(\hat{\mbox{\boldmath${\theta}$}}_{m}-%
\mbox{\boldmath${\theta}$}_{0})+\mathbf{R}_{m,k},  \label{NN}
\end{equation}%
and 
\begin{equation*}
\max_{1\leq k\leq {\mathcal{n}}}\frac{1}{{\mathcal{n}}^{1/2}(k/{\mathcal{n}}%
)^{\gamma }}\Vert \mathbf{R}_{m,k}\Vert =O_{P}({\mathcal{n}}^{1/2})\Vert 
\hat{\mbox{\boldmath${\theta}$}}_{m}-\mbox{\boldmath${\theta}$}_{0}\Vert
^{2}.
\end{equation*}%
The central limit for decomposable Bernoulli shifts (cf. %
\citealp{horvath2024detecting}) yields that 
\begin{equation}
\Vert \nabla {\mathcal{f}}_{m}(\mbox{\boldmath${\theta}$}_{0})\Vert
=O_{P}\left( m^{1/2}\right) ,  \label{NNN}
\end{equation}%
and therefore Lemma \ref{leone} and Assumption \ref{assh2} give 
\begin{equation}
\max_{1\leq k\leq {\mathcal{n}}}\frac{1}{{\mathcal{n}}^{1/2}(k/{\mathcal{n}}%
)^{\gamma }}\Vert \mathbf{R}_{m,k}\Vert =o_{P}(1).  \label{N2}
\end{equation}%
According to the ergodic theorem (cf. \citealp{breiman1968}) 
\begin{equation}
\max_{1\leq k\leq {\mathcal{n}}}\frac{1}{k}\left\Vert \nabla ^{2}{\mathcal{f}%
}_{m,k}(\mbox{\boldmath${\theta}$}_{0})\right\Vert =O_{P}(1),  \label{N1}
\end{equation}%
and therefore Lemma \ref{leone} implies 
\begin{align*}
\max_{1\leq k\leq {\mathcal{n}}}\frac{1}{{\mathcal{n}}^{1/2}(k/{\mathcal{n}}%
)^{\gamma }}\left\Vert \nabla ^{2}{\mathcal{f}}_{m}(\mbox{\boldmath${%
\theta}$}_{0})(\hat{\mbox{\boldmath${\theta}$}}_{m}-\mbox{\boldmath${%
\theta}$}_{0})\right\Vert & =O_{P}\left( m^{-1/2}\right) \max_{1\leq k\leq {%
\mathcal{n}}}\frac{k}{{\mathcal{n}}^{1/2}(k/{\mathcal{n}})^{\gamma }} \\
& =O_{P}\left( \left( \frac{{\mathcal{n}}}{m}\right) ^{1/2}\right) =o_{P}(1).
\end{align*}%
Given that $\bar{{\mathcal{r}}}_{m,k}(\mbox{\boldmath${\theta}$})$ is made
of the first two coordinates of $\nabla {\mathcal{f}}_{m,k}(%
\mbox{\boldmath${\theta}$})$, the proof is complete.
\end{proof}
\end{lemma}

\begin{lemma}
\label{lefo} If Assumptions \ref{aspos}--\ref{assh2} hold, then, for any $m$%
, we can a define Gaussian process $\{\mbox{\boldmath${\Gamma}$}_{m}(t)\in
R^{2},t\geq 0\}$ with $E\mbox{\boldmath${\Gamma}$}_{m}(t)=\mathbf{0}$ and $E%
\mbox{\boldmath${\Gamma}$}_{m}(t)\mbox{\boldmath${\Gamma}$}_{m}(s)=\min (t,s)%
\mathbf{D}$, 
\begin{equation}
\mathbf{D}=E\left[ \left( \frac{\partial \bar{f}_{1}(\mbox{\boldmath${%
\theta}$}_{0})}{\partial \alpha },\frac{\partial \bar{f}_{1}(%
\mbox{\boldmath${\theta}$}_{0})}{\partial \beta }\right) ^{\top }\left( 
\frac{\partial \bar{f}_{1}(\mbox{\boldmath${\theta}$}_{0})}{\partial \alpha }%
,\frac{\partial \bar{f}_{1}(\mbox{\boldmath${\theta}$}_{0})}{\partial \beta }%
\right) \right] ,  \label{lefod}
\end{equation}%
such that for any $0\leq \gamma <1/2$ 
\begin{equation*}
\max_{1\leq k\leq {\mathcal{n}}}\frac{1}{{\mathcal{n}}^{1/2}(k/{\mathcal{n}}%
)^{\gamma }}\left\Vert \bar{{\mathcal{r}}}_{m,k}(\mbox{\boldmath${\theta}$}%
_{0})-\mbox{\boldmath${\Gamma}$}_{m}(k)\right\Vert =o_{P}(1),
\end{equation*}

\begin{proof}
It is shown in \citet{horvath2024detecting} that $\{\nabla {\mathcal{f}}%
_{m,k}(\mbox{\boldmath${\theta}$}_{0}),k\geq 1\}$ is a decomposable
Bernoulli shift. Hence $\{\bar{{\mathcal{r}}}_{m,k}(\mbox{\boldmath${%
\theta}$}_{0}),k\geq 1\}$ is also a decomposable Bernoulli shift, and
therefore by Lemma S2.1 of \citet{aue2014dependent} we can define Wiener
processes $\mbox{\boldmath${\Gamma}$}_{m}(t)\in R^{2}$ with $E%
\mbox{\boldmath${\Gamma}$}_{m}(t)=\mathbf{0}$ and $E\mbox{\boldmath${%
\Gamma}$}_{m}(t)\mbox{\boldmath${\Gamma}$}_{m}(s)=\min (t,s)\mathbf{D}$ such
that 
\begin{equation}
\sup_{1\leq k<\infty }\frac{1}{k^{\zeta }}\left\Vert \bar{{\mathcal{r}}}%
_{m,k}(\mbox{\boldmath${\theta}$}_{0})-\mbox{\boldmath${\Gamma}$}%
_{m}(k)\right\Vert =O_{P}(1),  \label{auho}
\end{equation}%
with some $\zeta <1/2$. Since 
\begin{equation*}
\sup_{1\leq k\leq {\mathcal{n}}}\frac{k^{\zeta }}{{\mathcal{n}}^{1/2}(k/{%
\mathcal{n}})^{\gamma }}=o(1),
\end{equation*}%
the result is proven.
\end{proof}
\end{lemma}

\begin{lemma}
\label{lefi} If Assumptions \ref{aspos}--\ref{asrr} hold, then we have 
\begin{equation}
\max_{r\leq k\leq {\mathcal{n}}}\frac{1}{r^{1/2}(k/r)^{\gamma }}\sup_{%
\mbox{\boldmath${\theta}$}\in \mbox{\boldmath${\theta}$}}\left\Vert {%
\mathcal{r}}_{m,k}(\mbox{\boldmath${\theta}$})-\bar{{\mathcal{r}}}_{m,k}(%
\mbox{\boldmath${\theta}$})\right\Vert =o_{P}(1),  \label{lefi1}
\end{equation}%
\begin{equation}
\max_{r\leq k\leq {\mathcal{n}}}\frac{1}{r^{1/2}(k/r)^{\gamma }}\left\Vert 
\bar{{\mathcal{r}}}_{m,k}(\hat{\mbox{\boldmath${\theta}$}}_{m})-\bar{{%
\mathcal{r}}}_{m,k}(\mbox{\boldmath${\theta}$}_{0})\right\Vert =o_{P}(1),
\label{lefi2}
\end{equation}

and 
\begin{equation}
\max_{r\leq k\leq {\mathcal{n}}}\frac{1}{r^{1/2}(k/r)^{\gamma }}\left\Vert 
\bar{{\mathcal{r}}}_{m,k}(\mbox{\boldmath${\theta}$}_{0})-%
\mbox{\boldmath${\Gamma}$}_{m}(t)\right\Vert =o_{P}(1),  \label{lefi3}
\end{equation}

for all $\gamma >1/2$, where $\{\mbox{\boldmath${\Gamma}$}_{m}(t),t\geq 0\}$
is defined in Lemma \ref{lefo}.

\begin{proof}
It follows from \eqref{letwo1} that%
\begin{align*}
\max_{r\leq k\leq {\mathcal{n}}}& \frac{1}{r^{1/2}(k/r)^{\gamma }}\sup_{%
\mbox{\boldmath${\theta}$}\in \mbox{\boldmath${\theta}$}}\left\Vert {%
\mathcal{r}}_{m,k}(\mbox{\boldmath${\theta}$})-\bar{{\mathcal{r}}}_{m,k}(%
\mbox{\boldmath${\theta}$})\right\Vert \\
& =\max_{r\leq k\leq {\mathcal{n}}}\frac{k^{\zeta }}{r^{1/2}(k/r)^{\gamma }}%
\max_{1\leq k\leq {\mathcal{n}}}\sup_{\mbox{\boldmath${\theta}$}\in %
\mbox{\boldmath${\theta}$}}\frac{1}{k^{\zeta }}\left\Vert {\mathcal{r}}%
_{m,k}(\mbox{\boldmath${\theta}$})-\bar{{\mathcal{r}}}_{m,k}(%
\mbox{\boldmath${\theta}$})\right\Vert \\
& =O_{P}\left( r^{\zeta -1/2}\right) =o_{P}(1),
\end{align*}%
since $r\rightarrow \infty $; hence \eqref{lefi1} is proven. Turning to %
\eqref{lefi2}, the expansion in \eqref{NN} implies 
\begin{align*}
\max_{r\leq k\leq {\mathcal{n}}}\frac{1}{r^{1/2}(k/r)^{\gamma }}\Vert 
\mathbf{R}_{m,k}\Vert & =\left\{ 
\begin{array}{ll}
O_{P}\left( ({\mathcal{n}}/m)(r/{\mathcal{n}})^{\gamma }r^{-1/2}\right)
,\;\;\;\mbox{if}\;\;\;1/2<\gamma <1,\vspace{0.3cm} &  \\ 
O_{P}\left( r^{1/2}/m\right) ,\;\;\;\mbox{if}\;\;\;\gamma \geq 1, & 
\end{array}%
\right. \\
& =o_{P}(1).
\end{align*}%
Now Lemma \ref{leone}, \eqref{NN} and \eqref{NNN} yield 
\begin{align*}
\max_{r\leq k\leq {\mathcal{n}}}\frac{1}{r^{1/2}(k/r)^{\gamma }}\Vert \nabla 
\bar{{\mathcal{r}}}_{m,k}(\mbox{\boldmath${\theta}$}_{0})(\hat{%
\mbox{\boldmath${\theta}$}}_{m}-\mbox{\boldmath${\theta}$}_{0})\Vert &
=\left\{ 
\begin{array}{ll}
O_{P}\left( ({\mathcal{n}}/m)(r/{\mathcal{n}})^{\gamma -1/2}\right) ,\;\;\;%
\mbox{if}\;\;\;1/2<\gamma <1\vspace{0.3cm} &  \\ 
O_{P}\left( (r/{\mathcal{n}})^{\gamma -1/2}({\mathcal{n}}/m)^{1/2}\right)
,\;\;\;\mbox{if}\;\;\;\gamma \geq 1 & 
\end{array}%
\right. \\
& =o_{P}(1),
\end{align*}%
completing the proof of \eqref{lefi2}. Finally, the approximation in %
\eqref{lefi3} is an immediate consequence of Lemma \ref{lefo}, since 
\begin{align*}
\max_{r\leq k\leq {\mathcal{n}}}\frac{1}{r^{1/2}(k/r)^{\gamma }}\left\Vert 
\bar{{\mathcal{r}}}_{m,k}(\mbox{\boldmath${\theta}$}_{0})-%
\mbox{\boldmath${\Gamma}$}_{m}(t)\right\Vert & \leq \max_{r\leq k\leq {%
\mathcal{n}}}\frac{k^{\zeta }}{r^{1/2}(k/r)^{\gamma }}\max_{1\leq k\leq {%
\mathcal{n}}}\frac{1}{k^{\zeta }}\left\Vert \bar{{\mathcal{r}}}_{m,k}(%
\mbox{\boldmath${\theta}$}_{0})-\mbox{\boldmath${\Gamma}$}_{m}(t)\right\Vert
\\
& =o_{P}(1).
\end{align*}
\end{proof}
\end{lemma}

\begin{lemma}
\label{nost1} If Assumptions \ref{aspos} and \ref{asin} hold, then we have
for all $1<\overline{\omega }$ there is $0<\rho <1$ such that 
\begin{equation*}
\sup_{1\leq i<\infty }\sup_{1/\overline{\omega }\leq \omega \leq \overline{%
\omega }}\frac{1}{\rho ^{i}}|\bar{p}_{m+i,1}(\bar{\mbox{\boldmath${\theta}$}}%
_{0},\omega )-{\mathcal{z}}_{m+i,1}|=O_{P}(1)
\end{equation*}%
and 
\begin{equation*}
\sup_{1\leq i<\infty }\sup_{1/\overline{\omega }\leq \omega \leq \overline{%
\omega }}\frac{1}{\rho ^{i}}|\bar{p}_{m+i,2}(\bar{\mbox{\boldmath${\theta}$}}%
_{0},\omega )-{\mathcal{z}}_{m+i,2}|=O_{P}(1).
\end{equation*}

\begin{proof}
Using the proofs from \citet{jensen2004asymptotic}, %
\citet{horvath2024detecting} established the result as their Lemma 5.7.
\end{proof}
\end{lemma}

\begin{lemma}
\label{nost2} We assume that Assumptions \ref{aspos}--\ref{assh2} hold.%
\newline
(i) If $0\leq \gamma <1/2$, then we have 
\begin{equation*}
\max_{1\leq k\leq {\mathcal{n}}}\frac{1}{{\mathcal{n}}^{1/2}(k/{\mathcal{n}}%
)^{\gamma }}\left\Vert {\mathcal{r}}_{m,k}(\hat{\mbox{\boldmath${\theta}$}}%
_{m})-\mathbf{q}_{m+k}\right\Vert =o_{P}(1).
\end{equation*}%
(ii) If in addition, Assumption \ref{asrr} also holds, and $\gamma >1/2$,
then we have 
\begin{equation*}
\max_{r\leq k\leq {\mathcal{n}}}\frac{1}{r^{1/2}(k/r)^{\gamma }}\left\Vert {%
\mathcal{r}}_{m,k}(\hat{\mbox{\boldmath${\theta}$}}_{m})-\mathbf{q}%
_{m+k}\right\Vert =o_{P}(1).
\end{equation*}

\begin{proof}
Let $\bar{\mbox{\boldmath${\theta}$}}_{0}=(\alpha _{0},\beta _{0})$. %
\citet{jensen2004asymptotic} showed there is a neighbourhood of $\bar{%
\mbox{\boldmath${\theta}$}}_{0}$, say $\bar{\mbox{\boldmath${\theta}$}}_{0}$%
, such that 
\begin{equation}
\sup_{\underline{\omega }\leq \omega \leq \overline{\omega }}\sup_{(\alpha
,\beta )\in \bar{\mbox{\boldmath${\theta}$}}_{0}}\frac{1}{k}\left\Vert
\nabla {\mathcal{r}}_{m,k}(\alpha ,\beta ,\omega )\right\Vert =O_{P}(1).
\label{jr1}
\end{equation}%
Let $\bar{\mbox{\boldmath${\theta}$}}_{m}$ denote the vector of the first
two coordinates of $\hat{\mbox{\boldmath${\theta}$}}_{m}$. %
\citet{jensen2004asymptotic} proved that 
\begin{equation}
\sup_{\underline{\omega }\leq \omega \leq \overline{\omega }}\Vert \bar{%
\mbox{\boldmath${\theta}$}}_{m}-\bar{\mbox{\boldmath${\theta}$}}_{0}\Vert
=O_{P}\left( m^{-1/2}\right)  \label{jr2}
\end{equation}%
(note that the estimator $\bar{\mbox{\boldmath${\theta}$}}_{m}$ might depend
on $\omega )$. Using the mean value theorem coordinate--wise, we get from
Lemma \ref{nost1} that 
\begin{equation*}
\Vert {\mathcal{r}}_{m,k}(\hat{\mbox{\boldmath${\theta}$}}_{m})-\mathbf{q}%
_{m+k}\Vert \leq \sup_{\underline{\omega }\leq \omega \leq \overline{\omega }%
}\sup_{(\alpha ,\beta )\in \bar{\mbox{\boldmath${\theta}$}}_{0}}\left\Vert
\nabla {\mathcal{r}}_{m,k}(\alpha ,\beta ,\omega )\right\Vert \Vert \bar{%
\mbox{\boldmath${\theta}$}}_{m}-\bar{\mbox{\boldmath${\theta}$}}_{0}\Vert
^{1/2}
\end{equation*}%
with probability tending to 1. Thus we get from \eqref{jr1} and \eqref{jr2}
that 
\begin{align*}
\max_{1\leq k\leq {\mathcal{n}}}\frac{1}{{\mathcal{n}}^{1/2}(k/{\mathcal{n}}%
)^{\gamma }}\left\Vert {\mathcal{r}}_{m,k}(\hat{\mbox{\boldmath${\theta}$}}%
_{m})-\mathbf{q}_{m+k}\right\Vert & =O_{P}\left( m^{-1/2}{\mathcal{n}}%
^{-1/2+\gamma }\max_{1\leq k\leq {\mathcal{n}}}k^{1-\gamma }\right) \\
& =O_{P}\left( \left( \frac{{\mathcal{n}}}{m}\right) ^{1/2}\right) =o_{P}(1).
\end{align*}%
Similarly, 
\begin{align*}
\max_{r\leq k\leq {\mathcal{n}}}& \frac{1}{r^{1/2}(k/r)^{\gamma }}\left\Vert 
{\mathcal{r}}_{m,k}(\hat{\mbox{\boldmath${\theta}$}}_{m})-\mathbf{q}%
_{m+k}\right\Vert \\
& =O_{P}\left( m^{-1/2}r^{-1/2+\gamma }\max_{r\leq k\leq {\mathcal{n}}%
}k^{1-\gamma }\right) \\
& =\left\{ 
\begin{array}{ll}
O_{P}\left( ({\mathcal{n}}/m)^{1/2}(r/{\mathcal{n}})^{\gamma _{1}/2}\right)
,\quad \;\mbox{if}\;\;\;1/2<\gamma \leq 1\vspace{0.3cm} &  \\ 
O_{P}\left( ({\mathcal{n}}/m)^{1/2}\right) ,\quad \;\mbox{if}\;\;\;\gamma >1
& 
\end{array}%
\right. \\
& =o_{P}(1).
\end{align*}
\end{proof}
\end{lemma}

\begin{lemma}
\label{nost3} If Assumptions \ref{aspos}--\ref{assh2} hold, then for any $m$
we can define a Gaussian process $\{\mbox{\boldmath${\Gamma}$}_{m}(t),t\geq
0\}$, $E\mbox{\boldmath${\Gamma}$}_{m}(t)=\mathbf{0},E\mbox{\boldmath${%
\Gamma}$}_{m}(t)\mbox{\boldmath${\Gamma}$}_{m}(s)=\min (t,s)\mathbf{D}$ with 
\begin{equation}
\mathbf{D}=E[1-\epsilon _{m}^{2}]^{2}E\left[ \left( {\mathcal{z}}_{m,1},{%
\mathcal{z}}_{m,2}\right) ^{\top }\left( {\mathcal{z}}_{m,1},{\mathcal{z}}%
_{m,2}\right) \right] .  \label{D2}
\end{equation}%
such that \newline
(i) if $0\leq \gamma <1/2$, then 
\begin{equation}
\max_{1\leq k\leq {\mathcal{n}}}\frac{1}{{\mathcal{n}}^{1/2}(k/{\mathcal{n}}%
)^{\gamma }}\left\Vert \mathbf{q}_{m+k}-\mbox{\boldmath${\Gamma}$}%
_{m}(k)\right\Vert =o_{P}(1),  \label{nost31}
\end{equation}%
(ii) if in addition, Assumption \ref{asrr} also holds and $\gamma >1/2$,
then 
\begin{equation}
\max_{r\leq k\leq {\mathcal{n}}}\frac{1}{r^{1/2}(k/r)^{\gamma }}\left\Vert 
\mathbf{q}_{m+k}-\mbox{\boldmath${\Gamma}$}_{m}(k)\right\Vert =o_{P}(1).
\label{nost32}
\end{equation}

\begin{proof}
\citet{horvath2024detecting} showed that $\{\mathbf{q}_{m+k},k\geq 1\}$ is a
decomposable Bernoulli shift and therefore Lemma S2.1 of %
\citet{aue2014dependent} implies that for any $m$ there are Gaussian
processes $\{\mbox{\boldmath${\Gamma}$}_{m}(k),k\geq 1\}$ such that 
\begin{equation}
\sup_{1\leq k<\infty }\frac{1}{k^{\zeta }}\left\Vert \mathbf{q}_{m+k}-%
\mbox{\boldmath${\Gamma}$}_{m}(k)\right\Vert =O_{P}(1)  \label{hus}
\end{equation}%
with some $\zeta <1/2$, and $E\mbox{\boldmath${\Gamma}$}(t)=\mathbf{0},\;E%
\mbox{\boldmath${\Gamma}$}(t)\mbox{\boldmath${\Gamma}$}^{\top }(s)=\min (t,s)%
\mathbf{D}$ and $\mathbf{D}$ is defined in \eqref{D2}. The approximation in (%
\ref{hus}) implies both \eqref{nost31} and \eqref{nost32} in the same way as %
\eqref{auho} implies Lemma \ref{lefo} and \eqref{lefi3}.
\end{proof}
\end{lemma}

\begin{lemma}
\label{lealt1} We assume that $H_{A}$, Assumptions \ref{aspos}--\ref{assh2}
and \ref{asalt} are satisfied.\newline
(i) If $E\log (\alpha _{A}\epsilon _{0}^{2}+\beta _{A})<0$, then 
\begin{equation*}
\max_{1\leq k\leq k^{\ast }}\frac{{\mathcal{D}}_{m,k}}{g_{m}(k)}=O_{P}\left(
\left( \frac{k^{\ast }}{{\mathcal{n}}}\right) ^{1-\eta }\right) .
\end{equation*}%
(ii) If $E\log (\alpha _{A}\epsilon _{0}^{2}+\beta _{A})<0$, then 
\begin{equation*}
\max_{r\leq k\leq \max (r,k^{\ast })}\frac{{\mathcal{D}}_{m,k}}{\bar{g}%
_{m}(k)}=O_{P}\left( 1\right) .
\end{equation*}

\begin{proof}
The lemma follows from the proof of Theorem \ref{ma1}.
\end{proof}
\end{lemma}

\medskip

Next we assume that \eqref{nostalst} holds. The observations are generated
by \eqref{alty} but in contrast to the previous case, the sequence is
explosive and cannot be approximated with a stationary sequence. However, we
can approximate the gradient of the log likelihood function with a
stationary sequence, we need only minor modifications of the arguments used
before. We use the following sequences to approximate the gradient vector of
the log likelihood function when \eqref{alty} holds for an explosive
sequence: 
\begin{equation*}
{\mathcal{v}}_{m+k^{\ast }+i,1}=\sum_{j=1}^{\infty }\epsilon _{m+k^{\ast
}+i-j}^{2}\frac{1}{\beta _{A}}\sum_{k=1}^{j}\frac{\beta _{A}}{\alpha
_{A}\epsilon _{m+k^{\ast }+i-k}^{2}+\beta _{A}}\quad \;
\end{equation*}%
and 
\begin{equation*}
{\mathcal{v}}_{m+k^{\ast }+i,2}=\sum_{j=1}^{\infty }\frac{1}{\beta _{A}}%
\sum_{k=1}^{j}\frac{\beta _{A}}{\alpha _{A}\epsilon _{m+k^{\ast
}+i-k}^{2}+\beta _{A}}.
\end{equation*}%
Let 
\begin{equation*}
\tilde{\mathbf{q}}_{m,k^{\ast },i}=\sum_{j=k^{\ast }+1}^{i}\left(
(1-\epsilon _{m+j}^{2}){\mathcal{v}}_{m+j,1},(1-\epsilon _{m+j}^{2}){%
\mathcal{v}}_{m+j,2}\right) ^{\top }.
\end{equation*}

\begin{lemma}
\label{nosalap} If $H_{A}$, Assumptions \ref{aspos}--\ref{assh2}, \ref{asalt}
are satisfied and $E\log (\alpha _{A}\epsilon _{0}^{2}+\beta _{A})>0$, then
we have 
\begin{equation*}
\max_{k^{\ast }<i\leq {\mathcal{n}}}\frac{1}{(i-k^{\ast })^{\zeta }}%
\left\Vert [{\mathcal{r}}_{m,m+i}-{\mathcal{r}}_{m,m+k^{\ast }}]-(i-k^{\ast
})\mbox{\boldmath${\Upsilon}$}-\tilde{\mathbf{q}}_{m,k^{\ast },i}\right\Vert
=O_{P}(1),
\end{equation*}%
with some $\zeta <1/2$.

\begin{proof}
After the changepoint the observations change into an explosive sequence
with parameter $\mbox{\boldmath${\theta}$}_{A}$, hence the proofs of Lemmas %
\ref{nost1} and \ref{nost2} can be repeated.
\end{proof}
\end{lemma}

\begin{lemma}
\label{lealto} If $H_{A}$, Assumptions \ref{aspos}--\ref{assh2}, \ref{asalt}
are satisfied and $E\log (\alpha _{A}\epsilon _{0}^{2}+\beta _{A})>0$, then
we can define Gaussian processes $\left\{ \tilde{\mbox{\boldmath${\Gamma}$}}%
_{m}(t),t\geq 0\right\} $ with $E\tilde{\mbox{\boldmath${\Gamma}$}}_{m}(t)=%
\mathbf{0},E\tilde{\mbox{\boldmath${\Gamma}$}}_{m}(t)\tilde{%
\mbox{\boldmath${\Gamma}$}}_{m}^{\top }(s)=\min (t,s)\mbox{\boldmath${%
\Sigma}$}_{2}$, where $\mbox{\boldmath${\Sigma}$}_{2}$ is defined in %
\eqref{Sig2}, such that 
\begin{equation*}
\max_{k^{\ast }<k\leq {\mathcal{n}}}\frac{1}{(k-k^{\ast })^{\zeta }}%
\left\Vert \tilde{\mathbf{q}}_{m,k^{\ast },i}-(i-k^{\ast })%
\mbox{\boldmath${\Upsilon}$}-\tilde{\mbox{\boldmath${\Gamma}$}}%
_{m}(i-k^{\ast })\right\Vert =O_{P}(1).
\end{equation*}

\begin{proof}
The proof is the same as of that of Lemma \ref{nost3}.
\end{proof}
\end{lemma}

\newpage

\clearpage
\renewcommand*{\thesection}{\Alph{section}}

\setcounter{subsection}{-1} \setcounter{subsubsection}{-1} %
\setcounter{equation}{0} \setcounter{lemma}{0} \setcounter{theorem}{0} %
\renewcommand{\theassumption}{C.\arabic{assumption}} 
\renewcommand{\thetheorem}{C.\arabic{theorem}} \renewcommand{\thelemma}{C.%
\arabic{lemma}} \renewcommand{\thecorollary}{C.\arabic{corollary}} %
\renewcommand{\theequation}{C.\arabic{equation}}

\section{Proofs\label{proofs}}

\begin{proof}[Proof of Theorem \protect\ref{ma1}]
We begin by showing the theorem under stationarity. It follows from Lemmas %
\ref{letwo}--\ref{lefo} that 
\begin{equation*}
\sup_{1\leq k\leq {\mathcal{n}}}\frac{1}{{\mathcal{n}}(k/{\mathcal{n}}%
)^{\eta }}\left\Vert {\mathcal{r}}_{m,k}^{\top }(\hat{\mbox{\boldmath${%
\theta}$}}_{m})\mathbf{D}^{-1}{\mathcal{r}}_{m,k}(\hat{\mbox{\boldmath${%
\theta}$}}_{m})-\mbox{\boldmath${\Gamma}$}_{m}^{\top }(k)\mathbf{D}^{-1}%
\mbox{\boldmath${\Gamma}$}_{m}(k)\right\Vert =o_{P}(1).
\end{equation*}%
The covariance structure of $\mbox{\boldmath${\Gamma}$}_{m}(k)$ implies that 
\begin{equation*}
\sup_{1\leq k\leq {\mathcal{n}}}\frac{1}{{\mathcal{n}}(k/{\mathcal{n}}%
)^{\eta }}\left\vert \mbox{\boldmath${\Gamma}$}_{m}^{\top }(k)\mathbf{D}^{-1}%
\mbox{\boldmath${\Gamma}$}_{m}(k)\right\vert \overset{{\mathcal{D}}}{=}%
\sup_{1\leq k\leq {\mathcal{n}}}\frac{1}{{\mathcal{n}}(k/{\mathcal{n}}%
)^{\eta }}\left\Vert \mathbf{W}(k)\right\Vert ^{2},
\end{equation*}%
where $\{\mathbf{W}(t),t\geq 0\}$ has independent coordinates, identically
distributed Wiener processes. Using the the scale transformation and
continuity of the Wiener process we conclude 
\begin{equation*}
\sup_{1\leq k\leq {\mathcal{n}}}\frac{1}{{\mathcal{n}}(k/{\mathcal{n}}%
)^{\eta }}\left\Vert \mathbf{W}(k)\right\Vert ^{2}\overset{{\mathcal{D}}}{=}%
\sup_{1/{\mathcal{n}}\leq k/{\mathcal{n}}\leq }\frac{1}{(k/{\mathcal{n}}%
)^{\eta }}\left\Vert \mathbf{W}(k/{\mathcal{n}})\right\Vert ^{2}\overset{{%
\mathcal{D}}}{\rightarrow }\sup_{0<t\leq }\frac{1}{t^{\eta }}\left\Vert 
\mathbf{W}(t)\right\Vert ^{2}.
\end{equation*}%
The proof of part \textit{(i)} of the theorem is complete when \eqref{stalt}
holds since 
\begin{equation}
\left\Vert \hat{\mathbf{D}}_{m}-\mathbf{D}\right\Vert =o_{P}\left( 1\right) ,
\label{Dco}
\end{equation}%
with $\mathbf{D}$ defined in \eqref{lefod}. We now turn to proving the
second part of the Theorem. We note 
\begin{align*}
\max_{r\leq k\leq {\mathcal{n}}}\frac{1}{r(k/r)^{\eta }}\mbox{\boldmath${%
\Gamma}$}_{m}^{\top }(k)\mathbf{D}^{-1}\mbox{\boldmath${\Gamma}$}_{m}(k)& 
\overset{{\mathcal{D}}}{=}\max_{r\leq k\leq {\mathcal{n}}}\frac{1}{%
r(k/r)^{\eta }}\left\Vert \mathbf{W}(k)\right\Vert ^{2} \\
& \overset{{\mathcal{D}}}{=}\max_{1\leq k/r\leq {\mathcal{n}}/r}\frac{1}{%
(k/r)^{\eta }}\left\Vert \mathbf{W}(k/r)\right\Vert ^{2} \\
& \overset{{\mathcal{D}}}{\rightarrow }\max_{1\leq t<\infty }\frac{1}{%
t^{\eta }}\left\Vert \mathbf{W}(t)\right\Vert ^{2}
\end{align*}%
due to the almost sure continuity of the Wiener process, and the Law of the
Iterated Logarithm for Wiener processes. The result in part \textit{(ii)} of
the theorem now follows from Lemma \ref{lefi} and \eqref{Dco}. Under
nonstationarity, the proof is exactly the same as above, only Lemmas \ref%
{letwo}--\ref{lefi} are replaced with Lemmas \ref{nost1}--\ref{nost3}. We
also note that according to \citet{jensen2004asymptotic} (cf.\ also %
\citealp{francq2019garch}), it holds that $\left\Vert \hat{\mathbf{D}}_{m}-%
\mathbf{D}\right\Vert =o_{P}(1)$, where $\mathbf{D}$ is defined in \eqref{D2}%
, thus completing the proof.
\end{proof}

\begin{proof}[Proof of Theorem \protect\ref{ma2}]
We note that under our conditions, it holds that $\Vert \hat{\mathbf{D}}_{m}-%
\mathbf{D}\Vert =O_{P}(m^{-\delta })$, for some $\delta >0$. Hence in the
proofs we can replace $\hat{\mathbf{D}}_{m}^{-1}$ with $\mathbf{D}^{-1}$ in
the definition of the detector, without loss of generality. We wish to point
out that the definition of $\mathbf{D}$ depends on the sign of $E\log
(\alpha _{0}\epsilon _{0}^{2}+\beta _{0})$.\newline
First we assume that $E\log (\alpha _{A}\epsilon _{0}^{2}+\beta _{A})<0$
holds. To prove part \textit{(i)} of the theorem, we note that according to
Lemma \ref{letwo} 
\begin{equation}
\max_{1\leq k\leq {\mathcal{n}}}\frac{1}{k^{1/2}}\sup_{\mbox{\boldmath${%
\theta}$}\in \mbox{\boldmath${\theta}$}}\Vert {\mathcal{r}}_{m,k}(%
\mbox{\boldmath${\theta}$})-\bar{{\mathcal{r}}}_{m,k}(\mbox{\boldmath${%
\theta}$})\Vert =O_{P}(1),  \label{fm}
\end{equation}%
and 
\begin{equation}
\max_{a\leq k\leq {\mathcal{n}}}\frac{1}{k^{1/2}}\sup_{\mbox{\boldmath${%
\theta}$}\in \mbox{\boldmath${\theta}$}}\Vert {\mathcal{r}}_{m,k}(%
\mbox{\boldmath${\theta}$})-\bar{{\mathcal{r}}}_{m,k}(\mbox{\boldmath${%
\theta}$})\Vert =O_{P}\left( a^{-1/2+\zeta }\right) ,  \label{fmm}
\end{equation}%
where $\zeta <1/2$ and $a=(\log {\mathcal{n}})^{\phi }$ with $\phi >4$.
Following the proof of Lemma \ref{lethree} we get 
\begin{equation*}
\max_{1\leq k\leq {\mathcal{n}}}\frac{1}{k^{1/2}}\Vert \mathbf{R}_{m,k}\Vert
=O_{P}({\mathcal{n}}^{1/2})\Vert \hat{\mbox{\boldmath${\theta}$}}_{m}-%
\mbox{\boldmath${\theta}$}_{0}\Vert ^{2}=O_{P}\left( \frac{{\mathcal{n}}%
^{1/2}}{m}\right) ,
\end{equation*}%
and by \eqref{N2} 
\begin{equation*}
\max_{1\leq k\leq {\mathcal{n}}}\frac{1}{k^{1/2}}\left\Vert \nabla ^{2}{%
\mathcal{f}}_{m,k}(\mbox{\boldmath${\theta}$}_{0})(\hat{\mbox{\boldmath${%
\theta}$}}_{m}-\mbox{\boldmath${\theta}$})\right\Vert =O_{P}\left( \frac{{%
\mathcal{n}}^{1/2}}{m^{1/2}}\right) ,
\end{equation*}%
\begin{equation*}
\max_{a\leq k\leq {\mathcal{n}}}\frac{1}{k^{1/2}}\left\Vert \nabla ^{2}{%
\mathcal{f}}_{m,k}(\mbox{\boldmath${\theta}$}_{0})(\hat{\mbox{\boldmath${%
\theta}$}}_{m}-\mbox{\boldmath${\theta}$})\right\Vert =O_{P}\left( \frac{%
a^{1/2}}{m^{1/2}}\right) .
\end{equation*}%
Using the approximation in \eqref{auho} we conclude 
\begin{equation}
\max_{1\leq k\leq {\mathcal{n}}}\frac{1}{k^{1/2}}\left\Vert \bar{{\mathcal{r}%
}}_{m,k}(\mbox{\boldmath${\theta}$}_{0})-\mbox{\boldmath${\Gamma}$}%
_{m}(k)\right\Vert =O_{P}\left( 1\right)  \label{stac}
\end{equation}%
and 
\begin{equation}
\max_{1\leq k\leq a}\frac{1}{k^{1/2}}\left\Vert \bar{{\mathcal{r}}}_{m,k}(%
\mbox{\boldmath${\theta}$}_{0})-\mbox{\boldmath${\Gamma}$}_{m}(k)\right\Vert
=O_{P}\left( a^{-1/2+\zeta }\right) .  \label{stacc}
\end{equation}%
The Darling--Erd\H{o}s (cf. \citealp{csorgo1997}, pp.\ 363--365) yields 
\begin{equation*}
\frac{1}{2\log \log {\mathcal{n}}}\max_{1\leq k\leq {\mathcal{n}}}\frac{1}{k}%
\mbox{\boldmath${\Gamma}$}_{m}^{\top }(k)\mathbf{D}^{-1}\mbox{\boldmath${%
\Gamma}$}_{m}(k)\overset{P}{\rightarrow }1,
\end{equation*}%
and 
\begin{equation*}
\max_{a\leq k\leq {\mathcal{n}}}\frac{1}{k}\mbox{\boldmath${\Gamma}$}%
_{m}^{\top }(k)\mathbf{D}^{-1}\mbox{\boldmath${\Gamma}$}_{m}(k)=O_{P}\left(
\log \log \log {\mathcal{n}}\right) .
\end{equation*}%
Thus we obtain that 
\begin{equation*}
\lim_{m\rightarrow \infty }P\left\{ \max_{1\leq k\leq {\mathcal{n}}}\frac{1}{%
k}{\mathcal{r}}_{m,k}^{\top }(\hat{\mbox{\boldmath${\theta}$}}_{m})\mathbf{D}%
^{-1}{\mathcal{r}}_{m,k}(\hat{\mbox{\boldmath${\theta}$}}_{m})=\max_{a\leq
k\leq {\mathcal{n}}}\frac{1}{k}{\mathcal{r}}_{m,k}^{\top }(\hat{%
\mbox{\boldmath${\theta}$}}_{m})\mathbf{D}^{-1}{\mathcal{r}}_{m,k}(\hat{%
\mbox{\boldmath${\theta}$}}_{m})\right\} =1.
\end{equation*}%
Putting together our estimates on the set $a\leq k\leq {\mathcal{n}}$, we
get 
\begin{equation*}
\max_{a\leq k\leq {\mathcal{n}}}\frac{1}{k}\left\vert {\mathcal{r}}%
_{m,k}^{\top }(\hat{\mbox{\boldmath${\theta}$}}_{m})\mathbf{D}^{-1}{\mathcal{%
r}}_{m,k}(\hat{\mbox{\boldmath${\theta}$}}_{m})-\mbox{\boldmath${\Gamma}$}%
_{m}^{\top }(k)\mathbf{D}^{-1}\mbox{\boldmath${\Gamma}$}_{m}(k)\right\vert
=O_{P}\left( (\log {\mathcal{n}})^{-\bar{\phi}}\right) ,
\end{equation*}%
with some $\bar{\phi}>0$. Since 
\begin{equation*}
\max_{a\leq k\leq {\mathcal{n}}}\mbox{\boldmath${\Gamma}$}_{m}^{\top }(k)%
\mathbf{D}^{-1}\mbox{\boldmath${\Gamma}$}_{m}(k)\overset{{\mathcal{D}}}{=}%
\max_{a\leq k\leq {\mathcal{n}}}\Vert \mathbf{W}(k)\Vert ^{2},
\end{equation*}%
where $\{\mathbf{W}(t),t\geq 0\}$ is defined in Theorem \ref{ma1}. Using
Appendix 3 in \citet{csorgo1997}, we obtain that 
\begin{equation*}
\lim_{{\mathcal{n}}\rightarrow \infty }P\left\{ a(\log {\mathcal{n}}%
)\max_{a\leq k\leq {\mathcal{n}}}\Vert \mathbf{W}(k)\Vert \leq x+b_{2}(\log {%
\mathcal{n}})\right\} =\exp (-e^{-x})
\end{equation*}%
for all $x$. Hence the proof Theorem \ref{ma2}(i) is complete when $E\log
(\alpha _{0}\epsilon _{0}^{2}+\beta _{0})<0$.\newline
Lemma \ref{nost1}, \eqref{jr1} and \eqref{jr2} imply that \eqref{fm} and %
\eqref{fmm} can be replaced with 
\begin{equation*}
\max_{1\leq k\leq {\mathcal{n}}}\frac{1}{k^{1/2}}\left\Vert {\mathcal{r}}%
_{m,k}(\hat{\mbox{\boldmath${\theta}$}}_{m})-\mathbf{q}_{m+k}\right\Vert
=O_{P}(1),
\end{equation*}%
and 
\begin{equation*}
\max_{a\leq k\leq {\mathcal{n}}}\frac{1}{k^{1/2}}\left\Vert {\mathcal{r}}%
_{m,k}(\hat{\mbox{\boldmath${\theta}$}}_{m})-\mathbf{q}_{m+k}\right\Vert
=O_{P}(a^{-1/2+\zeta }),
\end{equation*}%
with some $\zeta <1/2$, when $E\log (\alpha _{0}\epsilon _{0}^{2}+\beta
_{0})>0$. The approximation in \eqref{hus} yields 
\begin{equation}
\max_{1\leq k\leq {\mathcal{n}}}\frac{1}{k^{1/2}}\left\Vert \mathbf{q}_{m+k}-%
\mbox{\boldmath${\Gamma}$}_{m}(k)\right\Vert =O_{P}(1),  \label{nostt}
\end{equation}%
and 
\begin{equation}
\max_{a\leq k\leq {\mathcal{n}}}\frac{1}{k^{1/2}}\left\Vert \mathbf{q}_{m+k}-%
\mbox{\boldmath${\Gamma}$}_{m}(k)\right\Vert =O_{P}(a^{-1/2+\zeta }).
\label{nosttt}
\end{equation}%
The approximations in \eqref{nostt} and \eqref{nosttt} imply part \textit{(i)%
} of the theorem, in the same way as \eqref{stac} and \eqref{stacc} yielded %
\eqref{nosttt} under condition $E\log (\alpha _{0}\epsilon _{0}^{2}+\beta
_{0})<0$.\newline
To prove part \textit{(ii)} of the theorem, we use again the approximations
in \eqref{stac} and \eqref{stacc}. By the scale transformation of Wiener
processes we have 
\begin{equation}
\sup_{r\leq t\leq {\mathcal{n}}}\frac{1}{t^{1/2}}\Vert \mathbf{W}(t)\Vert 
\overset{{\mathcal{D}}}{=}\sup_{1\leq t\leq {\mathcal{n}}/r}\frac{1}{t^{1/2}}%
\Vert \mathbf{W}(t)\Vert  \label{rende}
\end{equation}%
and therefore the Darling--Erd\H{o}s limit result (cf. \citealp{csorgo1997},
pp.\ 363--368) with \eqref{stac} and \eqref{stacc} yields 
\begin{equation*}
\frac{1}{2\log \log ({\mathcal{n}}/r)}\max_{r\leq k\leq {\mathcal{n}}}\frac{1%
}{k}{\mathcal{r}}_{m,k}^{\top }(\hat{\mbox{\boldmath${\theta}$}}_{m})\mathbf{%
D}^{-1}{\mathcal{r}}_{m,k}\overset{P}{\rightarrow }1
\end{equation*}%
and%
\begin{equation*}
\max_{r\leq k\leq rg}\frac{1}{k}{\mathcal{r}}_{m,k}^{\top }(\hat{%
\mbox{\boldmath${\theta}$}}_{m})\mathbf{D}^{-1}{\mathcal{r}}%
_{m,k}=O_{P}\left( \log \log \log ({\mathcal{n}}/r)\right)
\end{equation*}%
with 
\begin{equation*}
g=(\log ({\mathcal{n}}/r))^{\bar{\phi}},\;\;\;\bar{\phi}>4.
\end{equation*}%
Now we get 
\begin{equation*}
\lim_{m\rightarrow \infty }P\left\{ \max_{r\leq k\leq {\mathcal{n}}}\frac{1}{%
k}{\mathcal{r}}_{m,k}^{\top }(\hat{\mbox{\boldmath${\theta}$}}_{m})\mathbf{D}%
^{-1}{\mathcal{r}}_{m,k}=\max_{rg\leq k\leq {\mathcal{n}}}\frac{1}{k}{%
\mathcal{r}}_{m,k}^{\top }(\hat{\mbox{\boldmath${\theta}$}}_{m})\mathbf{D}%
^{-1}{\mathcal{r}}_{m,k}\right\} =1.
\end{equation*}%
Using the approximation of ${\mathcal{r}}_{m,k}(\hat{\mbox{\boldmath${%
\theta}$}}_{m})$ with Gaussian processes we get 
\begin{equation*}
\max_{rg\leq k\leq {\mathcal{n}}}\frac{1}{k}\left\vert {\mathcal{r}}%
_{m,k}^{\top }(\hat{\mbox{\boldmath${\theta}$}}_{m})\mathbf{D}^{-1}{\mathcal{%
r}}_{m,k}-\mbox{\boldmath${\Gamma}$}_{m}^{\top }(k)\mathbf{D}^{-1}%
\mbox{\boldmath${\Gamma}$}_{m}(k)\right\vert =O_{P}\left( (\log ({\mathcal{n}%
}/r))^{-\bar{\phi}}\right)
\end{equation*}%
with some $\bar{\phi}>0$. Now using the results on pp. 363--365 in %
\citet{csorgo1997} with \eqref{rende} yields 
\begin{equation*}
\lim_{m\rightarrow \infty }P\left\{ a(\log ({\mathcal{n}}/r))\max_{rg\leq
k\leq {\mathcal{n}}}\frac{1}{k^{1/2}}\left( \mbox{\boldmath${\Gamma}$}%
_{m}^{\top }(k)\mathbf{D}^{-1}\mbox{\boldmath${\Gamma}$}_{m}(k)\right)
^{1/2}\leq x+b_{2}(\log ({\mathcal{n}}/r))\right\} =\exp (-e^{-x}),
\end{equation*}%
completing the proof of part \textit{(ii)} of the theorem when $E\log
(\alpha _{0}\epsilon _{0}^{2}+\beta _{0})<0$. The same arguments can be used
when $E\log (\alpha _{0}\epsilon _{0}^{2}+\beta _{0})>0$.
\end{proof}

\begin{proof}[Proof of Theorem \protect\ref{thcons}]
First we assume that \eqref{stalt} holds. After $m+k^{\ast }$, the sequence
can turn into (asymptotically) stationary sequence. Since \eqref{stalt}
holds, we have 
\begin{equation*}
\max_{k^{\ast }+1\leq j<\infty }\sup_{\mbox{\boldmath${\theta}$}\in %
\mbox{\boldmath${\theta}$}}\frac{1}{j-k^{\ast }}\Vert {\mathcal{r}}_{m,j}(%
\mbox{\boldmath${\theta}$})-{\mathcal{r}}_{m,k^{\ast }}(\mbox{\boldmath${%
\theta}$})-(j-k^{\ast }){\mathcal{r}}^{(1)}(\mbox{\boldmath${\theta}$})\Vert
=o_{P}(1),
\end{equation*}%
where 
\begin{equation*}
{\mathcal{r}}^{(1)}(\mbox{\boldmath${\theta}$})=E\bar{{\mathcal{r}}}%
_{m,k^{\ast }+1}^{(1)}(\mbox{\boldmath${\theta}$}),
\end{equation*}%
\begin{equation*}
\bar{{\mathcal{r}}}_{m,k^{\ast }+j}^{(1)}(\mbox{\boldmath${\theta}$}%
)=\sum_{i=k^{\ast }+1}^{k^{\ast }+j}\left( \frac{\partial \hat{\ell}_{i}(%
\mbox{\boldmath${\theta}$})}{\partial \alpha },\frac{\partial \hat{\ell}_{i}(%
\mbox{\boldmath${\theta}$})}{\partial \beta }\right) ^{\top },
\end{equation*}%
Under $H_{A}$ we have that $\mbox{\boldmath${\Delta}$}={\mathcal{r}}^{(1)}(%
\mbox{\boldmath${\theta}$}_{0})\neq \mathbf{0}.$ Next we write for $%
k>k^{\ast }$ 
\begin{equation}
{\mathcal{r}}_{m,k}(\hat{\mbox{\boldmath${\theta}$}}_{m})=(k-k^{\ast })%
\mbox{\boldmath${\Delta}$}+\mathbf{m}_{m,k}  \label{mdec-1}
\end{equation}%
with 
\begin{equation}
\mathbf{m}_{m,k}={\mathcal{r}}_{m,k^{\ast }}(\hat{\mbox{\boldmath${\theta}$}}%
_{m})+[{\mathcal{r}}_{m,k}(\hat{\mbox{\boldmath${\theta}$}}_{m})-{\mathcal{r}%
}_{m,k^{\ast }}(\hat{\mbox{\boldmath${\theta}$}}_{m})-(k-k^{\ast })%
\mbox{\boldmath${\Delta}$}]  \label{mdec}
\end{equation}%
We use the decomposition 
\begin{equation}
{\mathcal{D}}_{m,k}=(k-k^{\ast })^{2}A_{m}+2(k-k^{\ast })\mbox{\boldmath${%
\Delta}$}^{\top }\hat{\mathbf{D}}_{m}^{-1}\mathbf{m}_{m,k}+\mathbf{m}%
_{m,k}^{\top }\mathbf{D}_{m}^{-1}\mathbf{m}_{m,k},  \label{mdec1}
\end{equation}%
where $A_{m}$ is defined in \eqref{deldef}. We note that, for any ${\mathcal{%
M}}$%
\begin{equation*}
\max_{1\leq k\leq {\mathcal{M}}}\frac{{\mathcal{D}}_{m,k}}{g_{m}(k)}=\max
\left( \max_{1\leq k\leq k^{\ast }}\frac{{\mathcal{D}}_{m,k}}{g_{m}(k)}%
,\max_{k^{\ast }+1\leq k\leq {\mathcal{M}}}\frac{{\mathcal{D}}_{m,k}}{%
g_{m}(k)}\right) .
\end{equation*}%
First we consider the case when \eqref{stalt} holds. The observations $y_{i}$
can be approximated with a stationary sequence defined in \eqref{hatx}.
Recalling the definitions of $t^{\ast },\bar{u},u_{{\mathcal{n}}}$ and $%
u^{\ast }$ in \eqref{tstar}--\eqref{ustar}, we define 
\begin{equation*}
{\mathcal{M}}=k^{\ast }+u_{{\mathcal{n}}}\left( {\mathcal{n}}^{1-\eta }\frac{%
{\mathcal{c}}}{A_{m}}\right) ^{1/(2-\eta )}+R,
\end{equation*}%
with 
\begin{equation*}
R=x{(u^{\ast }+t^{\ast })^{3/2}}\left( \left( {\mathcal{n}}^{1-\eta }\frac{{%
\mathcal{c}}}{A_{m}}\right) ^{1/(2-\tau )}\right) ^{1/2}.
\end{equation*}%
Now according to Theorem \ref{ma1} 
\begin{equation*}
\left( \max_{1\leq k\leq k^{\ast }}\frac{{\mathcal{D}}_{m,k}}{k^{\eta }}-{%
\mathcal{c}}{\mathcal{n}}^{1-\eta }\right) \left( {\mathcal{n}}^{(1-\eta
)/(2-\eta )}\right) ^{-3/2+\eta }\overset{\mathcal{P}}{\rightarrow }-\infty .
\end{equation*}%
It follows from the proof of Theorem \ref{ma1} 
\begin{equation*}
\max_{k^{\ast }+1\leq k\leq {\mathcal{M}}}\frac{\mathbf{m}_{m,k}^{\top }%
\mathbf{D}_{m}^{-1}\mathbf{m}_{m,k}}{g_{m}(k)}=O_{P}(1),
\end{equation*}%
and for any $\bar{\eta}<1/2$ 
\begin{equation*}
\max_{k^{\ast }+1\leq k\leq {\mathcal{M}}}\frac{1}{k-k^{\ast }}\frac{\Vert 
\mathbf{m}_{m,k}\Vert }{{\mathcal{n}}^{1/2}(k/{\mathcal{n}})^{\bar{\eta}}}%
=O_{P}(1).
\end{equation*}%
Let $0<\delta <1$ and write 
\begin{equation*}
{\mathcal{n}}^{1-\eta }\max_{k^{\ast }+1\leq k\leq (1-\delta ){\mathcal{M}}}%
\frac{{\mathcal{D}}_{m,k}}{g_{m}(k)}=(1-\delta )^{2-\eta }{\mathcal{M}}%
^{2-\eta }O_{P}(1),
\end{equation*}%
where $O_{P}(1)$ does not depend on $\delta $. Thus we conclude 
\begin{equation*}
\left( {\mathcal{n}}^{1-\eta }\max_{k^{\ast }+1\leq k\leq (1-\delta ){%
\mathcal{M}}}\frac{{\mathcal{D}}_{m,k}}{g_{m}(k)}-{\mathcal{n}}^{1-\eta
}\right) ({\mathcal{n}}^{1-\eta })^{(-3/2+\eta )/(2-\eta )}\overset{\mathcal{%
P}}{\rightarrow }-\infty .
\end{equation*}%
Our calculations show that we need to consider the asymptotic properties of 
\begin{align*}
& \hspace{-3cm}\left( \max_{(1-\delta ){\mathcal{M}}\leq k\leq {\mathcal{M}}}%
\frac{1}{k^{\eta }}\left( (k-k^{\ast })^{2}A_{m}+2(k-k^{\ast })%
\mbox{\boldmath${\Delta}$}^{\top }\hat{\mathbf{D}}_{m}^{-1}\mathbf{m}%
_{m,k}\right) \right. \\
& \left. -\left[ {\mathcal{c}}{\mathcal{n}}^{1-\eta }-\frac{1}{{\mathcal{M}}%
^{\eta }}({\mathcal{M}}-k^{\ast })^{2}A_{m}\right] \right) \left( {\mathcal{n%
}}^{(1-\eta )/(2-\eta )}\right) ^{-3/2+\eta }.
\end{align*}%
Since the model is the same until $y_{m+k^{\ast }}$, the $k^{\ast }$th
observation following the training sample, we have that 
\begin{align}
\mathbf{D}=\lim_{k^{\ast }\rightarrow \infty }& \frac{1}{m+k^{\ast }}%
E\left\{ \left( \frac{\partial \bar{{\mathcal{r}}}_{m,m+k^{\ast }}(%
\mbox{\boldmath${\theta}$}_{0})}{\partial \alpha },\frac{\partial \bar{{%
\mathcal{r}}}_{m,m+k^{\ast }}(\mbox{\boldmath${\theta}$}_{0})}{\partial
\beta }\right) ^{\top }\right.  \label{bsione} \\
& \left. \hspace{3cm}\times \left( \frac{\partial \bar{{\mathcal{r}}}%
_{m,m+k^{\ast }}(\mbox{\boldmath${\theta}$}_{0})}{\partial \alpha },\frac{%
\partial \bar{{\mathcal{r}}}_{m,m+k^{\ast }}(\mbox{\boldmath${\theta}$}_{0})%
}{\partial \beta }\right) \right\} .  \notag
\end{align}%
(Note that the formula for $\mathbf{D}$ depends on whether the training
sample is stationary or non stationary, cf. \citealp{berkes2003garch} and %
\citealp{jensen2004asymptotic}). Along the lines of the proof of Theorem \ref%
{ma1} one can show that 
\begin{equation}
{\mathcal{M}}^{-3/2}\left( \lfloor {\mathcal{M}}t\rfloor -\lfloor k^{\ast }/{%
\mathcal{M}}\rfloor \right) \mbox{\boldmath${\Delta}$}^{\top }\hat{\mathbf{D}%
}_{m}^{-1}\mathbf{m}_{m,{\mathcal{M}}t}\overset{{\mathcal{D}}[1-\delta ,1]}{%
\longrightarrow }(t-t^{\ast }/(u^{\ast }+t^{\ast }))Z(t),  \label{Zco}
\end{equation}%
where $\{Z(t),t\geq t^{\ast }/(t^{\ast }+u^{\ast })\}$ is a Gaussian process
with $EZ(t)=0$ and 
\begin{equation*}
EZ(t)Z(s)=t^{\ast }\mbox{\boldmath${\Delta}$}^{\top }\mathbf{D}^{-1}%
\mbox{\boldmath${\Delta}$}+[\min (t,s)-t^{\ast }/(t^{\ast }+u^{\ast })]%
\mbox{\boldmath${\Delta}$}^{\top }\mathbf{D}^{-1}\mbox{\boldmath${\Sigma}$}%
_{1}\mathbf{D}^{-1}\mbox{\boldmath${\Delta}$},
\end{equation*}%
where $\mbox{\boldmath${\Sigma}$}_{1}$ is defined in \eqref{Sig1}. Hence 
\begin{equation*}
{\mathcal{M}}^{-1/2+\eta }\max_{(1-\delta ){\mathcal{M}}\leq k\leq {\mathcal{%
M}}}\frac{1}{k^{\tau }}(k-k^{\ast })\mbox{\boldmath${\Delta}$}^{\top }\hat{%
\mathbf{D}}_{m}^{-1}\mathbf{m}_{m,k}\overset{{\mathcal{D}}}{\rightarrow }%
\sup_{1-\delta \leq t\leq 1}\frac{1}{t^{\eta }}(t-t^{\ast })Z(t).
\end{equation*}%
Also, for any $x>0$ 
\begin{align*}
\hspace{-1cm}\lim_{\delta \rightarrow 0}\limsup_{m\rightarrow \infty }&
P\left\{ {\mathcal{M}}^{-1/2+\eta }\max_{(1-\delta ){\mathcal{M}}\leq k\leq {%
\mathcal{M}}}\frac{1}{k^{\eta }}\left\vert (k-k^{\ast })\mbox{\boldmath${%
\Delta}$}^{\top }\hat{\mathbf{D}}_{m}^{-1}\mathbf{m}_{m,k}-k^{\ast }%
\mbox{\boldmath${\Delta}$}^{\top }\hat{\mathbf{D}}_{m}^{-1}\mathbf{m}_{m,{%
\mathcal{M}}}\right\vert >x\right\} \\
& =0.
\end{align*}%
Since $(k-k^{\ast })^{2}/t^{\eta }$ is increasing on $[(1-\delta ){\mathcal{M%
}},{\mathcal{M}}]$ we get that 
\begin{align*}
& \hspace{-3cm}\lim_{\delta \rightarrow 0}\limsup_{m\rightarrow \infty
}P\left\{ \sup_{(1-\delta ){\mathcal{M}}\leq k\leq {\mathcal{M}}}\frac{1}{%
k^{\eta }}\left( (k-k^{\ast })^{2}A_{m}+2(k-k^{\ast })\mbox{\boldmath${%
\Delta}$}^{\top }\hat{\mathbf{D}}_{m}^{-1}\mathbf{m}_{m,k}\right) \right. \\
& \left. =\frac{1}{{\mathcal{M}}^{\eta }}\left[ ({\mathcal{M}}-k^{\ast
})^{2}A_{m}+2({\mathcal{M}}-k^{\ast })\mbox{\boldmath${\Delta}$}^{\top }\hat{%
\mathbf{D}}_{m}^{-1}\mathbf{m}_{m,{\mathcal{M}}}\right] \right\} =1.
\end{align*}%
Next we write 
\begin{align*}
P& \left\{ \frac{1}{{\mathcal{M}}^{\eta }}\left[ ({\mathcal{M}}-k^{\ast
})^{2}A_{m}+2({\mathcal{M}}-k^{\ast })\mbox{\boldmath${\Delta}$}^{\top }\hat{%
\mathbf{D}}_{m}^{-1}\mathbf{m}_{m,{\mathcal{M}}}\right] \leq {\mathcal{c}}{%
\mathcal{n}}^{1-\eta }\right\} \\
& =P\left\{ 2({\mathcal{M}}-k^{\ast })[\mbox{\boldmath${\Delta}$}^{\top }%
\hat{\mathbf{D}}_{m}^{-1}\mathbf{m}_{m,{\mathcal{M}}}]{\mathcal{M}}%
^{-3/2}\leq \left[ {\mathcal{c}}{\mathcal{n}}^{1-\eta }-\frac{1}{{\mathcal{M}%
}^{\eta }}({\mathcal{M}}-k^{\ast })^{2}A_{m}\right] {\mathcal{M}}^{-3/2+\eta
}\right\} .
\end{align*}%
Using \eqref{Zco} we get 
\begin{equation*}
{\mathcal{M}}^{-1/2}\mbox{\boldmath${\Delta}$}^{\top }\hat{\mathbf{D}}%
_{m}^{-1}\mathbf{m}_{m,{\mathcal{M}}}\overset{{\mathcal{D}}}{\rightarrow }%
\bar{s}_{1}{\mathcal{N}},
\end{equation*}%
where ${\mathcal{N}}$ denotes a standard normal random variable, $\bar{s}%
_{1} $ is defined in \eqref{s1} and $\mathbf{D}$ is the covariance matrix in %
\eqref{bsione}. Also, 
\begin{equation*}
\left[ {\mathcal{n}}^{1-\eta }-\frac{1}{{\mathcal{M}}^{\eta }}({\mathcal{M}}%
-k^{\ast })^{2}A_{m}\right] {\mathcal{M}}^{-3/2+\eta }\rightarrow -2x,
\end{equation*}%
and thus we conclude 
\begin{align*}
\lim_{m\rightarrow \infty }& P\left\{ 2({\mathcal{M}}-k^{\ast })[%
\mbox{\boldmath${\Delta}$}^{\top }\hat{\mathbf{D}}_{m}^{-1}\mathbf{m}_{m,{%
\mathcal{M}}}]{\mathcal{M}}^{-3/2}\leq \left[ {\mathcal{c}}{\mathcal{n}}%
^{1-\eta }-\frac{1}{{\mathcal{M}}^{\eta }}({\mathcal{M}}-k^{\ast })^{2}A_{m}%
\right] {\mathcal{M}}^{-3/2+\eta }\right\} \\
& =P\left\{ \bar{s}_{1}{\mathcal{N}}\leq -x\right\} =1-\Phi (x/\bar{s}_{1}).
\end{align*}%
We assume now that $t^{\ast }=\infty $ and $\bar{u}=0$ also holds. We modify
the definition of ${\mathcal{M}}$ as 
\begin{equation*}
{\mathcal{M}}=k^{\ast }+\left( \frac{{\mathcal{c}}}{A_{m}}{\mathcal{n}}%
^{1-\eta }(k^{\ast })^{\eta }\right) ^{1/2}+R\quad \;\mbox{with}\quad \;R=%
\frac{x}{A_{m}}(k^{\ast })^{1/2}.
\end{equation*}%
Proceeding as in the previous case we get that for any $0<\delta <1$ 
\begin{align}
\lim_{m\rightarrow \infty }& P\left\{ \tau _{m}>{\mathcal{M}}\right\}
\label{mm3} \\
& =\lim_{m\rightarrow \infty }\left\{ \max_{(1-\delta ){\mathcal{M}}\leq
k\leq {\mathcal{M}}}\frac{1}{k^{\eta }}\left( (k-k^{\ast
})^{2}A_{m}+2(k-k^{\ast })\mbox{\boldmath${\Delta}$}^{\top }\hat{\mathbf{D}}%
_{m}^{-1}\mathbf{m}_{m.k}\right) \leq {\mathcal{c}}{\mathcal{n}}^{1-\eta
}\right\}  \notag \\
& =\lim_{m\rightarrow \infty }P\left\{ 2({\mathcal{M}}-k^{\ast })%
\mbox{\boldmath${\Delta}$}^{\top }\hat{\mathbf{D}}_{m}^{-1}\mathbf{m}_{m,{%
\mathcal{M}}}\leq \left( {\mathcal{c}}{\mathcal{n}}^{1-\eta }-\frac{1}{{%
\mathcal{M}}^{\eta }}({\mathcal{M}}-k^{\ast })^{2}A_{m}\right) {\mathcal{M}}%
^{\eta }\right\}  \notag \\
& =\lim_{m\rightarrow \infty }P\left\{ 2(k^{\ast })^{-1/2}%
\mbox{\boldmath${\Delta}$}^{\top }\hat{\mathbf{D}}_{m}^{-1}\mathbf{m}_{m,{%
\mathcal{M}}}\leq \left( {\mathcal{c}}{\mathcal{n}}^{1-\eta }-\frac{1}{{%
\mathcal{M}}^{\tau }}({\mathcal{M}}-k^{\ast })^{2}A_{m}\right) \frac{{%
\mathcal{M}}^{\eta }}{({\mathcal{M}}-k^{\ast })(k^{\ast })^{1/2}}\right\} . 
\notag
\end{align}%
Using the definitions of ${\mathcal{M}}$ and $R$ we have 
\begin{equation*}
\left( {\mathcal{c}}{\mathcal{n}}^{1-\eta }-\frac{1}{{\mathcal{M}}^{\eta }}({%
\mathcal{M}}-k^{\ast })^{2}A_{m}\right) \frac{{\mathcal{M}}^{\eta }}{({%
\mathcal{M}}-k^{\ast })(k^{\ast })^{1/2}}\rightarrow -2x.
\end{equation*}%
As in the proof of \eqref{Zco} we have 
\begin{equation*}
(k^{\ast })^{-1/2}\mbox{\boldmath${\Delta}$}^{\top }\hat{\mathbf{D}}_{m}^{-1}%
\mathbf{m}_{m,{\mathcal{M}}}\overset{{\mathcal{D}}}{\rightarrow }\bar{s}_{2}{%
\mathcal{N}},
\end{equation*}%
where ${\mathcal{N}}$ is a standard normal random variable and $\bar{s}_{2}$
is defined in \eqref{s2}. Now we conclude 
\begin{equation*}
\lim_{m\rightarrow \infty }P\{\tau _{m}>{\mathcal{M}}\}=1-\Phi (x/\bar{s}%
_{2}),
\end{equation*}%
where $\Phi $ is the standard normal distribution function. \newline
In the final case we write 
\begin{equation*}
k^{\ast }=u_{m}{\mathcal{n}}\;\;\;\;\mbox{with}\;\;\;u_{m}\rightarrow \bar{u}%
\in (0,1).
\end{equation*}%
In this case, we define 
\begin{equation*}
{\mathcal{M}}=k^{\ast }+x(k^{\ast })^{1/2}.
\end{equation*}%
We write 
\begin{equation*}
\max_{1\leq k\leq {\mathcal{M}}}\frac{{\mathcal{D}}_{m,k}}{{\mathcal{n}}(k/{%
\mathcal{n}})^{\eta }}=\max \left( U_{1,m},U_{2,m}\right) ,
\end{equation*}%
\begin{equation*}
U_{1,m}=\max_{1\leq k\leq k^{\ast }}\frac{{\mathcal{D}}_{m,k}}{{\mathcal{n}}%
(k/{\mathcal{n}})^{\eta }}
\end{equation*}%
and 
\begin{equation*}
U_{2,m}=\max_{k^{\ast }+1\leq k\leq {\mathcal{M}}}\frac{{\mathcal{D}}_{m,k}}{%
{\mathcal{n}}(k/{\mathcal{n}})^{\eta }}.
\end{equation*}%
We use the decomposition in \eqref{mdec} and we note 
\begin{equation*}
\left( \frac{1}{k^{\ast }}\right) ^{1/2}\max_{k^{\ast }+1\leq k\leq {%
\mathcal{M}}}\Vert {\mathcal{r}}_{m,k}(\hat{\mbox{\boldmath${\theta}$}}_{m})-%
{\mathcal{r}}_{m,k^{\ast }}(\hat{\mbox{\boldmath${\theta}$}}_{m})\Vert
=O_{P}(1).
\end{equation*}%
It follows from the proof of Theorem \ref{ma1} that 
\begin{equation*}
\left( \max_{1\leq k\leq k^{\ast }}\frac{{\mathcal{D}}_{m,k}}{k^{\ast
}(k/k^{\ast })^{\eta }},\left( \frac{1}{k^{\ast }}\right) ^{1/2}{\mathcal{r}}%
_{m,k^{\ast }}(\hat{\mbox{\boldmath${\theta}$}}_{m})\right) \overset{{%
\mathcal{D}}}{\rightarrow }\left( \sup_{0\leq s\leq 1}\frac{%
\mbox{\boldmath${\Gamma}$}(s)^{\top }\mathbf{D}^{-1}\mbox{\boldmath${%
\Gamma}$}(s)}{s^{\eta }},\mbox{\boldmath${\Gamma}$}(1)\right) ,
\end{equation*}%
where $\{\mbox{\boldmath${\Gamma}$}(s),s\geq 0\}$ is a Gaussian process with 
$E\mbox{\boldmath${\Gamma}$}(s)=\mathbf{0}$ and $E\mbox{\boldmath${\Gamma}$}%
(s)\mbox{\boldmath${\Gamma}$}^{\top }(t)=\min (t,s)\mathbf{D}.$ We write 
\begin{equation*}
U_{2,m}=\max_{1\leq j\leq x(k^{\ast })^{1/2}}\left( \frac{k^{\ast }+j}{m}%
\right) ^{1-\eta }\left( {\mathcal{r}}_{m,k^{\ast }}+j\mbox{\boldmath${%
\Delta}$}+[{\mathcal{r}}_{m,k^{\ast }}-{\mathcal{r}}_{m,k}]\right) ^{\top }%
\hat{\mathbf{D}}_{m}^{-1}\left( {\mathcal{r}}_{m,k^{\ast }}+j%
\mbox{\boldmath${\Delta}$}+[{\mathcal{r}}_{m,k^{\ast }}-{\mathcal{r}}%
_{m,k}]\right) ^{\top }.
\end{equation*}%
Thus we get 
\begin{equation*}
\left( U_{m,1},U_{m,2}\right) \overset{{\mathcal{D}}}{\rightarrow }\left(
\sup_{0<t\leq 1}\frac{\bar{u}^{1-\eta }}{t^{\eta }}\mbox{\boldmath${\Gamma}$}%
^{\top }\mathbf{D}^{-1}\mbox{\boldmath${\Gamma}$}(t),\bar{u}^{1-\eta
}\sup_{0\leq t\leq x}(t\mbox{\boldmath${\Delta}$}+\mbox{\boldmath${\Gamma}$}%
(1))^{\top }\mathbf{D}^{-1}(t\mbox{\boldmath${\Delta}$}+\mbox{\boldmath${%
\Gamma}$}(1))\right) .
\end{equation*}

Next we assume that \eqref{nostalst} holds. We replace \eqref{mdec} and %
\eqref{mdec1} with%
\begin{equation*}
{\mathcal{r}}_{m,k}(\hat{\mbox{\boldmath${\theta}$}}_{m})=(k-k^{\ast })%
\mbox{\boldmath${\Upsilon}$}+\tilde{\mathbf{m}}_{m,k},
\end{equation*}%
\begin{equation*}
\tilde{\mathbf{m}}_{m,k}={\mathcal{r}}_{m,k^{\ast }}(\hat{%
\mbox{\boldmath${\theta}$}}_{m})+[{\mathcal{r}}_{m,k}(\hat{%
\mbox{\boldmath${\theta}$}}_{m})-{\mathcal{r}}_{m,k^{\ast }}(\hat{%
\mbox{\boldmath${\theta}$}}_{m})-(k-k^{\ast })\mbox{\boldmath${\Upsilon}$}],
\end{equation*}%
\begin{equation*}
{\mathcal{D}}_{m,k}=(k-k^{\ast })^{2}B_{m}+2(k-k^{\ast })\mbox{\boldmath${%
\Upsilon}$}^{\top }\hat{\mathbf{D}}_{m}^{-1}\tilde{\mathbf{m}}_{m,k}+\tilde{%
\mathbf{m}}_{m,k}^{\top }\hat{\mathbf{D}}_{m}^{-1}\tilde{\mathbf{m}}_{m,k}.
\end{equation*}%
Following the same arguments as in the case of \eqref{stalt}, we get that if 
$\bar{u}=0$, then 
\begin{align*}
& \lim_{m\rightarrow \infty }P\{\tau _{m}>\tilde{{\mathcal{M}}}\} \\
& =\lim_{m\rightarrow \infty }P\left\{ 2(\tilde{{\mathcal{M}}}-k^{\ast })[%
\mbox{\boldmath${\Upsilon}$}^{\top }\hat{\mathbf{D}}_{m}^{-1}\tilde{\mathbf{m%
}}_{m,\tilde{{\mathcal{M}}}}]\tilde{{\mathcal{M}}}^{-3/2}\leq \left[ {%
\mathcal{c}}{\mathcal{n}}^{1-\eta }-\frac{1}{\tilde{{\mathcal{M}}}^{\eta }}%
\left( \tilde{{\mathcal{M}}}-k^{\ast }\right) ^{2}B_{m}\right] \tilde{{%
\mathcal{M}}}^{-3/2+\eta }\right\} ,
\end{align*}%
and the definition of $\tilde{{\mathcal{M}}}$ depends on if $\tilde{t}^{\ast
}<\infty $ is finite or not. If $\tilde{t^{\ast }}<\infty $ we use 
\begin{equation*}
\tilde{{\mathcal{M}}}=k^{\ast }+\tilde{u}_{{\mathcal{n}}}\left( {\mathcal{n}}%
^{1-\eta }\frac{{\mathcal{c}}}{B_{m}}\right) ^{1/(2-\eta )}+\tilde{R},
\end{equation*}%
with 
\begin{equation*}
\tilde{R}=x{(\tilde{u}^{\ast }+\tilde{t}^{\ast })^{3/2}}\left( \left( {%
\mathcal{n}}^{1-\eta }\frac{{\mathcal{c}}}{B_{m}}\right) ^{1/(2-\eta
)}\right) ^{1/2}.
\end{equation*}%
Using the definition of $\bar{{\mathcal{M}}}$ we get 
\begin{equation}
\lim_{m\rightarrow \infty }\left[ {\mathcal{c}}{\mathcal{n}}^{1-\tau }-\frac{%
1}{\tilde{{\mathcal{M}}}^{\tau }}\left( \tilde{{\mathcal{M}}}-k^{\ast
}\right) ^{2}B_{m}\right] \tilde{{\mathcal{M}}}^{-3/2+\tau }=-2x.
\label{mconv}
\end{equation}%
Applying \eqref{hus} and Lemma \ref{nosalap} with the independence of the
approximating Gaussian processes we conclude 
\begin{equation}
(\tilde{{\mathcal{M}}}-k^{\ast })[\mbox{\boldmath${\Upsilon}$}^{\top }\hat{%
\mathbf{D}}_{m}^{-1}\tilde{\mathbf{m}}_{m,\tilde{{\mathcal{M}}}}]\tilde{{%
\mathcal{M}}}^{-3/2}\overset{{\mathcal{D}}}{\rightarrow }\bar{s}_{1}{%
\mathcal{N}},  \label{mmco}
\end{equation}%
where ${\mathcal{N}}$ is a standard normal random variable and $\bar{s}_{1}$
is defined in \eqref{s1}. The covariance matrix $\mathbf{D}$ is defined in %
\eqref{bsione}. If $\tilde{t}^{\ast }=\infty $ and $\bar{u}=0$, then we use 
\begin{equation*}
{\mathcal{M}}=k^{\ast }+\left( \frac{{\mathcal{c}}}{B_{m}}{\mathcal{n}}%
^{1-\eta }(k^{\ast })^{\eta }\right) ^{1/2}+\tilde{R},\quad \;\mbox{with}%
\quad \;\tilde{R}=\frac{x}{B_{m}}(k^{\ast })^{1/2}.
\end{equation*}%
We still have \eqref{mconv} but \eqref{mmco} is replaced with 
\begin{equation*}
(\tilde{{\mathcal{M}}}-k^{\ast })[\mbox{\boldmath${\Upsilon}$}^{\top }\hat{%
\mathbf{D}}_{m}^{-1}\tilde{\mathbf{m}}_{m,\tilde{{\mathcal{M}}}}]\tilde{{%
\mathcal{M}}}^{-3/2}\overset{{\mathcal{D}}}{\rightarrow }\bar{s}_{2}{%
\mathcal{N}},
\end{equation*}%
where $\bar{s}_{2}$ is given in \eqref{s2}. As in the previous case, if $0<%
\bar{u}<1$, then we use again ${\mathcal{M}}=k^{\ast }+x(k^{\ast })^{1/2}$.
Since the sum of gradients (after removing the mean) of the log likelihood
function of the observations after the change is negligible, we get that 
\begin{equation*}
\max_{1\leq k\leq {\mathcal{M}}}\frac{{\mathcal{D}}_{m}(k)}{g_{m}(k)}\overset%
{{\mathcal{D}}}{\rightarrow }\bar{u}^{1-\eta }\max \left( \sup_{0<t\leq 1}%
\frac{1}{t^{\eta }}\mbox{\boldmath${\Gamma}$}^{\top }(t)\mathbf{D}^{-1}%
\mbox{\boldmath${\Gamma}$}(t),\sup_{0\leq t\leq x}(t\mbox{\boldmath${%
\Upsilon}$}+\mbox{\boldmath${\Gamma}$}(1)^{\top }\mathbf{D}^{-1}(t%
\mbox{\boldmath${\Upsilon}$}+\mbox{\boldmath${\Gamma}$}(t))\right) ,
\end{equation*}%
where recall that $\{\mbox{\boldmath${\Gamma}$}(t),t\geq 0\}$ is a Gaussian
process $E\mbox{\boldmath${\Gamma}$}(t)=\mathbf{0}$ and $E%
\mbox{\boldmath${\Gamma}$}(t)\mbox{\boldmath${\Gamma}$}^{\top }(s)=\min (t,s)%
\mathbf{D}$.

Next we consider the R\'{e}nyi type detector. First we investigate the case
when after the change we have an (asymptotically) stationary sequence.
Assume that $k^{\ast }\leq r$, i.e.\ the change occurred before the
monitoring started. By definition, 
\begin{equation*}
P\{\bar{\tau}_{m}=r\}=P\left\{ \frac{{\mathcal{D}}_{m}(r)}{\bar{g}_{m}(r)}%
\geq 1\right\} .
\end{equation*}%
The proof is based again on \eqref{mdec-1} and \eqref{mdec}. It follows from
Lemma \ref{lealt1} that 
\begin{equation*}
\frac{1}{(k^{\ast })^{1/2}}\Vert {\mathcal{r}}_{m,k^{\ast }}\Vert =O_{P}(1)
\end{equation*}%
and if $k^{\ast }\rightarrow \infty $ and $r-k^{\ast }\rightarrow \infty $,
then%
\begin{equation*}
\left( \frac{1}{(k^{\ast })^{1/2}}{\mathcal{r}}_{m,k^{\ast }},\frac{1}{%
(r-k^{\ast })^{1/2}}\left[ {\mathcal{r}}_{m,r}(\hat{\mbox{\boldmath${%
\theta}$}}_{m})-{\mathcal{r}}_{m,k^{\ast }}(\hat{\mbox{\boldmath${\theta}$}}%
_{m})-(r-k^{\ast })\mbox{\boldmath${\Delta}$}\right] \right) \overset{{%
\mathcal{D}}}{\rightarrow }\left( \mathbf{N}_{1},\mathbf{N}_{2}\right)
\end{equation*}%
where $\mathbf{N}_{1},\mathbf{N}_{2}$ are independent standard normal random
variables with covariance matrices $\mbox{\boldmath${\Sigma}$}_{1}$ and $%
\mbox{\boldmath${\Sigma}$}_{2}$ defined in \eqref{Sig1} and \eqref{Sig2},
respectively. Thus we get 
\begin{equation*}
\frac{1}{r^{1/2}}\left( (r-k^{\ast })\mbox{\boldmath${\Delta}$}+\mathbf{m}%
_{m,r}\right) \overset{{\mathcal{D}}}{\rightarrow }{\mathcal{a}}^{1/2}%
\mathbf{N}_{1}+{\mathcal{b}}+\mathbf{N}_{2}.
\end{equation*}%
If ${\mathcal{a}}<\infty $, then we use 
\begin{equation*}
{\mathcal{M}}=k^{\ast }+xr^{1/2}.
\end{equation*}%
It follows from the proof of Theorem \ref{ma1} that for all $K>0$ 
\begin{align*}
& \left\{ \frac{1}{r^{1/2}}{\mathcal{r}}_{m,rt}(\hat{\mbox{\boldmath${%
\theta}$}}_{m}),\frac{1}{r^{1/2}}\left( {\mathcal{r}}_{m,k^{\ast }}+[{%
\mathcal{r}}_{m,xr^{1/2}}-{\mathcal{r}}_{m,k^{\ast }}-xr^{1/2}]\right)
,t\geq 0,x\geq 0\right\} \\
& \hspace{2cm}\overset{{\mathcal{D}}[0,{\mathcal{a}}]\times \lbrack 0,K]}{%
\longrightarrow }\left\{ \mbox{\boldmath${\Gamma}$}_{1}(t),%
\mbox{\boldmath${\Gamma}$}_{1}({\mathcal{a}})+x\mbox{\boldmath${\Delta}$}%
\right\}
\end{align*}%
where $\{\mbox{\boldmath${\Gamma}$}_{1}(t),t\geq 0\}$ is Gaussian process
with $E\mbox{\boldmath${\Gamma}$}_{1}(t)=\mathbf{0}$, $E\mbox{\boldmath${%
\Gamma}$}_{1}(t)\mbox{\boldmath${\Gamma}$}_{1}^{\top }(s)=\min (t,s)%
\mbox{\boldmath${\Sigma}$}_{1}$ with $\mbox{\boldmath${\Sigma}$}_{1}$
defined in \eqref{Sig1}. Thus we conclude 
\begin{align*}
\sup_{r\leq k\leq {\mathcal{M}}}& \frac{{\mathcal{D}}_{m}(k)}{g_{m}(k)}=%
\frac{1}{{\mathcal{c}}}\max \left( \max_{r\leq k\leq k^{\ast }}\frac{{%
\mathcal{D}}_{m}(k)}{(k/r)^{\eta }r},\;\max_{k^{\ast }<k\leq {\mathcal{M}}}%
\frac{{\mathcal{D}}_{m}(k)}{(k/r)^{\eta }r}\right) \\
& \overset{{\mathcal{D}}}{\rightarrow }\frac{1}{{\mathcal{c}}{\mathcal{a}}%
^{\eta }}\max \left( \sup_{1\leq t\leq {\mathcal{a}}}\mbox{\boldmath${%
\Gamma}$}_{1}^{\top }(t)\mathbf{D}^{-1}\mbox{\boldmath${\Gamma}$}%
_{1}(t),\;\max_{0\leq s\leq x}(\mbox{\boldmath${\Gamma}$}_{1}({\mathcal{a}}%
)+s\mbox{\boldmath${\Delta}$})^{\top }\mathbf{D}^{-1}(\mbox{\boldmath${%
\Gamma}$}_{1}({\mathcal{a}})+s\mbox{\boldmath${\Delta}$}).\right) .
\end{align*}%
For the final case ${\mathcal{a}}=\infty $ we use 
\begin{equation*}
{\mathcal{M}}=k^{\ast }+\left( \frac{{\mathcal{c}}}{A_{m}}\frac{(k^{\ast
})^{\eta }}{r^{\eta -1}}\right) ^{1/2}+x(k^{\ast })^{1/2}.
\end{equation*}%
We observe that 
\begin{equation}
\frac{k^{\ast }}{((k^{\ast })^{\eta }/r^{\eta -1})^{1/2}}=\left( \left( 
\frac{k^{\ast }}{r}\right) ^{2-\eta }r\right) ^{1/2}\rightarrow \infty
\label{la1}
\end{equation}%
and 
\begin{equation}
\frac{(k^{\ast })^{\eta }/r^{\eta -1}}{k^{\ast }}=\left( \frac{k^{\ast }}{r}%
\right) ^{\eta -1}\rightarrow \infty .  \label{la2}
\end{equation}%
As in the previous cases, 
\begin{align*}
\lim_{m\rightarrow \infty }& P\{\bar{\tau}_{m}>{\mathcal{M}}\} \\
& =\lim_{m\rightarrow \infty }P\left\{ 2({\mathcal{M}}-k^{\ast })%
\mbox{\boldmath${\Delta}$}^{\top }\hat{\mathbf{D}}_{m}^{-1}\mathbf{m}_{m,{%
\mathcal{M}}}<\left[ {\mathcal{c}}r^{1-\eta }-\frac{({\mathcal{M}}-k^{\ast
})^{2}A_{m}}{{\mathcal{M}}^{\eta }}\right] {\mathcal{M}}^{\eta }\right\} \\
& =\lim_{m\rightarrow \infty }P\left\{ 2{\mathcal{M}}^{-1/2}%
\mbox{\boldmath${\Delta}$}^{\top }\hat{\mathbf{D}}_{m}^{-1}\mathbf{m}_{m,{%
\mathcal{M}}}<\left[ {\mathcal{c}}r^{1-\eta }-\frac{({\mathcal{M}}-k^{\ast
})^{2}A_{m}}{{\mathcal{M}}^{\eta }}\right] \frac{1}{{\mathcal{M}}-k^{\ast }}{%
\mathcal{M}}^{-1/2+\eta }\right\} .
\end{align*}%
We obtain from \eqref{la1} and \eqref{la2} that 
\begin{equation*}
\left[ {\mathcal{c}}r^{1-\eta }-\frac{({\mathcal{M}}-k^{\ast })^{2}A_{m}}{{%
\mathcal{M}}^{\eta }}\right] \frac{1}{{\mathcal{M}}-k^{\ast }}{\mathcal{M}}%
^{-1/2+\eta }\rightarrow -2x.
\end{equation*}%
We get from \eqref{Zco} that 
\begin{equation*}
{\mathcal{M}}^{-1/2}\mbox{\boldmath${\Delta}$}^{\top }\hat{\mathbf{D}}%
_{m}^{-1}\mathbf{m}_{m,{\mathcal{M}}}\overset{{\mathcal{D}}}{\rightarrow }%
\bar{s}_{2}{\mathcal{N}},
\end{equation*}%
where ${\mathcal{N}}$ is a standard normal random variable and $\bar{s}_{2}$
is defined in \eqref{s2}. We now conclude 
\begin{equation*}
\lim_{m\rightarrow \infty }P\{\bar{\tau}_{m}>{\mathcal{M}}\}=1-\Phi (x/\bar{s%
}_{2}).
\end{equation*}%
The proof when \eqref{nostalst} holds is essentially the same so the details
are omitted.
\end{proof}

\end{document}